\documentclass[floatfix, noshowpacs, nofootinbib, aps, prx, 10pt, notitlepage,twopage]{revtex4-1}

\usepackage{amsmath,amsthm,amssymb,amsfonts,setspace}
\usepackage{bm}
\usepackage{float}
\usepackage{color}
\usepackage[usenames,dvipsnames]{xcolor}
\usepackage[colorlinks]{hyperref}
\usepackage{url}
\usepackage{cleveref} 

\usepackage[normalem]{ulem} 

\usepackage{graphicx}

\crefformat{equation}{Eq.~(#2#1#3)} 
\crefformat{section}{Section~#2#1#3}
\crefformat{theorem}{Theorem~#2#1#3}
\crefformat{lemma}{Lemma~#2#1#3}
\crefformat{appendix}{Appendix~#2#1#3}
\Crefformat{equation}{Equation~(#2#1#3)}
\crefformat{figure}{Fig.~#2#1#3}
\crefformat{corollary}{Corollary~#2#1#3}
\crefrangeformat{equation}{Eqs.~#3(#1)#4--#5(#2)#6}

\newtheorem{innercustomgeneric}{\customgenericname}
\providecommand{\customgenericname}{}
\newcommand{\newcustomtheorem}[2]{%
  \newenvironment{#1}[1]
  {%
   \renewcommand\customgenericname{#2}%
   \renewcommand\theinnercustomgeneric{##1}%
   \innercustomgeneric
  }
  {\endinnercustomgeneric}
}

\newcustomtheorem{customthm}{Theorem}
\newcustomtheorem{customlemma}{Lemma}

\newcommand{\bra}[1]{\left\langle{#1}\right\vert} 
\newcommand{\ket}[1]{\left\vert{#1}\right\rangle} 
\newcommand{\ketbra}[2]{\left\vert{#1}\right\rangle\!\left\langle{#2}\right\vert}

\renewcommand{\L}{\left(} 
\newcommand{\R}{\right)} 

\newcommand{\dg}{^\dagger}

\newcommand{\tr}{\mathrm{tr}}

\newcommand{\Id}{\openone}
\newcommand{\ii}{\mathrm{i}} 

\renewcommand{\H}{\mathcal{H}} 
\newcommand{\Hsym}{\mathcal{S}} 
\newcommand{\LO}[1]{\mathrm{LO}\left(#1\right)} 

\newcommand{\n}{\mathbf{n}} 
\renewcommand{\t}{\boldsymbol{\tau}} 
\newcommand{\bi}{\mathbf{i}} 

\newcommand{\OO}{\mathbf{0}}

\newcommand{\Sep}{\mathrm{Sep}} 
\newcommand{\DD}{\mathcal{D}} 

\newcommand{\Sym}{\textrm{Sym}} 
\newcommand{\SYM}[1]{\mathcal{S}_{#1}} 

\newcommand{\symm}[2]{\mathcal{S}_{#1,#2}} 
\newcommand{\targ}[1]{\rho_{#1,n}} 
\newcommand{\true}{\Psi_n} 

\newcommand{\dsep}[1]{\Delta_{#1,n}} 
\newcommand{\dsepB}[1]{\Delta\left( #1,n \right )} 

\renewcommand{\o}{\textrm{o}} 
\renewcommand{\O}{\textrm{O}} 

\newcommand{\D}{\mathcal{D}} 
\renewcommand{\S}{\mathcal{S}} 
\newcommand{\N}{\mathcal{N}} 
\newcommand{\NN}{\mathrm{N}} 
\newcommand{\C}{\mathbb{C}}  
\newcommand{\U}{\mathrm{U}} 
\renewcommand{\Pr}{\mathrm{Pr}_\tran}

\newcommand{\IO}{\textsc{IO}} 
\newcommand{\step}{\textsc{step}~}

\newcommand{\dtr}{\mathrm{d}_{\mathrm{tr}}}

\newcommand{\tran}{\eta} 
\newcommand{\tranef}{\eta_\mathrm{eff}} 

\newcommand{\eqdef}{\mathrel{:=}}
\newcommand{\proj}[2]{\mathbb{P}_{\textrm{sym}}^{#1,#2}}

\theoremstyle{plain}
\newtheorem{theorem}{Theorem}
\theoremstyle{plain}
\newtheorem{corollary}{Corollary} 
\theoremstyle{plain}
\newtheorem{lemma}{Lemma}
\theoremstyle{plain}

\theoremstyle{plain}
\newtheorem{remark}{Remark}
\theoremstyle{plain}
\newtheorem{problem}{Problem}

\global\long\global\long\global\long\def\SET#1#2{\mbox{\ensuremath{\ensuremath{\left\lbrace\left. #1\ \right|\ #2\right\rbrace }}}} 
\global\long\global\long\global\long\def\kb#1#2{\mbox{\ensuremath{\ensuremath{\ensuremath{|#1\rangle\!\langle#2|}}}}} 

\begin{document}
\title{Classical simulation of photonic linear optics with lost particles}

\author{Micha\l\ Oszmaniec}
\email{michal.oszmaniec@gmail.com}
\affiliation{ 
Institute of Theoretical Physics and Astrophysics, National Quantum Information Centre, Faculty of Mathematics, Physics and
Informatics, University of Gdansk, Wita Stwosza 57, 80-308 Gdańsk, Poland}

\author{Daniel J.\ Brod}
\email{brod@if.uff.br}
\affiliation{Instituto de F\'{\i}sica, Universidade Federal Fluminense, Av. Gal. Milton Tavares de Souza s/n, Niter\'oi, RJ, 24210-340, Brazil}

\date{\today}

\begin{abstract}
We explore the possibility of efficient classical simulation of linear optics experiments under the effect of particle losses. Specifically, we investigate the canonical boson sampling scenario in which an $n$-particle Fock input state propagates through a linear-optical network and is subsequently measured by particle-number detectors in the $m$ output modes. We examine two models of losses. In the first model a fixed number of particles is lost. We prove that in this scenario the output statistics can be well approximated by an efficient classical simulation, provided that the number of photons that is left grows slower than $\sqrt{n}$. In the second loss model, every time a photon passes through a beamsplitter in the network, it has some probability of being lost. For this model the relevant parameter is $s$, the smallest number of beamsplitters that any photon traverses as it propagates through the network. We prove that it is possible to approximately simulate the output statistics already if $s$ grows logarithmically with $m$, regardless of the geometry of the network. The latter result is obtained by proving that it is always possible to commute $s$ layers of uniform losses to the input of the network regardless of its geometry, which could be a result of independent interest. We believe that our findings put strong limitations on future experimental realizations of quantum computational supremacy proposals based on boson sampling.
\end{abstract}
\maketitle

\section{Introduction} \label{sec:intro}

Quantum computers are expected to offer an advantage in a wide variety of computational problems relative to their classical counterparts. However, experimental limitations mean quantum computers are notoriously hard to build, and to date it remains unclear which technological architecture is the most promising for this purpose. Thus, a full-purpose, large-scale universal quantum computer still seems like a long-term goal for the field.

With this in mind, a recent intermediate milestone was proposed in the form of the quantum computational supremacy paradigm \cite{Lund2017,Harrow2017}. This consists of a series of proposed restricted models of quantum computing which are not expected to be universal, but for which there are strong arguments that they outperform classical computers in some computational task. One example of such proposals is boson sampling \cite{Aaronson2013a}, where the quantum device is restricted to single-photon states, linear optics and photon-number detectors, which we focus on in this work. Other examples include circuits of commuting gates \cite{Bremner2011}, a variant of the one-clean-qubit model \cite{Knill1998,Morimae2014}, random quantum circuits \cite{Aaronson2016b,Boixo2016}, among others. These results follow a similar reasoning: one identifies some computational problem that is expected to be hard (such as computing the permanent of Gaussian matrices, in the case of boson sampling), posits a hardness conjecture regarding that problem, and shows formally how that conjecture implies that an efficient classical algorithm for simulating the physical system in question would have unexpected or surprising complexity-theoretic consequences. And so, the computational problem that is ``solved'' by these restricted models is just to simulate their own behavior (according to the predictions of Quantum Mechanics), which a classical computer should not be able to do efficiently.

Shortly after the seminal boson sampling \cite{Aaronson2013a} paper was published, several works started investigating the effects of realistic experimental imperfections on the idealized theoretical model. The robustness of boson sampling was analyzed under the effect of partial photon distinguishability \cite{Shchesnovich2015,Renema2017}, fabrication imperfections on the linear-optical transformation \cite{Kalai2014,Arkhipov2015,Leverrier2014}, losses \cite{Aaronson2016,Rahimi-Keshari2016}, probabilistic sources \cite{Motes2013} and so on. On the experimental side, several small-scale implementations of boson sampling have been reported so far \cite{Broome2013,Crespi2013b,Spring2013,Tillmann2013,Spagnolo2014,Carolan2014,Carolan2015,Bentivegna2015,Loredo2017,He2017}, with state-of-the-art implementations with up to four photons using near-deterministic quantum dot sources \cite{Loredo2017,He2017}. Through the development of the field there has been much interaction between the experimental and theoretical communities in order to assess the most important obstacles towards a large-scale implementation of boson sampling. One successful example of such an interaction was the conception of scattershot boson sampling \cite{Lund2014,Scottblog,Bentivegna2015}, an alternative approach that allows circumvention of the exponential slowdown incurred due to the use of probabilistic sources. From the complementary perspective, recently developed classical algorithms \cite{Neville2017,Clifford2017} allow simulation of up to 40-photon boson sampling experiments in current supercomputers. 

One of the main obstacles that still hinders the scaling of boson sampling, both on the theoretical and experimental sides, is photon loss. In \cite{Aaronson2016}, the authors investigated the complexity of lossy boson sampling, and concluded that it remains as hard as standard boson sampling when a constant number of photons is lost. Unfortunately, this falls short of the realistic regime where we expect some \emph{fraction} of the photons to be lost. Even worse, most current implementations \cite{Broome2013,Crespi2013b,Spring2013,Tillmann2013,Spagnolo2014,Carolan2014,Carolan2015,Bentivegna2015,Loredo2017,He2017} are expected to suffer from \emph{exponential} losses. The intuition behind this claim is the following. Every time a photon traverses a beamsplitter, it has some probability of being lost, say $(1- \eta)$. And so, if it must traverse $s$ beamsplitters in a network, the probability that it arrives at the output should decay as $\eta^s$. The best known lower bound on the depth of a linear-optical circuit for it to satisfy the complexity requirements laid out in \cite{Aaronson2013a} is that it grows as $n \log n$. Thus, the probability that each photon survives the experiment decays exponentially and quickly becomes negligible.

It is then a timely question to ask how much loss a boson sampling device can support before it admits an efficient classical simulation. To our knowledge, the best known bound of this type states that an efficient classical simulation becomes possible when all but $\log n$ photons are lost (this was stated without proof in \cite{Aaronson2016}, but is a straightforward consequence of the more recent algorithm of \cite{Clifford2017}). In \cite{Rahimi-Keshari2016}, the authors also investigate a similar question, and show that linear optics becomes classically simulable (via a positive quasiprobability representation) if a fraction of the photons is lost, \emph{as long as} this is accompanied by a constant dark-count probability per detector.

In this work we propose a procedure to classically simulate certain linear-optical processes in a regime where losses are large. We do this by considering different physically-motivated models of loss. First we assume that exactly $n-l$ out of $n$ photons have been lost, and show that efficient classical simulation is possible whenever $l$ scales slower than $\sqrt{n}$. We then consider a more usual loss model, when each photon has some (mode-independent) probability of being lost in the circuit, and so the output state does not have a well-defined number of particles. We show that, also in that case, if the \emph{average} number of remaining photons is less than $\sqrt{n}$ an efficient classical simulation is possible. Finally, we consider the most general case where the photons can traverse some linear-optical network of beamsplitters and where the loss probability can depend on the path followed by each photon. In this case we prove that, if the smallest number of lossy elements any input photon has to traverse in its path through the network is $s$, we can effectively model that network as $s$ rounds of uniform losses followed by a linear-optical channel with an efficient classical description. This leads to an efficient classical algorithm for lossy linear-optical networks of depth greater than $C  \log  n$, for some suitable constant $C$. 

Our results also have implications outside boson sampling. First, in order to obtain our main result in \cref{thm:main1} we need to find the best approximation, in trace distance, to lossy bosonic states via particle-separable states (i.e.\ convex combinations of symmetric product pure states). To this end we use techniques that can be useful in the study of  entanglement properties of symmetric states \cite{Ghune2009,YU2016,Quesada2017}. Our results might also be relevant for the design and characterization of complex linear-optical networks \cite{Reck1994,Clements2016,Carolan2015}. In \cref{thm:PULLoutNOISE} we prove that a network where the smallest number of lossy elements in any path between inputs and outputs is $s$ (but that has an otherwise arbitrary geometry) is equivalent to another network with $s$ rounds of mode-independent losses at the input. This proves the intuition expressed previously regarding the exponential losses in current implementations of linear-optical networks, but also formally justifies the approximation of mode-independent losses which is often assumed in the literature \cite{Aaronson2016,Rahimi-Keshari2016,Neville2017,Demkowicz2015}.

Our paper is organized as follows. In \cref{sec:background} we give some theoretical background that underlies our work. In particular we review how linear-optical processes can be described both in first and second quantization formalisms, how different models of loss can be described mathematically, and some foundations on the boson sampling problem. In \cref{sec:genARG} we describe the general idea of our simulation algorithm. In \cref{sec:allmainresults} we describe our main results and discuss their consequences, although we defer the technical proofs to \cref{sec:mainProofs}. In \cref{sec:bosonsampling} we present the relation of our results to other works, with particular emphasis on how they fit in the  complexity-theoretic formalism of boson sampling. In \cref{sec:mainProofs} we detail the proofs of the results of \cref{sec:allmainresults}. Finally, in \cref{sec:conclusions} we conclude with some general remarks and open questions.

{\textbf{Notation:}} Throughout the paper we use the following notation. Given two positive-valued functions $f$ and $g$, we write  $f=\o(g)$ if $\lim_{x\rightarrow\infty} f(x)/g(x) = 0$ and   $f=\O(g)$ if $ \lim_{x\rightarrow\infty} f(x)/g(x) < \infty$. Likewise, we say $f=\omega(g)$ if $f=\o(g)$. Finally, $f\approx g$ will refer to the situation when $ \lim_{x\rightarrow\infty} f(x)/g(x)=1$.

\section{Theoretical background} \label{sec:background}

\subsection{Description of bosonic states} \label{sec:defCONC}

In this work we consider systems of bosons that can occupy $m$ modes. One natural description of these systems is in the language of second quantization, commonly used in the context of quantum optics \cite{Bachor2004,GaussQuantInfo2012} to describe the boson sampling problem.  Here we focus on the complementary description based on first quantization, more natural for the description of the so-called particle entanglement, a concept that will prove useful in what follows.  

In the language of first quantization, a system of $n$ particles is described by the Hilbert space $\symm{n}{m}\eqdef \Sym^n \L \C^m \R$, i.e.\ the symmetric subspace of $\H_{n,m}\eqdef \L \C^m \R^{\otimes n}$, which is the Hilbert space of $n$ distinguishable particles that can occupy $m$ modes. This reflects the fact that bosonic wave-functions are symmetric upon exchange of particles. A basis $\lbrace \ket{i} \rbrace_{i=1}^m$ of the single-particle space $\C^m$ defines the so-called Dicke basis of $\symm{n}{m}$. Elements of this basis are labeled by $m$-element tuples $\n=\L n_1,\ldots,n_m\R$ of non-negative integers $n_i$ satisfying $\sum_{i=1}^m n_i = n$. The Dicke state $\ket{\n}$ can be defined by $ \ket{\n}\eqdef N(\n) \proj{n}{m} \ket{\bi}$, where $\proj{n}{m}$ is a projector onto $\symm{n}{m}$ (acting on $\H_{n,m}$) and $\ket{\bi}\eqdef \ket{i_1} \ldots \ket{i_n}$. The relation between $\n$ and $\bi$ is as follows: $n_k$ equals the number of times $\ket{k}$ appears in  $\ket{\bi}$, and so the basis $\ket{\n}$ is also called occupation-number basis.     Finally, $N(\n)$ is a normalization factor [see \cref{eq:coefFORM}]. We also define the shorthand for the standard state
\begin{equation}\label{eq:defINPUTstate}
\ket{\Psi_n}\eqdef |\overbrace{1,\ldots,1}^{n},\overbrace{0,\ldots,0}^{m-n} \rangle,
\end{equation} 
which is often assumed to be at the input of a boson sampling device (see \cref{sec:BosSamProp}). Additionally, by $\true$ we denote the projector onto $\ket{\Psi_n}$.

In our analysis, a predominant role will be played by (bosonic) \emph{particle-separable} or \emph{symmetric separable states}. For a fixed number of particles $l$, they are defined \cite{Eeckert2002,Killoran2014} as states $\sigma$ that can be written as convex combinations of pure product bosonic states,  i.e.\
\begin{equation}\label{eq:partSEP}
\sigma =\sum_{\alpha} p_{\alpha} \kb{\phi_\alpha}{\phi_\alpha}^{\otimes n}
\end{equation} 
where $\lbrace p_\alpha \rbrace$ is some probability distribution and $\kb{\phi_\alpha}{\phi_\alpha}$ are states on a single-particle space $\C^m$.  In what follows we denote the set of symmetric separable $n$-particle states by $\Sep\L \symm{n}{m} \R$. 
Bosonic $n$-particle states that cannot be decomposed in this form are called \emph{particle-entangled}. 
When we restrict the standard notion of entanglement of distinguishable particles from $\H_{n,m}$ to $\symm{n}{m}$, it coincides with the above definition of particle entanglement.
It is also a resource for quantum sensing \cite{Toth2014,Demkowicz2015} and necessary for violations of Bell inequalities in the presence of superselection rules \cite{Wasak2016}.

If we remove the constraint on the total number of particles, the corresponding Hilbert space (known as the bosonic Fock space) is the direct sum of the corresponding Hilbert spaces $\SYM{m} \eqdef \bigoplus_{n=0}^{\infty} \symm{n}{m}$, where $\symm{0}{m}$ is the one-dimensional space spanned by the Fock vacuum  $\ket{\OO}$. The occupation number states $\ket{\n}$ (also called Fock states in this context) constitute a basis of $\SYM{m}$.  The notion of particle separability extends naturally to  $\SYM{m}$: a state $\tilde{\sigma}$ supported on $\SYM{m}$ is particle-separable [i.e.\ $\tilde{\sigma}\in\Sep\L \SYM{m} \R$)] if it can be expressed as
\begin{equation}\label{eq:diffNsep2}
\tilde{\sigma} = \sum_{n=0}^\infty p_n  \sigma^{(n)}\  ,
\end{equation}
where $\lbrace p_n \rbrace$ is a probability distribution and $\sigma^{(n)}\in\Sep\left(\Hsym_{n,m} \right)$. From our perspective, particle-separable states will be important as they yield classically-simulable instances of boson sampling (see \cref{sec:BosSamProp}).

Second quantization constitutes a complementary language for the description of bosonic systems. In this formalism, a central role is played by creation ($a^{\dag}_i$) and annihilation  ($a_i$) operators that act on $\SYM{m}$ and satisfy canonical commutation relations $[a_i,a^{\dag}_j] =\delta_{ij}$. These operators define the number operator on mode $\ket{i}$, $\hat{n}_i \eqdef a^{\dag}_i a_i$, which satisfies the relation $\hat{n}_i \ket{\n}=n_i \ket{\n}$. Fock states can then be expressed as 
\begin{equation*}
\ket{\n}=\ket{n_1,\ldots,n_m}=\prod_{i=1}^m \frac{(a^{\dag}_i)^{n_i}}{\sqrt{n_{i}!}}\ket{\OO}.
\end{equation*}
Every operator acting on $\SYM{m}$ can be expressed as a polynomial in creation and annihilation operators. This, together with the commutation relations, leads to the mode-dependent tensor product decomposition 
\begin{equation}\label{eq:modTENS}
\SYM{m} \approx \bigotimes_{i=1}^{m} \S^{(i)}_1\ ,
\end{equation}
 where $\S^{(i)}_1$ is the Fock space associated to the $i$th mode. This decomposition leads to the notion of \emph{mode separability} (and, by negation, mode entanglement), which differs from that of particle separability. In particular, Fock states admit the decomposition $\ket{\n} \approx \otimes_{i=1}^m \ket{n_i}$, while being particle-entangled whenever more than two modes are occupied. Conversely, particle-separable states $\ket{\phi}^{\otimes l}$ are generically mode-entangled.


\subsection{Linear optics in first and second quantization }

A lossless linear-optical transformation  is encoded by an $m\times m$ unitary matrix $U$. In the first quantization description, this matrix induces the independent evolution of particles as
\begin{equation}\label{eq:firstQUANTtran}
\rho \mapsto U^{\otimes n} \rho \L U^{\otimes n} \R^\dag ,\ \text{for } \rho\ \text{ - a state on } \symm{n}{m}\ . 
\end{equation}
In second  quantization, the action of $U$ can be conveniently described as a transformation on creation  operators 
\begin{equation}\label{eq:secondQUANTtran}
a^{\dag}_{i, \mathrm{in}}\mapsto a^{\dag}_{i, \mathrm{out}} = \sum_{j=1}^m U_{ij} a^{\dag}_{i, \mathrm{in}}
\end{equation}
where $a^{\dag}_{i, \mathrm{in}}$ and $a^{\dag}_{i, \mathrm{out}}$ are the operators for the $i$th mode at the input and output, respectively, of the linear-optical transformation. Transformations described by the above equations give rise to the quantum channel $\Lambda_U$ acting on states defined on $\SYM{m}$.

Consider now the paradigmatic boson sampling instance, where an $n$-photon input state $\true$ [see \cref{eq:defINPUTstate}] evolves according to an $m$-mode linear-optical transformation $U$ and is measured at the end in the photon-number basis, as represented in \cref{fig:linearoptics}. One way to obtain a classical simulation of this system is to map it to a regular quantum circuit, and then leverage known results for classical simulation of quantum circuits. The representations defined in \cref{eq:firstQUANTtran} and \cref{eq:secondQUANTtran} lead to two methods to perform this mapping.

\begin{figure}
    \centering
    \includegraphics[width=0.3\textwidth]{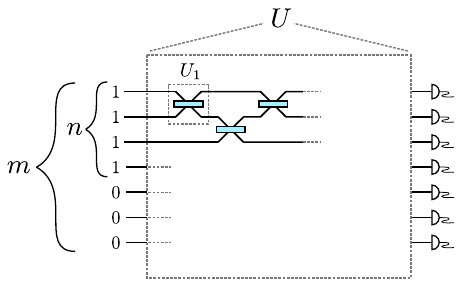}
    \caption{A linear-optical circuit with ``boson sampling resources'', i.e.\ Fock input states and number-resolving detectors.}
    \label{fig:linearoptics}
\end{figure} {}

In the first-quantized description arising from \cref{eq:firstQUANTtran}, each of the $n$ photons gets mapped to a single $m$-level system, which now labels in which mode that photon is. We represent the corresponding quantum circuit in \cref{fig:first_second_quant}(a). In order to simulate a bosonic state, the system must be initialized in a fully symmetric state.
The linear-optical transformation $U$ is then implemented as the local unitary transformation $U^{\otimes n}$. This description treats all linear-optical transformations in the same footing, and does not bring much insight e.g.\ into geometrical properties of a circuit. On the other hand, several properties of the {\em input state} are transparent. For example, if one photon is perfectly distinguishable from the others (maybe by some internal degree of freedom such as polarization), the simulation of \cref{fig:first_second_quant}(a) just leaves that particle out of the symmetrization step. Thus, particle indistinguishability in the physical system translates into entanglement in the first-quantized simulation. Another effect with a natural manifestation in the first-quantized simulation is photon loss, as we discuss in more detail in the next Section.

In the second-quantized description that arises from \cref{eq:secondQUANTtran}, we have a decomposition according to \cref{eq:modTENS} where each of the $m$ {\em modes} is mapped to an $n$-level system labeling its occupation number. In this case, each 2-mode transformation in the original linear-optical circuit is mapped to a gate acting only on the two corresponding subsystems in the simulation, as in \cref{fig:first_second_quant}(b). This representation has the benefit of preserving geometrical properties of the circuit, and we might try to adapt results that limit the computational power of quantum circuits based e.g.\ on their depth \cite{Terhal2004,Jozsa2006,Brod2015}.

\begin{figure}
    \centering
    \includegraphics[width=0.7\textwidth]{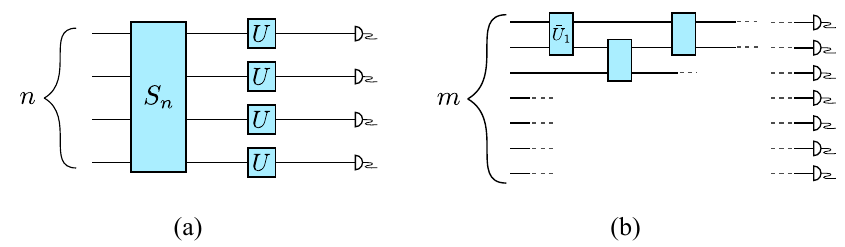}
    \caption{Two different quantum circuits that can simulate the system of \cref{fig:linearoptics}, based on the two ways to represent bosonic particles.}
    \label{fig:first_second_quant}
\end{figure}

\subsection{Losses in linear optics}  \label{sec:bslosses}
Let us now discuss how losses affect linear-optical systems. We survey three models of losses: the fixed-loss model and the beamsplitter loss model, the latter which can be divided in the uniform case (i.e.\ where all modes are subject to the same loss) and the non-uniform case describing more general (passive) linear-optical networks.

\begin{figure}
    \centering
    \includegraphics[width=0.3\textwidth]{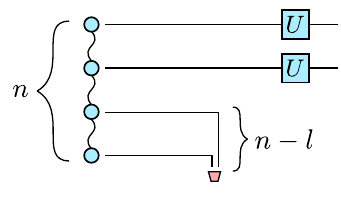}
    \caption{The fixed-loss model. Since the number of lost particles, $n-l$, is fixed, the effect of losses is easy to describe in the first-quantization mapping by tracing out $n-l$ subsystems. Here the initial state is the bosonic symmetric state, which can be viewed either as the physical state itself, or as the state in the simulation of \cref{fig:first_second_quant}(a) after the symmetrization procedure.}
    \label{fig:fixedlosses}
\end{figure}

\paragraph*{Fixed-loss model---} In the fixed-loss model, we have initially $n$ particles and assume that {\em exactly} $n-l$ of them are lost. If we further assume that losses are mode-independent, i.e.\ any photon is equally likely to have been lost, this model has a very simple interpretation in the first-quantization picture. The loss corresponds to tracing out $n-l$ of the $n$ particles, as shown in \cref{fig:fixedlosses},
\begin{equation}\label{eq:fixedLOSSdef}
\rho \mapsto \tr_{n-l}\L \rho \R\ . 
\end{equation}
Since $\rho$ is symmetric it does not matter which $n-l$ particles we trace over. By standard mathematical properties of trace we can show that these losses commute with any linear-optical unitary transformation,
\begin{equation}\label{eq:commutationOFlosses}
\tr_{n-l}\L \U^{\otimes n} \rho (U\dg)^{\otimes n} \R = U^{\otimes l} \tr_{n-l}\L \rho \R (U\dg)^{\otimes l}\ .  
\end{equation}
This is important as it implies that it does not matter whether losses occurred before or after the linear-optical transformation. 

The fixed-loss model was used in \cite{Aaronson2016}\footnote{In fact, the authors of \cite{Aaronson2016} only investigate boson sampling when losses occur at the input to the circuit, and leave as an open question whether their results also apply when losses happen at the output. The observation below \cref{eq:commutationOFlosses} answers that question affirmatively.} and in the context of quantum metrology in \cite{Oszmaniec2016}. Physically, this model is suitable for cases where number of lost particles can be effectively controlled and accounted for, e.g.\ in atomic interferometry \cite{Schumm2005,Sebby2007,Zhang2012,Stroescu2015}. It is also more convenient for defining lossy boson sampling as a \emph{computational} problem, since we can restrict our attention to the number of losses observed in a single run of the device with fewer assumptions about the physical process that caused them \cite{Aaronson2016}. In the context of optics, however, a more faithful physical model of losses is given by the beamsplitter model \cite{Barnett1998}, which we describe now. 

\begin{figure}[h]
    \centering
    \includegraphics[width=0.4\textwidth]{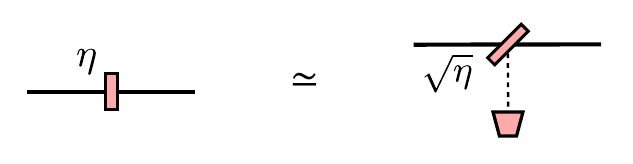}
    \caption{The beamsplitter model of losses. On the left we have the photonic mode passing through a lossy element with survival probability $\eta$. On the right we have an equivalent circuit where the lossy element is replaced by a beamsplitter with transmission probability $\sqrt{\eta}$, in which case each photon in the input is routed into the unmeasured mode with probability $1-\eta$.}
    \label{fig:bslosses}
\end{figure}

\paragraph*{Uniform beamsplitter loss model---}  In the beamsplitter loss model, whenever a photon reaches a lossy element, it is lost with some probability $(1-\tran)$. For example, a state $\ket{1}$ with a single photon, when passing through a lossy element, leaves in state $(1-\eta) \ketbra{0}{0} + \eta \ketbra{1}{1}$. We call this the beamsplitter model as it is equivalent to replacing the loss by a beamsplitter with transmissivity $\sqrt{\eta}$, as in \cref{fig:bslosses}. This is the most common way of modeling losses in the quantum optics literature \cite{Barnett1998,Demkowicz2015}.  If losses are mode-independent they can be conveniently described in the first quantization picture\footnote{For the formal proof of \cref{eq:BSuniformMODEL}  see the appendix of  \cite{Oszmaniec2016}.}:  $n$-particle states are transformed by the channel $\Lambda_\tran$,
\begin{equation}\label{eq:BSuniformMODEL}
\rho \mapsto \Lambda_\tran \L \rho \R \eqdef \sum_{l=0}^n \binom{n}{l} \tran ^l (1-\tran)^{n-l} \tr_{n-l} \L \rho \R\ ,
\end{equation}
where $\tran^l (1-\tran)^{n-l}\binom{n}{l}$ is the probability that exactly $l$ particles remain. Comparing \cref{eq:fixedLOSSdef}  and \cref{eq:BSuniformMODEL} we see that the fixed-loss and uniform beamsplitter loss models are closely connected, which we use often in this work. In particular, in \cref{sec:modelequiv} we show that in certain regimes of parameters $l$ and $\eta$ they are also \emph{computationally} equivalent from the perspective of complexity theory. 

\paragraph*{General lossy optical networks---} To analyze the case when losses affect each mode differently we turn our attention to general passive linear-optical networks. This provides a convenient mathematical framework to formulate this question, but also corresponds to what is commonly done in experiments. Since arbitrary multimode transformations are not readily available, the simplest way to implement a complex $m$-mode linear-optical operation is to decompose it as a network (or circuit) of elements acting on a few modes at a time, such as in \cite{Reck1994,Clements2016}. If each of the smaller elements has some loss associated with it, this might lead to a distribution of losses that affects a photon differently depending on which path it takes within the network, and which is determined by the overall geometry of the circuit.  

Throughout most of this work we assume that a linear-optical network is composed entirely of 2-mode transformations, i.e.\ beamsplitters and phase shifters. We further assume, for simplicity, that losses are caused only by beamsplitters, disregarding those associated with phase shifters, detectors, sources and transmission paths between beamsplitters, and furthermore that each beamsplitter induces the same loss probability $(1-\eta)$ in each of its input arms. Thus, the building block of a network is the element depicted in \cref{fig:buildingblock}, where red elements represent pure losses (as in \cref{fig:bslosses}) and blue collectively represents beamsplitters and any neighboring phase shifters (although we refer to them only as beamsplitters from hereon). From \cref{eq:commutationOFlosses} and \cref{eq:BSuniformMODEL} it follows that a round of uniform losses on $m$ modes commutes with an $m$-mode passive linear-optical transformation, and in particular this holds for $m=2$ as shown in \cref{fig:buildingblock}. For a network of arbitrary geometry, however, the analysis is more complicated since it is not immediately obvious to what extent losses can be commuted within the network. We return to this question in \cref{sec:allmainresults}.

\begin{figure}
    \centering
    \includegraphics[width=0.4\textwidth]{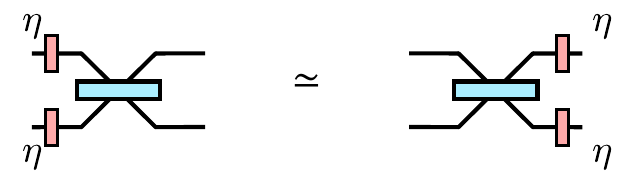}
    \caption{The building block of a linear-optical network used throughout the paper. Red elements are loss elements defined as in \cref{fig:bslosses}, while blue are lossless two-mode linear-optical transformations, including beamsplitters and phase shifters. It is not hard to show that a round of uniform losses on $m$ commutes with any linear optics on the same $m$ modes, and here we depict this for $m=2$.}
    \label{fig:buildingblock}
\end{figure}


We will describe the action of a general passive linear-optical network $\mathcal{N}$ by a channel $\Lambda_\mathcal{N}$, which can be obtained by composition of the corresponding channels for loss elements and beamsplitters. The explicit description of $\Lambda_\mathcal{N}$ in terms of \cref{eq:BSuniformMODEL} and \cref{eq:firstQUANTtran} is possible in principle, but can be very cumbersome. For a more practical description, we direct the interested reader to a formalism that describes it in terms of the action of linear-optical transformations on coherent states \cite{RahimiKeshari2013}.

\subsection{Boson sampling}\label{sec:BosSamProp}

Here we consider the simulation of general lossy linear optics, but the main setup we consider is the boson sampling model of quantum computing \cite{Aaronson2013a}. The ideal boson sampling instance is defined by three ingredients: (i) preparation of an $n$-photon, $m$-mode input state, (ii) application of an $m$-mode unitary linear-optical transformation $U$, and (iii) measurement of occupation number in every output mode. This is exactly embodied in the physical setup described in \cref{fig:linearoptics}. From the action of $U$ on  bosonic operators [cf.\ \cref{eq:secondQUANTtran}], it is easy to show that the transition probability from some input state $\ket{\n}$ to some output state $\ket{\mathbf{p}}$ is given by 

\begin{equation} \label{eq:permanentboson}
\textrm{Pr}(\n \rightarrow \mathbf{p}) = \frac{\left\vert\textrm{Per}(U_{\n,\mathbf{p}})\right\vert^2}{n_1! \ldots n_m! p_1! \ldots p_m!},
\end{equation}
where $U_{\n,\mathbf{p}}$ is an $n \times n$ submatrix of $U$ constructed by repeating $n_i$ times the $i$th column of $U$, and $p_j$ times its $j$th row \cite{Scheel2004,Aaronson2013a}, and Per$(A)$ denotes the permanent of matrix $A$ \cite{Valiant1979}. Throughout this paper we assume that the input is the state $\ket{\true}$ from \cref{eq:defINPUTstate}, i.e.\ that the photons initially occupy the first $n$ modes, without loss of generality. 

In \cite{Aaronson2013a} the authors showed that, modulo certain complexity-theoretic conjectures, it would be impossible for a classical computer to efficiently simulate the output of this idealized process, even in an approximate sense. We defer a more technical discussion of the results of \cite{Aaronson2013a} to \cref{sec:bosonsampling}, when we discuss the complexity-theoretic consequences of our own results. For now it suffices to say that the existence of an efficient classical algorithm capable of simulating a distribution sufficiently close to the ideal boson sampling one would imply a collapse of the polynomial hierarchy to the third level, which is considered a very unlikely outcome in the complexity theory literature.

The results of \cite{Aaronson2013a} concern the possibility of classical simulation of boson sampling, but there are several notions of simulation one can use. A \emph{strong} simulation of a quantum process is a classical algorithm that outputs the probability of any outcome, which may also include marginal probabilities. A \emph{weak} simulation is a classical algorithm that produces a sample from the same distribution as the quantum device. It is generally regarded that, from the point of view of assessing whether some quantum device outperforms classical computers, the notion of weak simulation is more adequate, as it requires the classical algorithm to perform the same task as the quantum device. 

Within the notion of weak simulation we can consider an exact simulator (which produces samples from the exact ideal distribution) or an approximate simulator. Given the ideal distribution ${p_x}$, an approximate simulator samples from some distribution ${q_x}$ such that
\begin{equation} \label{eq:TVDbs}
\frac{1}{2} \sum_x |p_x - q_x|\leq \epsilon
\end{equation}
for some $\epsilon > 0$. In the original boson sampling paper \cite{Aaronson2013a}, the authors require a classical algorithm to simulate the quantum device only in the approximate weak sense described above. Furthermore, their result also requires that the runtime of the classical algorithm is poly$(n, 1/\epsilon)$. In other words, the classical algorithm should be able to give better approximations at the cost of only polynomial more time. 

In our main results of \cref{sec:allmainresults}, by contrast, the classical algorithm is only able to perform weak simulation up to some threshold precision $\epsilon(n)$. We cannot improve the simulation simply by spending more computational time, but under some conditions the precision of our algorithm improves with the size of the problem. This means it is a different notion of approximation from the one used e.g.\ in \cite{Aaronson2013a}, although other recent papers use similar definitions to ours \cite{Bremner2016}. We discuss these distinctions and complexity-theoretic implications of our results in \cref{sec:bosonsampling}. For a more thorough discussion on this subject, see \cite{Nest2011,Pashayan2017}.

One last important thing to point out is that there are cases where boson sampling is known to be classically simulable, such as when inputs and measurements are in the Gaussian basis \cite{Bartlett2003}, or when there is a combined high rate of losses in the system and dark counts in the detectors \cite{Rahimi-Keshari2016}.  One such situation which will be important later on is when photons are perfectly \emph{distinguishable} (e.g.\ due to some undetected degree of freedom such as polarization). In this case the transition probabilities are given by
\begin{equation} \label{eq:permanentclassical}
\textrm{Pr}(\n \rightarrow \mathbf{p}) = \frac{\textrm{Per}(\left\vert U_{\n,\mathbf{p}}\right\vert^2)}{p_1! \ldots p_m!},
\end{equation}
where $\left\vert U_{\n,\mathbf{p}}\right\vert^2$ is obtained from $U_{\n,\mathbf{p}}$ by taking the absolute value squared of each of its elements. This equation is surprisingly similar to \cref{eq:permanentboson}, except that now the matrix has only positive elements. This, however, is crucial in enabling an efficient classical simulation. One way to see this is to note that, although the permanent is a hard function in general, for matrices with only positive elements it can be approximated efficiently by a classical algorithm due to Jerrum, Sinclair and Vigoda \cite{Jerrum2004}. Alternatively, if one is satisfied with a weak simulation, one can do it simply by sampling the outcomes of the photons one at a time, as they do not interfere.
The first-quantization counterpart of this statement is that distinguishable particles correspond, in the simulation of \cref{fig:first_second_quant}(a), to particles that do not undergo the symmetrization procedure. So, if all photons are distinguishable, the system simply evolves as a number of parallel $m$-level systems. 

Although we do not actually consider partial photon distinguishability, note that the bosons in a particle-separable state such as in \cref{eq:partSEP} can be considered as \emph{effectively} distinguishable. To see this, first let $T$ be linear-optical transformation that takes a single photon from state $\ket{1}$ into state $\ket{\phi}$. Then write
\begin{equation}
(\kb{\phi}{\phi})^{\otimes n} = T^{\otimes n} (\kb{1}{1})^{\otimes n} (T\dg)^{\otimes n}.
 \end{equation}
That is, any symmetric \emph{product} state is equivalent to initializing all photons in the same mode and applying some linear-optical transformation. But now note that if all photons start in the same mode, the submatrix constructed in \cref{eq:permanentboson} has only $n$ repeated columns, and in this case the permanent collapses simplify to a combinatorial factor multiplied by the product of all elements in that column. This, in turn, becomes identical to same probability defined for distinguishable particles in \cref{eq:permanentclassical}. Thus, we conclude that any bosonic particle-separable state effectively behaves as classical mixture of states of distinguishable particles, which is crucial for our simulation scheme in what follows. In the first-quantization picture, this statement corresponds to the fact that the only $n$-particle state that is symmetric and where all particles are in the same state is a product state.

\section{General outline of approximate classical simulation}\label{sec:genARG}

Our classical simulation results all follow from the same general reasoning, which we present in this section. Let us first introduce two definition of distance to measure the quality of our simulations. The first is the \emph{total variation distance}, which measures the distance between two probability distributions $P=\{p_x\}$  and $Q=\{q_x\}$ and can be written as 
\begin{equation}\label{eq:TV distance}
\left \lVert P-Q \right \lVert \eqdef \frac{1}{2} \sum_x \lvert p_x - q_x \lvert\  ,
\end{equation}
The total variation distance is the typical figure of merit in the context of boson sampling, as in \cref{eq:TVDbs}. The second measure  is the \emph{trace distance} between two quantum states $\rho$ and $\sigma$
\begin{equation}\label{eq:TRACEdist}
\dtr(\sigma,\rho) \eqdef \frac{1}{2} \mathrm{tr}\left[ \sqrt{(\rho- \sigma)\dg(\rho- \sigma)}\right].
\end{equation}
The trace distance can be also defined \cite{Nielsen2010} as the maximum, over all POVMs $\lbrace M_x \rbrace$, of the total variation distance between the probability distributions $\D_\rho$ and $\D_\sigma$ that arise when one measures $M_x$ on states $\rho$ and $\sigma$ respectively. Therefore, all such probability distributions $\D_\rho$ and $\D_\sigma$ satisfy
\begin{equation}\label{eq:UPPboundTV}
\left \lVert \D_\rho- \D_\sigma \right \lVert \leq \dtr(\sigma,\rho)\ .
\end{equation}

We are now ready to outline our classical simulation scheme. In all cases, we are interested in simulating a linear-optical protocol (as in \cref{fig:linearoptics}) that uses as input some state $\rho$, corresponding to the action of some lossy channel on the initial state $\ket{\Psi_n}$. 
The simulation is achieved by constructing an alternative state, $\sigma$, with the following properties (see \cref{fig:SEP})
\begin{itemize}
	\item[(i)] $\sigma$ is as a convex combination of particle-separable states (possibly with different numbers of particles):
	\begin{equation}
	\sigma = \sum_{l=0}^\infty p_l \sigma_l\ ,
	\end{equation}
	with $\sigma_l \in \Sep\L \symm{l}{m} \R$, and so $\sigma\in\Sep\L \SYM{m}  \R $. Every $\sigma_l$ can be further decomposed as $\sigma_l= \sum_{\alpha} q^{l}_\alpha \kb{\phi_{l,\alpha}}{\phi_{l,\alpha}}^{\otimes l}$.  We also require that the \emph{joint} probability distribution ${p_l q^{l}_\alpha}$ can be efficiently sampled on a classical computer.
	\item[(ii)] The trace distance $\dtr(\sigma,\rho)$ between $\sigma$ and $\rho$ is small, in a sense to be made precise later. 
\end{itemize}

\begin{figure}
    \centering
    \includegraphics[width=0.3\textwidth]{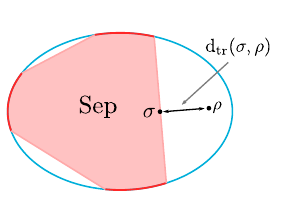}
    \caption{A schematic view of our simulation strategy. The set $\Sep$ is the subset of separable states within some state space of interest, and $\sigma$ is the closest state in $\Sep$ to the lossy state $\rho$.}
    \label{fig:SEP}
\end{figure}

Let $\DD(\rho;\Lambda)$ and $\DD(\sigma;\Lambda)$ be the probability distributions obtained by applying some linear-optical channel $\Lambda$ to $\rho$ and $\sigma$, respectively, and measuring the resulting state in the particle-number basis. By \cref{eq:UPPboundTV}  we have $\left \lVert\DD(\rho;\Lambda)-\DD(\sigma;\Lambda) \right \lVert \leq \dtr(\sigma,\rho)$. In other words, simulating any linear-optical channel $\Lambda$ acting on $\sigma$ will give a simulation of the corresponding channel acting on $\rho$, up to some error determined by $\dtr(\sigma,\rho)$. In the following sections $\Lambda$ might be $\Lambda_U$, $\Lambda_\eta$ or $\Lambda_\mathcal{N}$, as appropriate from context (cf.\ definitions in \cref{sec:defCONC}).

Our claim now is that sampling from the distribution $\DD(\sigma;\Lambda)$ can be done efficiently on a classical computer. We restrict our attention to \emph{weak} simulation, i.e.\ the task of classically producing a sample from the correct distribution. This can be done in three steps. First, we draw a pair $(\alpha,l)$ according to the probability distribution $\{p_l q^l_{\alpha}\}$, which can be done efficiently by assumption. Second, it is easy to simulate the effect of the linear-optical channel $\Lambda$ on $\kb{\phi_{l,\alpha}}{\phi_{l,\alpha}}^{\otimes l}$.  Indeed, in the first-quantization formalism [cf.\ \cref{eq:firstQUANTtran}] any linear-optical channel that we care about will act as $U^{\otimes l}$, i.e.\ as a tensor product of single-particle unitaries\footnote{At worse, the linear-optical channel will act as some channel $\Lambda$, but that can easily be described as a linear-optical transformation on at twice the number of modes.}. Finally, since the particle-number measurements can also be implemented by a product of single-particle measurements, the last step correspond simply to simulating $\O(n)$ independent parallel $m$-level quantum systems, which can be trivially be done on a classical computer.

Combining these observations, it is clear this procedure efficiently samples from a distribution closer than $\dtr(\sigma,\rho)$ in total variation distance to the desired one $\DD(\rho;\Lambda)$. In \cref{sec:allmainresults} we obtain explicit expressions for the separable states and their trace distances to the desired state under different physically-motivated loss models, and in \cref{sec:bosonsampling} we return to the complexity-theoretic consequences of our simulations and further interpretations. 

\subsection{Computational equivalence of loss models}  \label{sec:modelequiv}

Our main results include claims about the loss models described in \cref{sec:bslosses}. There we described how the models in which losses are mode-independent are mathematically related, namely via the binomial distribution of \cref{eq:BSuniformMODEL}. From that expression it is clear that, in the beamsplitter model with uniform transmission probability $\eta$ per mode, the average number of lost photons is $(1- \eta) n$. Thus, if we set $\eta = 1- k/n$, we obtain any instance of the fixed-loss model with exactly $k$ losses as the \emph{expected} outcome in the beamsplitter model. However, that does not suffice to prove that the models are \emph{computationally} equivalent. To achieve that, we need to show under which conditions one model can simulate the other efficiently, i.e.\ with only polynomial (in $n$) overhead. This will allow us to re-purpose results obtained from one model into the other.

Consider two black boxes that simulate lossy $n$-photon boson sampling experiments. The first, \textsc{BS}, outputs samples according to the beamsplitter model for some uniform loss probability $(1- \eta)$ per photon, per mode. The second, \textsc{FL}, outputs samples according to the fixed-loss model with $k$ losses, for some range of values of $k$. The formal question is whether we can use poly$(n)$ samples produced by one black box to output a single sample from the distribution generated by the other (with high probability), and what are the corresponding asymptotic behaviors of $k$ and $\eta$.

To prove that \textsc{BS} can simulate \textsc{FL} note that, for large $n$, we can use Stirling's approximation to write
\begin{align}
\textrm{Pr} (n \eta) & = \binom{n}{n \eta} \eta^{n \eta} (1- \eta)^{n(1-\eta)} \notag \\
& \approx \sqrt{\frac{1}{2\pi \eta(1- \eta)n}}.
\end{align}
That is, since $\eta\in[0,1]$, the probability of observing exactly the expected number of losses is O$(1/\sqrt{n})$. Thus, if we set $\eta = 1 - k/n$, we only need O$(\sqrt{n})$ samples from \textsc{BS} in order to obtain one sample from \textsc{FL} with high probability. This was used in \cite{Aaronson2016} to extend their hardness claim from the fixed-loss model [i.e.\ that it is computationally as hard as ideal boson sampling if $k = \textrm{O}(1)$] carried over to the beamsplitter model in the regime $1- \eta = \textrm{O}(1/n)$.

Consider now the reverse question, i.e.\ when can \textsc{FL} simulate \textsc{BS}? We already saw that, by setting $k = (1- \eta) n$, \textsc{FL} samples from a distribution with the same number of photons as the average of the \textsc{BS} distribution, but that is not enough. We also need to show that the distribution over photon numbers produced by \textsc{BS}, i.e.\ the distribution of \cref{eq:BSuniformMODEL}, is sufficiently concentrated around its mean. If that holds, then a \textsc{FL} black box that samples from events with O$(\eta n)$ transmitted photons can approximately simulate a \textsc{BS} black box by first sampling $k$ according to \cref{eq:BSuniformMODEL} truncated to a small range of values around the mean, and then outputting a distribution for that value of $k$.

The number of transmitted photons can be described by a random variable $X = \sum_{i=1}^{n} X_i$, where each $X_i$ is an independent random variable which takes value 0 is the $i$th photon was lost and 1 otherwise. Clearly, $X$ has the expectation $E(X) = \eta n$. Then, by a standard multiplicative Chernoff bound \cite{Chernoff1952} we can write:
\begin{equation}\label{eq:Chernoff}
\begin{aligned} 
\textrm{Pr} (X \leq (1 - \delta) \eta n) & \leq e^{- \frac{\delta^2 \eta n}{2} }, \\
\textrm{Pr} (X \geq (1 + \delta) \eta n) & \leq e^{- \frac{\delta \eta n}{3} }.
\end{aligned}
\end{equation}
We can repeat this argument for $Y = \sum_{i=1}^{n} (1-X_i)$, which is the number of lost photons and has expectation value $E(Y) = n (1- \eta)$.  The Chernoff bounds for $X$ and $Y$ show that, in the extremal regimes where the photons are either very likely to survive or to be lost (specifically if either $n (1- \eta)$ or $n \eta$ go to 0 sufficiently fast with $n$), then the number of photons at the end of the circuit concentrates around its average as $n$ increases. This happens, for instance, if the average number of photons lost is constant (as in \cite{Aaronson2016}), or if there are only an average of o$(\sqrt{n})$ photons left after the losses (as is the regime of our main results in following sections).

\section{Simulation of lossy linear optics} \label{sec:allmainresults}

In this Section we describe our main results, namely under which loss regimes linear optics can be simulated classically by approximating the lossy states by with particle-separable states. We defer most technical proofs to \cref{sec:mainProofs}, rather focusing here on the interpretation and consequences of the main theorems.

We begin in \cref{sec:resFIXED} by considering lossy linear optics in the fixed-loss model described in \cref{sec:bslosses}. Although this model is less realistic from a physical standpoint, it is mathematically elegant and leads to more straightforward results.

In subsequent sections we focus on lossy linear-optical \emph{networks}, defined in \cref{sec:bslosses}. Two common examples of these networks are depicted in \cref{fig:balanced_unbalanced} \cite{Clements2016,Reck1994}. These networks model good approximations for experiments based on integrated photonics such as \cite{Broome2013,Crespi2013b,Spring2013,Tillmann2013,Carolan2015}.

\begin{figure}[h]
    \centering
    \includegraphics[width=0.8\textwidth]{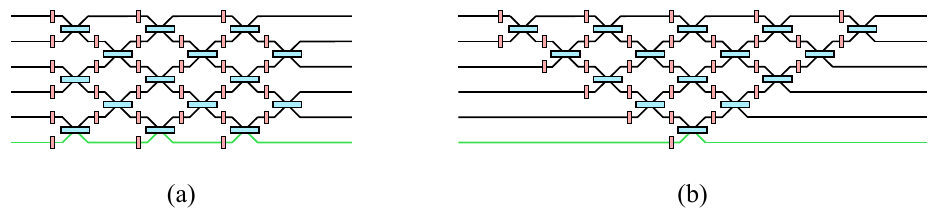}
    \caption{Two universal decompositions of linear-optical transformations. The building block in these diagrams and all that follow is as defined in \cref{fig:buildingblock}. The circuit in (a) does not have perfectly balanced losses because of the first and last modes, but that is a minor effect as the dimension of the unitary increases. The circuit in (b) is very unbalanced, since one mode always traverses a single lossy element, while there is another which traverses roughly $m$. The green paths indicate one example of a shortest path that a photon can take in the circuit, as will be used in \cref{thm:PULLoutNOISE}. We caution the reader not to confuse the red elements with phase shifters, as this figure is similar to representations of lossless linear-optics used in the literature.}
        \label{fig:balanced_unbalanced}
\end{figure}

In \cref{sec:unif} we initially address losses for ``uniform'' (or ``balanced'') networks, namely those where every mode is affected by losses in the exact same amount. This is a simpler case where loss elements commute with unitary transformation [cf.\ \cref{sec:bslosses}], and so all losses can be assumed to happen at the input. The results in this section rely heavily on the Chernoff bounds from \cref{sec:modelequiv} and the results from the fixed-loss model.

Finally,  in \cref{sec:nonUNIF} we justify the use of the uniform loss model for networks with arbitrary geometries, such as \cref{fig:balanced_unbalanced}(b). To this end we prove that it is possible to ``extract'' a certain amount of uniform losses from an arbitrary network. The relevant parameter for how much loss can be pulled out of the network is the minimal number of beamsplitters that a photon must traverse between any input and any output.

\subsection{Simulation in the fixed-loss model}\label{sec:resFIXED}

Suppose the the input of the bosonic state is $\ket{\true}$ (i.e.\ the first $n$ out of $m$ modes are occupied by a single particle each). Suppose then that we trace $n-l$ out of the $n$ particles, corresponding to losing exactly $n-l$ photons at the input [cf.\ \cref{sec:bslosses}]. We can write the resulting state as
\begin{equation} \label{eq:target}
\targ{l}= \frac{1}{\binom{n}{l}}\sum_{\n:\ |\n|=l,\ n_i \leq 1} \ketbra{\n}{\n}\ ,
\end{equation}
where the sum is over all occupation-number vectors of length $|\n|=l$ satisfying $n_i \leq 1$. We are now interested in finding the {\em separable} symmetric state $\sigma^* \in \Sep( \Hsym_{l,n} )$ that is closest in trace distance to $\targ{l}$. In other words we seek for $\sigma_\ast$ satisfying $\dtr(\sigma_\ast,\targ{l})= \dsep{l}$, where
\begin{equation}\label{eq:distanceKN}
\dsep{l} \eqdef \min_{\sigma\in\Sep(\Hsym_{l,n})} \dtr(\sigma,\targ{l})\ .  
\end{equation} 
As discussed in \cref{sec:genARG}, if we find $\sigma_\ast$ such that $\dsep{l}$ is suitably small, we can use $\sigma_\ast$  to simulate any boson sampling task performed on the state $\targ{l}$ up to some precision that depends on $\dsep{l}$. The following result establishes precisely the exact dependence of $\dsep{l}$ on $n$ and  $l$. We defer the proof of to \cref{sec:proofDIST}. Let us just state it for now, to investigate its basic consequences.

\begin{theorem}[Approximation of the lossy input state by a particle-separable state] \label{thm:main1}
Let $\targ{l}$ be the lossy bosonic state defined in \cref{eq:target}.  The trace distance $\dsep{l}$ of $\targ{l}$ to the set of symmetric separable $l$-particle states $\Sep(\Hsym_{l,n})$ is given by
\begin{equation}\label{eq:distanceValue}
\dsep{l}=1 - \frac{n!}{n^{l} (n-l)!}\ . 
\end{equation} 
Moreover, a separable state $\sigma_\ast\in \Sep(\Hsym_{l,n})$ that attains $\dsep{l}$ is
\begin{equation}\label{eq:optstateL}
\sigma_\ast=\frac{1}{(2\pi)^n}\int_0^{2\pi}  d \varphi_1 \ldots  \int_0^{2\pi}d  \varphi_n  \left( V_{\varphi_1,\ldots,\varphi_n} \ketbra{\phi_0}{\phi_0} V_{\varphi_1,\ldots,\varphi_n}\dg \right) ^{\otimes l}\ ,
\end{equation}
where $\ket{\phi_0}\eqdef (1/\sqrt{n})\sum_{i=1}^n \ket{i}$ and $V_{\varphi_1,\ldots,\varphi_n} \eqdef \exp\left(- \ii \sum_{i=1}^n \varphi_i \ketbra{i}{i} \right) $.
\end{theorem} 

Let us look at the asymptotic behavior of $\dsep{l}$ in the limit of large and small losses.  Using Stirling's approximation we write $\log \left(\frac{n!}{n^{l} (n-l)!} \right) = -\frac{l^2}{2n}+ \o \left(\frac{l^3}{n^2} \right)$, which leads to the following immediate result.

\begin{corollary}
After $n-l$ losses we have 
\begin{equation}
  \dsep{l}= \frac{l^2}{2n} +\o \left(\frac{l^3}{n^2} \right)\ .
\end{equation}
Therefore, if all but $l = \o (\sqrt{n})$ photons are lost, $\dsep{l}$ asymptotically tends to $0$. Conversely, if $l=\omega(\sqrt{n})$, then $\dsep{l}$ tends to 1.
\end{corollary}

All that remains is to show that it is indeed possible to use $\sigma^*$ efficiently in a classical simulation. But this is quite straightforward. The state $\sigma^*$ is obtained by starting with the state 
\begin{equation*}
\ket{\phi}^{\otimes  l} = \left(\frac{1}{\sqrt{n}}\sum_{i=1}^n e^{i \varphi_i}\ket{i}\right)^{\otimes l},
\end{equation*}
and averaging over all phases $\varphi_i$ uniformly in the $[0, 2\pi]$ interval. A \emph{weak} simulation then consists of choosing a set of phases $\varphi_i$ uniformly at random, then simulating the action of $U^{\otimes l}$ on the $\ket{\phi}^{\otimes l}$, together with the final single-qudit computational basis measurements. This is simply the simulation of $l$ parallel $m$-level systems\footnote{Recall that in general the total number of modes $m$ is greater than the number of modes $n$ on which the state $\targ{l}$ is supported.}, and can be done efficiently on a classical computer.  Combining this  with the discussion in \cref{sec:genARG} we obtain the following result.

\begin{corollary} \label{cor:fixed}
Let $\DD(\targ{l},\Lambda_U)$ be the probability distribution obtained by applying the linear-optical transformation $U$ to state $\targ{l}$ of \cref{eq:target} and measuring in the particle-number basis. There exists a probability distribution $\DD(\sigma_\ast,U)$ that can be classically sampled from efficiently and such that $\left \lVert\DD(\targ{l};\Lambda_U)-\DD(\sigma_\ast;\Lambda_U) \right \lVert \leq \dsep{l}$, where $\dsep{l}$ is given in \cref{eq:distanceValue}. Moreover, in the limit  $l=\o(n^{1/2})$ we have $\dsep{l}\approx \frac{l^2}{2n}$. 
\end{corollary}

To our knowledge, the previous best known bound of this type \cite{Aaronson2016} gave an efficient classical simulation when all but $l= \mathrm{O}(\log{n})$ of the photons were lost, since in that regime, the transition probabilities are associated with permanents of log-sized matrices, which can be computed classically in polynomial time. Although this is stated in \cite{Aaronson2016} without proof, it follows straightforward from a recent simulation boson sampling algorithm \cite{Clifford2017}.

As discussed previously, the model where a specific number of particles is lost is of limited practical relevance, since in an experiment this- number is typically a random variable. However, as we show next, insights gained from the analysis of this simplified model are useful to derive classical simulations for the more realistic ones.

\subsection{Uniform losses} \label{sec:unif}

Consider now a lossy linear-optical network under the assumption of uniform losses, i.e.\ every particle in every input can independently be lost with probability $1-\tran$. As discussed in \cref{sec:bslosses}, the action of such a loss channel onto the initial $n$-particle Fock state $\ket{\true}$ leads to the state
\begin{equation}\label{eq:FOCKlost}
\targ{\tran}=\sum_{l=0}^{n} \tran^l (1-\tran)^{n-l}\binom{n}{l} \targ{l}\ .
\end{equation}
To obtain an approximate classical simulation of the resulting linear-optical process, we choose to approximate $\targ{\eta}$ by states from $\Sep\L \SYM{n} \R$ i.e. convex combinations of symmetric separable states with \emph{different numbers of particles}, 
\begin{equation}\label{eq:diffNsep}
\tilde{\sigma} = \sum_{l=0}^\infty p_l  \sigma^{(l)}\  ,
\end{equation}
where $\lbrace p_l \rbrace$ is a probability distribution and $\sigma^{(l)}$ is a symmetric separable state supported on $l$ particles and $n$ modes. As an \emph{ansatz} for a good approximation to $\targ{\eta}$ we take the mixture of optimal $l$-particle separable states with appropriate binomial weights,
\begin{equation}\label{eq:guessSTATE} 
\sigma_\tran \eqdef \sum_{l=0}^n \eta^l (1-\eta)^{n-l} {\textstyle\binom{n}{l}} \sigma^{(l)}_\ast ,
\end{equation}
where $\sigma^{(l)}_\ast$ are given in Eq. \eqref{eq:optstateL}. The resulting trace distance between $\sigma_\tran$ and $\targ{\tran}$ is
\begin{equation}\label{eq:distETAsep}
\dsepB{\tran} \eqdef \dtr\left(\sigma_\tran,\targ{\tran}\right) = \sum_{l=0}^{n} \tran^l (1-\tran)^{n-l}\binom{n}{l} \dsep{l}\ ,
\end{equation}
where $\dsep{l}$ is given in Eq. \eqref{eq:distanceValue}. The following auxiliary result gives bounds on $\dsepB{\tran}$.
\begin{lemma}\label{lem:auxBOUNDS}
Let $ \dsepB{\tran}$ be the trace distance between states $\sigma_\tran$ and $\targ{\tran}$ defined in \cref{eq:FOCKlost} and \cref{eq:diffNsep} respectively. We have the following inequalities 
\begin{align}
\dsepB{\tran}  \leq & \frac{\tran^2 n}{2} + \frac{\tran(1-\tran)}{2}\ , \label{eq:UPPERbound1} \\  
\dsepB{\tran}   \geq & \dsep{\lceil \tran n (1-\delta) \rceil} \left[ 1 - \exp\L-\delta^2 n\tran/2\R  \right].   \label{eq:LOWERbound1} 
\end{align}
where $\delta\in[0,1]$ and $\lceil x \rceil$ denotes the smallest integer greater than $x$. 
\end{lemma}

The proof of the above lemma is given in \cref{app:technicRES}. Combining this result with the results of \cref{sec:resFIXED} and the arguments of \cref{sec:genARG} we conclude that it is possible to approximately simulate any boson sampling protocol that uses the lossy state $\targ{\tran}$ as input, provided that $\langle l \rangle= \tran n =\o(n^{1/2})$, where $\langle l \rangle$ is the average number of photons that are left in the network.  On the other hand, if $\langle l \rangle =\omega(\sqrt{n})$, then $\dsepB{\tran}$ goes asymptotically to $1$.  

\begin{corollary}\label{cor:SIMbeamsplitter}
Let $U$ be a linear-optical unitary matrix. Let $\DD(\targ{\tran};\Lambda_U)$ be the probability distribution obtained by applying $\Lambda_U$ to $\targ{\tran}$ [cf.\ \cref{eq:FOCKlost}] and measuring the resulting state in the particle-number basis. There exists a probability distribution $\DD(\sigma_\tran;\Lambda_U)$ that can be sampled efficiently on a classical computer such that $\left \lVert\DD(\targ{\tran};\Lambda_U)-\DD(\sigma_\tran;\Lambda_U) \right \lVert = \dsepB{\eta}$, where $\dsepB{\eta}$ is given in \cref{eq:distETAsep} and satisfies $\dsepB{\eta} \leq  \frac{\tran^2 n}{2} + \frac{\tran(1-\tran)}{2}$.
\end{corollary}

Our arguments give an approximate classical weak simulation to $\DD(\targ{\tran};\Lambda_U)$ only when $\tran=\o(n^{-1/2})$, which at first sight does not seem physically relevant. Usually, lossy devices (such as sources, beamsplitters or detectors) have loss probabilities that are constant, rather than decreasing with $n$. However, as we discuss in the next section, standard practical realizations of linear-optical networks often lead to dependence of the \emph{effective} overall transmissivity of the network, $\tran_{\mathrm{eff}}$, on the size of the system. Since the usual boson sampling paradigm requires networks that increase with the number of photons, this can easily ensure that $\tran_{\mathrm{eff}}$ decreases at least as $\tran_{\mathrm{eff}}=\o(n^{-1/2})$.

\subsection{Non-uniform losses}\label{sec:nonUNIF}

The effect of losses on linear-optical protocols has been extensively studied in the literature \cite{Barnett1998,Aaronson2016,Rahimi-Keshari2016,Oszmaniec2016,Neville2017}. A common assumption in these works  is that losses are uniform, as we considered in the previous section. This model can be justified, for example, by assuming that every path that a photon can take through a network is equally lossy. For some geometries of the network this can indeed be the case, such as that shown in \cref{fig:balanced_unbalanced}(a) \cite{Clements2016}. Moreover, this is a natural assumption for losses that model e.g.\ $m$ identical imperfect particle detectors \cite{Barnett1998,Rahimi-Keshari2016}. 

However, if the network has an unbalanced geometry, different paths through it accumulate different amounts losses, as in the case of \cref{fig:balanced_unbalanced}(b) \cite{Reck1994}. The analysis of such networks is more complex, as losses affecting only certain modes in general do not commute with linear-optical elements connecting them to different modes (see \cref{fig:lossesstuck}). Indeed, the main point of \cite{Clements2016} is to propose a balanced decomposition of a universal linear-optical network with the goal of making it more loss-resistant. The authors argue that, in a balanced network such as that of \cref{fig:balanced_unbalanced}(a), losses ``only'' act as overall losses, i.e.\ uniformly, whereas in the more commonly used decomposition of \cref{fig:balanced_unbalanced}(b) losses also behave as non-unitary noise that is harder to mitigate and characterize.

\begin{figure}[t]
    \centering
    \includegraphics[width=0.4\textwidth]{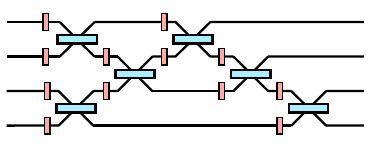}
    \caption{A lossy linear-optical network that is not balanced, in the sense that different paths within the network suffer different amounts of losses. In these cases, it is not always obvious how many layers of losses can be ``extracted'' from the circuit.}
    \label{fig:lossesstuck}
\end{figure}

To the best of our knowledge, until now there was no rigorous justification for the phenomenological assumption that losses are uniform, unless the network has a very special geometry such as in \cref{fig:balanced_unbalanced}(a). The following result remedies this by showing that it is always possible to rewrite a lossy network of arbitrary geometry as a round of uniform losses followed by another (possibly still lossy) linear-optical network.

\begin{theorem}\label{thm:PULLoutNOISE}
Let $\N$ be an $m$-mode linear-optical network constructed out of the two-mode elements of \cref{fig:buildingblock}, each inducing the same loss probability $(1-\tran)$ in both of its input modes. Let $s$ be the smallest number of beamsplitters traversed by any path connecting an input mode to an output mode in the network. Let $\Lambda_\N$ be the quantum channel associated with the entire network $\N$. Then we have the following decomposition of $\Lambda_\N$:
\begin{equation}\label{eq:nonBALdec}
\Lambda_{\N} = \Lambda_{\tilde{\N}} \circ \Lambda_{\tranef}\ ,  
\end{equation}
where $\Lambda_{\tranef}$ is a channel describing uniform losses with effective transmissivity $\tranef=\tran^s$, and $\tilde{\N}$ is a linear-optical network obtained by removing $s m$ loss elements from $\N$. The description of $\tilde{\N}$ from $\N$ can be obtained efficiently.
\end{theorem}

\begin{figure} [b]
    \centering
    \includegraphics[width=0.6\textwidth]{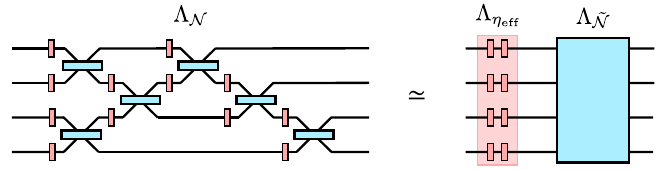}
    \caption{If $s$ is the smallest number of beam splitters that any path in the network traverses (in this example, $s=2$), \cref{thm:PULLoutNOISE} allows us to extract $s$ layers of losses from the circuit and gives us an efficient description of the remaining network $\tilde{\N}	$.}
    \label{fig:mainthm}
\end{figure}

A schematic representation of this Theorem is shown in \cref{fig:mainthm}. We present the complete proof of \cref{thm:PULLoutNOISE} (including a prescription to obtain $\tilde{\N}$ from $\N$) in \cref{sec:mainProofs} and focus now on its consequences.

A version of \cref{thm:PULLoutNOISE} is well-known if one can assume that the network $\N$ can be decomposed onto layers of commuting linear-optical elements, each layer covering all  modes. Since uniform $m$-mode losses commute with arbitrary $m$-mode linear-optical networks [cf.\ \cref{sec:bslosses}], whenever a network has a layer of uniform losses it can be  commuted out of the network as in \cref{fig:commutinglayer}. Thus, if a network has $s$ layers of this type they can be extracted layer by layer, leading to a decomposition as in \cref{thm:PULLoutNOISE} in a straightforward manner. \cref{thm:PULLoutNOISE} strengthens this statement by showing that it holds in more general geometries, as long as every path between an input and an output passes through $s$ loss elements. 

\begin{figure}
    \centering
    \includegraphics[width=0.5\textwidth]{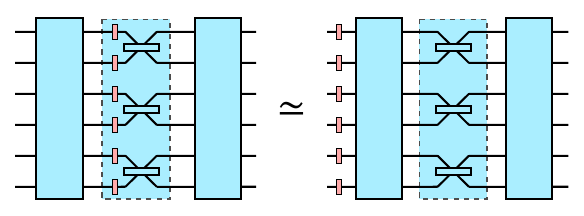}
    \caption{If a linear-optical network has a layer of equal loss elements on \emph{all} modes, those losses can be ``pulled out'' of the network, since a round of equal losses on every mode commutes with any linear optics. A linear-optical network where every layer has this property can be modeled exactly by the uniform model of loss.}
    \label{fig:commutinglayer}
\end{figure}

\cref{thm:PULLoutNOISE} already formalizes the claim, stated in \cref{sec:intro}, that losses in current implementations of boson sampling increase exponentially with the depth of the circuit (if one identifies the depth of the circuit with the definition $s$ in \cref{thm:PULLoutNOISE}), which is a serious scalability issue. Let us now consider the consequences of   this result for classical simulability.

First note that, from the efficient description of $\tilde{\N}$, we can efficiently sample from the probability distribution $\DD\L \kb{\phi}{\phi}^{\otimes l}, \Lambda_{\tilde{\N}}   \R$ (see \cref{sec:BosSamProp} and \cref{sec:genARG}). Furthermore, from \cref{eq:nonBALdec}, we have that $\DD\L \true; \Lambda_\N \R =\DD\L \Lambda_{\tranef} (\true); \Lambda_{\tilde{\N}} \R$. From the strategy outlined in \cref{sec:genARG}, if $\Lambda_{\eta_{\mathrm{eff}}}\L \true \R$ is close in trace distance to $\Sep\L \SYM{n} \R$  it is possible to classically sample from a distribution approximating the desired one $\DD\L \true; \Lambda_\N \R$. On the other hand, from \cref{lem:auxBOUNDS} we know that there exist a state $\sigma_{\tran_{\mathrm{eff}}} \in \Sep\L \SYM{n} \R$ such that $\dtr\L \Lambda_{\eta_{\mathrm{eff}}}\L \true \R, \sigma_{\tranef} \R  \leq \dsepB{\tranef}$ , where $\dsepB{\tranef}$ is defined in \cref{eq:distETAsep}. Putting this all together leads to the following result.

\begin{corollary}\label{cor:SIMarbit}
Let $\N$ be the $m$-mode linear-optical network of \cref{thm:PULLoutNOISE} and $\Lambda_\N$ the corresponding linear-optical channel. Let $s$ be the smallest number of beamsplitters traversed by any path between an input and an output mode. Let $\DD(\true;\Lambda_\N)$ be the probability distribution obtained by applying $\Lambda_\N$ to the standard boson sampling input state $\true$ and measuring the resulting state in the particle-number basis. There exists a probability distribution $\DD(\sigma_{\tranef};\Lambda_{\tilde{\N}})$ that can be sampled efficiently on a classical computer and such that $\left \lVert\DD(\true;\Lambda_\N)-\DD(\sigma_{\tranef};\Lambda_{\tilde{\N}}) \right \lVert = \dsepB{\tranef}$, where $\dsepB{\eta}$ is given in \cref{eq:distETAsep} and satisfies $\dsepB{\eta} \leq  \frac{\tran^2 n}{2} + \frac{\tran(1-\tran)}{2}$, and $\tranef= \eta^s$.
\end{corollary}

It is clear from \cref{cor:SIMarbit} that our simulation gives a good approximation whenever there are less than O$(\sqrt{n})$ photons left on average. In terms of $s$, the regime of $\alpha \sqrt{n}$ photons left corresponds to $s = - \log n /(2 \log \eta) + \log \alpha / \log \eta$. Thus, our results show that boson sampling becomes classically simulable for networks of depth greater $C \log n$, where $C$ is some suitably large constant that depends on $\eta$. We should note that a similar scaling could also be drawn from the results described at the end of \cref{sec:resFIXED}, since the regime where only $\alpha \log n$ photons are left on average corresponds to depth $s = - \log n / \log \eta + \log \log n / \log \eta + \log \alpha /\log \eta$.

It is curious to note that our \cref{cor:SIMarbit} generates a tension with the conclusions drawn by the authors of \cite{Clements2016}. There, they argue that balanced networks are better, since unbalanced losses act as non-unitary noise. But here we conclude the opposite: if a network is too balanced, more losses can be uniformly extracted from the network, improving the quality of our simulation of \cref{cor:SIMarbit}. 

We conclude this section by pointing out a few ways that the above results can be strengthened. For the first note that, if photons are only inserted in particular input modes, we can replace the network by another effective network where unused inputs are ignored. For example, consider the standard decomposition of \cref{fig:balanced_unbalanced}(b). One could think that \cref{thm:PULLoutNOISE} would not be relevant there since, in principle, $s=1$. However, if it is known that photons will only enter the first $n$ modes we can restrict definition of $s$ in \cref{thm:PULLoutNOISE} to only take into account paths starting from those inputs, effectively increasing $s$. This can also be relevant if one does not have control over the input, such as in scattershot boson sampling \cite{Scottblog,Bentivegna2015}.

Another way \cref{thm:PULLoutNOISE} can be strengthened is that it can be applied even if losses are not all equal over all beam splitters. As long as every beamsplitter has some nonzero loss probability in each input mode (which may even differ between the two inputs of the same beamsplitter), we can repeat the procedure of \cref{thm:PULLoutNOISE} by pulling out $s$ layers of uniform loss, each with loss probability $1-\eta^\ast$, where $\eta^\ast$ is the highest among all value of $\eta$ for all beamsplitters. Finally, we conjecture that \cref{cor:SIMarbit} holds even if a network is composed not only of two-mode elements, but linear-optical elements of various sizes, and we leave this as a question for future work.

\section{Relation to other works and relevance for boson sampling} \label{sec:bosonsampling}

\subsection{Complexity-theoretic considerations}

All main simulation theorems contained in \crefrange{sec:resFIXED}{sec:nonUNIF} have a common feature, in that they allow us to approximate boson sampling with the desired lossy state $\rho$, by replacing it by another state $\sigma$, up to some total variation distance $\epsilon$. The precision of the simulation, $\epsilon$, is not a controllable parameter, but rather a function of $n$ and possibly other physical parameters such as the transmission probability $\eta$. In suitable regimes (for instance, if $l=\o(n^{1/2})$ in \cref{cor:fixed}), $\epsilon$ decreases as the problem size $n$ increases. Let us investigate the complexity-theoretic consequences of these results.

In the original paper on boson sampling, \cite{Aaronson2013a}, the authors define the following computational problem, which they call Gaussian Permanent Estimation:
\begin{problem}[$|\textrm{GPE}|_{\pm}^2$, \cite{Aaronson2013a}]\label{prob:GPE}
Given as input a matrix $X$ of i.i.d.\ elements taken from the standard (complex) normal distribution, together with bounds $\epsilon', \delta > 0$, estimate $|\mathrm{Per} X|^2$ to within additive error $\pm \epsilon' \cdot n!$ with probability at least $1 - \delta$ over choices of $X$ in poly$(n, 1/\epsilon', 1/\delta)$ time.
\end{problem}
Their main result (Theorem 1.3 of \cite{Aaronson2013a}) then can be stated roughly as follows. Let $U$ be an $m \times m$ Haar-random matrix, where $m = \mathrm{O}(n^2)$, and denote by $\mathcal{D}$ the probability distribution over outcomes generated by a boson sampling computer, i.e.\ with probabilities given by \cref{eq:permanentboson}. Suppose there is a classical algorithm $\mathcal{C}$ that takes as input the unitary matrix $U$ and an error bound $\epsilon > 0$ and outputs a sample from some distribution $\mathcal{D}'$ such that $\left \lVert \mathcal{D}-\mathcal{D}' \right \lVert<\epsilon$, in time poly$(n,1/\epsilon)$. Then it would be possible to solve $|\textrm{GPE}|_{\pm}^2$ in the complexity class $\mathrm{BPP}^{\mathrm{NP}}$ (for definitions and discussions on these complexity classes we direct the reader to \cite{Aaronson2013a} or to the Complexity Zoo \cite{compzoo}).

The authors argue that such a consequence would be very surprising, and so an efficient classical algorithm $\mathcal{C}$ is unlikely, but they also take a step further. By positing two additional complexity conjectures, the Permanent-of-Gaussian conjecture and the Permanent Anti-Concentration conjecture, which together state that estimating $|\textrm{Per}(X)|^2$ for Gaussian $X$ to high precision and with high probability is a $\#$P-hard problem, they show that the existence of an efficient classical algorithm $\mathcal{C}$ would imply a collapse of the polynomial hierarchy. The polynomial hierarchy is a hierarchy of complexity classes which, in the complexity theory community, is widely believed to be infinite, and so this results further strengthens the initial claim that boson sampling cannot be efficiently simulated on a classical computer. 

By inspecting our results, we see that they do not follow the same definition of efficient classical simulation as used in \cite{Aaronson2013a}. There, the classical algorithm should  be able to produce samples from a distribution that is $\epsilon$-close to the ideal one, for \emph{arbitrarily small} $\epsilon$, just by spending poly$(1/\epsilon)$ more computational time. This definition of classical simulation has also recently been discussed in detail in \cite{Pashayan2017}, where it is called an $\epsilon$-simulation.

One might ask whether the authors of \cite{Aaronson2013a} could have used a less stringent definition of simulation, where $\epsilon$ is some \emph{fixed} threshold value rather than a controllable parameter that can be made arbitrarily small. This is used, for example, in \cite{Bremner2016}, where the authors apply a similar reasoning used for boson sampling to the model of commuting quantum circuits, concluding that it would be unlikely for a classical computer to be able to efficiently sample from a distribution that is closer than 1/192 in total variation distance to the ideal one. By adapting the argument of \cite{Aaronson2013a} one can see that, if there existed a classical algorithm that sampled from a distribution $\mathcal{D}'$ such that $\left \lVert \mathcal{D}_U-\mathcal{D}' \right \lVert<\epsilon$ for \emph{fixed} $\epsilon$, it would possible be solve, in $\mathrm{BPP}^{\mathrm{NP}}$, an easier version of $|\textrm{GPE}|_{\pm}^2$ where $\delta$ and $\epsilon'$ are not free parameters, but rather satisfy $\epsilon' \delta = \mathrm{O} (\epsilon)$. That is, there would be a trade-off between the probability that the $\mathrm{BPP}^{\mathrm{NP}}$ algorithm outputs an estimate for $|\mathrm{Per} X|^2$ for Gaussian $X$ and the quality of that estimate. That would be still an unlikely complexity-theoretic outcome, albeit a less surprising than an efficient classical algorithm $\mathcal{C}$ described above.

This discussion does explicitly influence our conclusions, however. Ultimately, our main results are not addressing any complexity conjecture, given that to our knowledge no one has formally conjectured that boson sampling is hard in the limit of large losses or that the probabilities generated by a lossy device are associated with any natural hard problem. At the same time, for any given problem size $n$ our simulations work only up to fixed precision $\epsilon(n)$ as described in the corresponding theorems. Thus our results do not rule out, technically, that lossy boson sampling has interesting computational capabilities, if one defines the corresponding computational \emph{task} in terms of the ``arbitrarily small $\epsilon$'' prescription advocated in \cite{Aaronson2013a,Pashayan2017}. On the other hand our results do place stringent requirements on proposed experimental \emph{demonstrations} of quantum computational supremacy based on lossy boson sampling. The fact that the precision of our simulations improves with the size of the problems means that, in order for a realistic boson sampling machine to outperform (asymptotically) our simulation, the experimental imperfections (most crucially the loss probability per-mode per-component) must decrease sufficiently fast as experiments move to larger problem size. 

\subsection{Relation to other work}

\paragraph*{Connection to the mean field state---}
Consider the closest separable state to $\targ{l}$ from \cref{eq:target}, namely the state $ \sigma_\ast$ in \cref{thm:main1}. This state has appeared previously in the boson sampling literature, where it was called a mean-field state \cite{Tichy2014}. A device that performs a boson sampling experiment using this state as input has been proposed in \cite{Tichy2014} as a semi-classical approximation to a boson sampling device that would be classically simulable, but that would be able to fool methods capable of distinguishing boson sampling distributions from alternative ones \cite{Aaronson2014,Gogolin2013,Carolan2014}. In particular, the mean-field state has the ability to effectively display  more ``bunching'' effects than distinguishable classical particles, which is usually considered a genuinely bosonic signature. It is interesting then that this state has also appeared in our main results as the closest symmetric separable state to the lossy bosonic state, as answer to a question that is seemingly unrelated to the  matter of verification of boson sampling devices.

\paragraph*{Quantum de Finetti theorem---} It is natural to investigate the connection of our results to quantum de Finetti theorem \cite{Doherty2002,Christandl2007,Harrow2013}. In particular, the quantum de Finetti theorem for bosonic systems is concerned with the best  approximation (in trace distance) of the state $\tr_{n-l}\L \Psi \R$ by separable product states $\sigma \in \Sep\L  \symm{l}{m} \R$, where $\Psi$ is an $n$-particle $m$-mode pure bosonic state.  This is exactly the setting that we consider in the case of the fixed-loss model considered in \cref{thm:main1}. Applying directly the bosonic de Finetti \cite{Christandl2007} theorem yields
\begin{equation}\label{eq:DEfinetti}
\min_{\sigma\in\Sep\L \symm{l}{n} \R} \dtr \L \tr_{n-l}\L \Psi \R, \sigma \R \leq \frac{2lm}{n}\ .
\end{equation}
In the context of boson sampling we always have $m\geq n$ and thus the above bound  becomes useless. The much stronger result given in \cref{thm:main1} was possible only because it concerns separable approximation to the particular lossy state $\tr_{n-l}\L \true \R$, where $\true$ is defined in \cref{eq:defINPUTstate}.

\paragraph*{Entanglement of symmetric states---} Our results can be also linked to studies of entanglement of symmetric states, which recently received a lot attention \cite{Ghune2009, Oszmaniec2014,YU2016,Quesada2017,TuraSep2017}. In, particular, due to the difficulty of the general separability problem, many works  consider the problem of deciding whether states diagonal in the Dicke basis are entangled \cite{YU2016,Quesada2017,TuraSep2017}. However, little attention has been given to computing entanglement measures or entanglement indicators for such states. The trace distance  to the set of symmetric separable states \cref{eq:distanceKN} can be used as an indicator of entanglement and as a lower bound the geometric measure of entanglement \cite{Streltsov2010}. Thus, our results can be also understood as  quantitative description for a class of lossy bosonic states. We believe that tools introduced in this work will be useful in quantitative studies of entanglement of symmetric states diagonal in the Dicke basis. 

\section{Proofs of the main results}\label{sec:mainProofs}

In this section we prove the main results of the paper: Theorems \ref{thm:main1} and \ref{thm:PULLoutNOISE}. The proofs given below are essentially self-contained, however along the way we use some auxiliary technical statements whose justification we defer to the Appendix.

\subsection{Best separable approximation with respect to trace distance}\label{sec:proofDIST}

In order to find the best separable approximation to $\targ{l}$ (in trace distance) we prove a few intermediate results. First, in Lemma \ref{lem:stabOPT} we show that the symmetries of the state $\targ{l}$ allow us to limit our attention to separable approximations having a particularly simple ``twirled'' structure, which we describe explicitly in Corollary \ref{cor:doubleAV} and Lemma \ref{lem:twirledSTATES}. Then, in Lemma \ref{lem:optPURE}  we show that the closest symmetric separable state can be chosen  as the twirling of a \emph{pure product state}. Finally, in Lemma \ref{lem:optialisationPROD} we carry out the optimization over pure product seed states which directly yields Theorem \ref{thm:main1}. Proofs of some of the intermediate results are technical and can be left out during the first reading. 

In what follows we denote by $\LO{n}$ the group of $n\times n$ unitary matrices encoding linear-optical transformations on $n$ modes.

\begin{lemma}[Stabilizing subgroups and the structure of the best separable approximation]\label{lem:stabOPT}
 Let  $\sigma_\ast \in \Sep\left(\symm{l}{n}\right)$ be a symmetric $l$-particle separable state supported on $n$ modes satisfying 
\begin{equation}
\dtr\left(\targ{l},\sigma_\ast\right) = \dsep{l}= \min_{\sigma\in\Sep\left(\symm{l}{n}\right) } \dtr\left(\targ{l},\sigma\right)\ .
\end{equation}
Let $K$ be some subgroup of  $\LO{n}$ consisting of unitaries stabilizing the state $\targ{l}$\footnote{In other words, $K$ is a subgroup of the stabilizer group $\mathrm{Stab}\left(\targ{l}\right)$ of $\targ{l}$ in $\LO{n}$.}, i.e.
\begin{equation}\label{eq:stabilisation}
V \in K\ \Rightarrow\ V^{\otimes l} \targ{l} (V^{\otimes l})^\dagger\ =\targ{l}\ .
\end{equation} 
Then the state 
\begin{equation}\label{eq:twirledState}
\tilde{\sigma}_\ast \eqdef \Lambda_K(\sigma_\ast) = \int_{K} dV\ V^{\otimes l} \sigma_\ast (V^{\otimes l})^\dagger  
\end{equation}
satisfies $\dtr\left(\targ{l},\tilde{\sigma}_\ast \right)=\dtr\left(\targ{l},\sigma_\ast \right)=\dsep{l}$. The integration in \cref{eq:twirledState} is done with respect to the Haar measure on $K$.
\end{lemma}

\begin{proof}
First note that $\Lambda_K$ satisfies $\Lambda_K(\targ{l})=\targ{l}$. Moreover, $\Lambda_K$ is a CPTP map and so, by the data-processing inequality \cite{Nielsen2010}	we obtain
\begin{equation}\label{eq:INEQUALITIEScp}
\dtr\left(\targ{l},\sigma_\ast\right)\geq \dtr\left(\Lambda_K\left(\targ{l}\right),\Lambda_K\left(\sigma_\ast\right)\right)= \dtr\left(\targ{l},\tilde{\sigma}_\ast\right)\ .
\end{equation}
Now note that transformations of the form $V^{\otimes l}$ preserve separable symmetric states. This means that $\Lambda_K\left(\Sep\left(\symm{l}{n}\right)\right)\subset \Sep\left(\symm{l}{n}\right)$, and hence $\tilde{\sigma}_\ast \in \Sep\left(\symm{l}{n}\right)$. From the definition of $\sigma_\ast$ we thus get $ \dtr\left(\targ{l},\tilde{\sigma}_\ast\right) \geq  \dtr\left(\targ{l},\sigma_\ast\right)$ which together with Eq. \eqref{eq:INEQUALITIEScp} concludes the proof. 
\end{proof}

\begin{remark}
Inspection of the above proof shows that an analogous result holds if $\targ{l}$ is replaced by arbitrary $l$-particle symmetric $n$-mode state $\rho\in\D\left( \symm{l}{n} \right)$ and $K$ is replaced by any subgroup of the stabilizer of $\rho$ in $\LO{n}$. 
\end{remark}

The following subgroups of the stabilizer of $\targ{l}$ will prove especially useful in the computation of $\dsep{l}$ and finding $\sigma_\ast$.
\begin{align}\label{eq:stabSUBROUPS}
K_c & \eqdef \SET{  V_{\varphi_1,\ldots,\varphi_n} }{ \varphi_i \in[0,2\pi) }\ , \\
K_d & \eqdef \SET{V_\pi}{V_\pi \ket{i} \eqdef \ket{\pi(i)},\ \pi\in\mathcal{P}(n)}\ \nonumber,
\end{align}
where $V_{\varphi_1,\ldots,\varphi_n} = \exp\left(\ii i \sum_{i=1}^n \varphi_i \ketbra{i}{i}\right)$ and $\mathcal{P}(n)$ is the permutation group of the set $[n]$. Intuitively speaking, $K_c$ represents linear-optical transformations that are diagonal in the Fock basis, while $K_d$ represents permutations of the $n$ modes. Using group-theoretic techniques similar to the ones used in \cite{Bouland2014,Sawicki2016} we can show that, by composing elements from $K_c$ and $K_d$, one can obtain the entire stabilizer of $\targ{l}$ in $\LO{n}$ (however, in what follows we do not need this result).  
The twirling operations for $K_c$ and $K_d$ read
\begin{align}\label{eq:explicitINT}
\Lambda_{K_c} (\rho) & = \frac{1}{(2\pi)^n}  \int_0^{2\pi}  d \varphi_1 \ldots  \int_0^{2\pi}d  \varphi_n   \left(V_{\varphi_1,\ldots,\varphi_n}\right)^{\otimes l}  \rho \left(V_{\varphi_1,\ldots,\varphi_n}\dg \right) ^{\otimes l}\ , \\
\Lambda_{K_d} (\rho) & = \frac{1}{l!} \sum_{\pi\in \mathcal{P}(n)} V_{\pi}^{\otimes l}  \rho  \left(V_{\pi}^\dagger\right)^{ \otimes l} \nonumber  \ .
\end{align}

\begin{corollary}\label{cor:doubleAV}
By application of Lemma \ref{lem:stabOPT} first to $K_c$ and then to $K_d$, we obtain that the closest separable symmetric state to $\targ{l}$ (in trace distance) can be chosen as $\sigma_\ast = \Lambda_{K_d} \circ  \Lambda_{K_c} \left(\sigma\right)$, for $\sigma\in \Sep\left(\symm{l}{n}\right)$. 
\end{corollary}
We now describe the action of the operations of Eq.\eqref{eq:explicitINT} on product symmetric states $\ketbra{\phi}{\phi}^{\otimes l}$. To this end we need to introduce some auxiliary notation. We define  \emph{the type} of the $n$-mode occupation number string $\n=(n_1,\ldots,n_n)$ as the string  $\t(\n)=(\tau_1,\ldots,\tau_n)$ consisting of non-increasingly ordered components of $\n$. Two Fock states $\ket{\n}$ and $\ket{\n'}$ can be mapped to each other by a linear-optical transformations from $V_\pi \in K_d$ if and only if $\t(\n)=\t(\n')$. Just like for occupation-number strings, we define the length of the type $\t$ as $|\t|\eqdef \sum_{i=1}^n \tau_i$. We use the following notation for multinomial coefficients: $\binom{|\t|}{\t}\eqdef\frac{|\t|!}{\tau_{1}! \cdots \tau_{n}!}$. For a given type $\t$ we can define the corresponding mixed state 
\begin{equation}\label{eq:permutAVfock}
\rho_\tau \eqdef \frac{1}{\NN(\t)}\sum_{\n :\t(\n)=\t} \ketbra{\n}{\n}\ , 
\end{equation}
where $\NN(\t)$ is a normalization constant equal to the cardinality of the set $\SET{\n}{\t(\n)=\t}$. In particular, for $\t_0 =(\overbrace{1,1,\ldots,1}^{l},0,\ldots,0)$  we have $\NN(\t_0)=\binom{n}{l}$ and  $\rho_{\t_0}=\targ{l}$. For s vector $\alpha=(\alpha_1, \ldots, \alpha_n)$ and an occupation number string $\n$, define $|\alpha|^{2\n}\eqdef \prod_{i=1}^n |\alpha_i|^{2n_i}$. Finally, to state our results we also need the following polynomial in $\alpha$ depending on the type $\t$: 
\begin{equation}\label{eq:symPOLY}
m_{\t}(\alpha)\eqdef \sum_{\n: \t(\n)=\t} |\alpha|^{2 \n}\ .
\end{equation}
  
\begin{lemma}[Structure of the twirled product states]\label{lem:twirledSTATES}
Let $\ket{\phi}=\sum_{i=1}^n \alpha_i \ket{i}$ be a pure state. Let $K_c$, $K_d$ be as in Eq.\eqref{eq:stabSUBROUPS}. Let $\Lambda_K$  be a twirling map defined by Eq.\eqref{eq:twirledState}. Then we have the following result
\begin{equation}\label{eq:Avgtotal}
\Lambda_{K_d} \circ \Lambda_{K_c} \left(\ketbra{\phi}{\phi}^{\otimes l}\right) = \sum_{\t: |\t|=l} m_{\t}(\alpha) \binom{l}{\t} \rho_{\t}
\end{equation}
with the notation defined below Corollary \ref{cor:doubleAV}.
\end{lemma}
The proof of the above Lemma is cumbersome, so we present it in \cref{app:technicRES}. The above characterization of the twirled product states allows us to obtain the following important result.

\begin{lemma}[A seed state can be chosen to be a pure product state] \label{lem:optPURE} 
Let $\sigma =\sum_k p_k  \ketbra{\phi_k}{\phi_k}^{\otimes l}$ be an arbitrary symmetric separable state on $\symm{l}{n}$. Let $K_c$, $K_d$ be as in Eq.\eqref{eq:stabSUBROUPS} and let $\Lambda_K$  be the map defined by Eq.\eqref{eq:twirledState}. Then the trace distance between $\Lambda_{K_d} \circ \Lambda_{K_c}(\sigma)$ and $\targ{l}$ is given by
\begin{equation}\label{eq:distanceSEPexpl}
\dtr\left(\targ{l}, \Lambda_{K_d} \circ \Lambda_{K_c}(\sigma)\right) = 1 -\sum_k p_k q \L \ketbra{\phi_k}{\phi_k} \R\ .
\end{equation}
where for $\ket{\phi}=\sum_{i=1}^n \alpha_i \ket{i}$ we have defined $q(\kb{\phi}{\phi})\eqdef l! m_{\t_0}(\alpha)$. Therefore, the closest separable symmetric state to $\targ{l}$ (in trace distance) can be chosen as $\sigma_\ast = \Lambda_{K_d} \circ \Lambda_{K_c}\left(\ketbra{\phi}{\phi}^{\otimes l}\right)$, for a suitable product state $\ketbra{\phi}{\phi}^{\otimes l}$.  
\end{lemma}

\begin{proof}
Let $\ket{\phi}=\sum_{i=1}^n \alpha_i \ket{i}$. We note that Eq.\eqref{eq:Avgtotal} implies the following convex decomposition 
\begin{equation}\label{eq:twirledSTRUCTURE}
\Lambda_{K_d} \circ \Lambda_{K_c} \left(\ketbra{\phi}{\phi}^{\otimes l}\right) = q \targ{l} +(1-q)\targ{l}^\perp\ ,
\end{equation} 
where $\targ{l}^\perp$ is supported on the subspace orthogonal to the support of $\targ{l}$ and 
\begin{equation}\label{eq:formOFq}
q:= q(\ketbra{\phi}{\phi})=l! m_{\t_0}(\alpha)\ ,
\end{equation}
Now, for an arbitrary separable symmetric state $\sigma =\sum_k p_k  \ketbra{\phi_k}{\phi_k}^{\otimes l}$  we obtain 
\begin{equation}
\Lambda_{K_d} \circ \Lambda_{K_c} \left(\sigma\right) = q' \targ{l} +(1-q')\targ{l}'\ ,
\end{equation}
where $q'= \sum_k p_k q_k $, $q_k= q(\ketbra{\phi_k}{\phi_k})$ and again $\targ{l}'$ is supported on the subspace orthogonal to the support of $\targ{l}$. Now, using the explicit  formula for the trace distance given in \cref{eq:TRACEdist}, we obtain directly \cref{eq:distanceSEPexpl}.
\end{proof}

We now state the result that settles the optimization problem for seed symmetric product states. This together with Lemma \ref{lem:optPURE} and Corollary \ref{cor:doubleAV} concludes the proof Theorem \ref{thm:main1}.

\begin{lemma}[Optimization over seed symmetric product states]\label{lem:optialisationPROD} 
Let $\ket{\phi}=\sum_{i=1}^n \alpha_i \ket{i}$ be a pure state on $\C^n$. The maximal value of the function 
\begin{equation}\label{eq:funcSINGprod}
q(\ketbra{\phi}{\phi}) = l! m_{\t_0}(\alpha)\ .
\end{equation}
is attained for the ``maximally coherent'' state $\ket{\phi_0}=\frac{1}{\sqrt{n}}\sum_{i=1}^n \ket{i}$. Consequently, 
\begin{equation}\label{eq:finalOPTIMIZATION}
\min_{\ket{\phi}\in\C^n, \bra{\phi}\phi\rangle=1} \dtr\left(\targ{l},\Lambda_{K_d} \circ \Lambda_{K_c}\left(\ketbra{\phi_0}{\phi_0}^{\otimes l}\right)\right)\ =1 - \frac{n!}{n^{l} (n-l)!}\ .
\end{equation}  
\end{lemma}
\begin{proof}
Note that $q(\ketbra{\phi}{\phi})$ is proportional to the elementary symmetric polynomial of degree $l$ in variables $\textbf{p}=(p_1,\ldots, p_n)$, where  $p_i=|\alpha_i|^2$. Elementary symmetric polynomials are Schur-concave \cite{SchurConcave1991} i.e.\ satisfy $q(\textbf{p})\leq q(S \textbf{p})$, for any doubly-stochastic $n\times n$ matrix  $S$. Setting $S=\frac{1}{n}\Id$ we obtain that for any pure state $\ket{\phi}$  we have the inequality $q(\kb{\phi_0}{\phi_0}) \geq  q(\kb{\phi}{\phi})$. Finally, Eq.\eqref{eq:finalOPTIMIZATION} follows by inserting $|\alpha_i|^2=1/n$ into Eq.\eqref{eq:funcSINGprod} and using Eq.\eqref{eq:distanceSEPexpl}.
\end{proof}
\begin{remark}
For the optimal pure state $\ket{\psi_0}$ from \eqref{eq:finalOPTIMIZATION} we have
\begin{equation}
\Lambda_{K_d} \circ \Lambda_{K_c}\left(\ketbra{\phi_0}{\phi_0}^{\otimes l}\right)= \Lambda_{K_c}\left(\ketbra{\phi_0}{\phi_0}^{\otimes l}\right)\ ,
\end{equation}
which gives exactly the form of the optimal separable state given in Theorem \ref{thm:main1}.

\end{remark}

\subsection{Extraction of uniform losses from a lossy network}\label{sec:proofEXTR}

Let us now prove \cref{thm:PULLoutNOISE}. Our proof is constructive, in the sense that it follows directly from an algorithm to obtain $\tilde{\N}$ from $\N$. A central ingredient of our proof is the equivalence shown in \cref{fig:buildingblock}, namely that two equal losses at the output modes of a beamsplitter can be commuted to left to act on its input modes.   Throughout this section we will say that losses can be ``pulled out'' from a network if, by repeated application of that identity to individual beamsplitters, they can moved all the way to the input of the circuit. Whenever we pull out one loss element for \emph{each} input in the network, we say we pulled out one uniform layer of losses. The network $\tilde{\N}$ will be obtained by pulling out $s$ layers of losses from $\N$, one at a time. 

Recall that, by assumption, $\N$ is such that every path from the input to the output has to traverse at least $s$ beamsplitters. Let us temporarily call the \emph{length} of a path within the network the number of loss elements that a photon would traverse following that path, irrespective of the actual number of beamsplitters. This distinction is initially irrelevant since every beamsplitter has exactly one loss element in each of its input paths, but it becomes important in intermediate stages after we start to commute loss elements around. Let us call a path that connects an input to an output an \IO-path. By assumption, the shortest \IO-path in $\mathcal{N}$ has length $s$. 

We now describe the procedure step by step.

\begin{quote}
\step 0: Label every beamsplitter in $\N$ by the shortest path between it and \emph{any} input. There might be exponentially within the network, but this procedure can be done efficiently if it is performed in a breadth-first manner. First, label all beamsplitters immediately connected to the inputs as 1. Then, move to all beamsplitters that are connected to those and label them. Repeat this procedure until all beamsplitters have been labeled. For any beamsplitter, its label is at most 1 greater than the smaller of the labels of the two that immediately precede it. This procedure can be done in time linear in the number of beamsplitters, and is exemplified in \cref{fig:numberednetwork}.
\end{quote}

\begin{figure}
    \centering
    \includegraphics[width=0.5\textwidth]{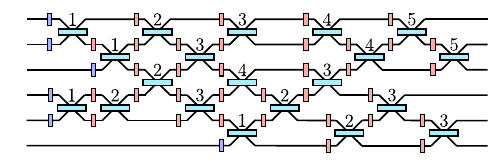}
    \caption{Example of an unbalanced network on which we can apply the procedure outlined in \cref{sec:proofEXTR}. Each beamsplitter is labeled by the smallest number of loss elements between it and any input. The loss elements marked in blue are those described in \step 1. This network only has nearest-neighbor beamsplitters as it is easier to visualize, but at no point in the proof this is required.}
    \label{fig:numberednetwork}
\end{figure}

\begin{quote}
\textsc{Step} 1: Consider the set of loss elements that lie between a beamsplitter labeled 1 and an input (e.g.\ the highlighted ones in \cref{fig:numberednetwork}). Pull these losses out to the input as an uniform layer, and remove them from the network. This reduces the shortest \IO-path by 1. Recompute the labels of all beamsplitters accordingly.
\end{quote}

We now apply mathematical induction. Suppose that, on \step $j-1$, it was possible to pull out losses from every beamsplitter labeled 1 at that stage, and that these losses were pulled out as a uniform layer (which we then removed from the network, relabeling all beamsplitters in the process). Furthermore, suppose that, by this procedure, the length of the shortest \IO-path was reduced by 1. Let $G_{j-1}$ be the set of beamsplitters that were relabeled from 1 to 0 in that step. Then apply the following.

\begin{figure}
    \centering
    \includegraphics[width=0.7\textwidth]{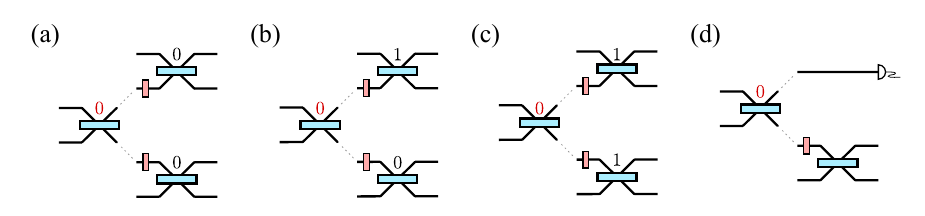}
    \caption{The four possible configurations of a beamsplitter from set $G_{j-1}$, as described in \step $j$ in the main text.}
    \label{fig:commutingcases}
\end{figure}

\begin{quote}
\step $j$: If the shortest \IO-path has length 0, terminate the procedure, as $s$ layers of losses have already been pulled out. Otherwise consider the beamsplitters in $G_{j-1}$. Each of them must be in one of the configurations of \cref{fig:commutingcases}(a-c), since  \cref{fig:commutingcases}(d) would imply that the shortest \IO-path has length 0. For each, if there are two loss elements at their output, commute them through to the two inputs. We now claim that this must be possible for every beamsplitter in $G_{j-1}$ in every configuration of \cref{fig:commutingcases}(a-c). To see this, note that we only ever commute losses through beamsplitters labeled 0. But at the beginning of \step $j-1$, every beamsplitter in $G_{j-1}$ was labeled 1, and they were only relabeled from 1 to 0 because losses that \emph{preceded} them were pulled out. This means that the procedure has not yet touched the losses that \emph{follow} them along both output modes.

After losses have been commuted through all beamsplitters in $G_{j-1}$, note that this would reverse their label from 0 back to 1. But this places them in the configuration they were at the beginning of \step $j-1$, relative to the part of the network that precedes them. This means that we can pull out a uniform layer of losses, by the inductive hypothesis. After this uniform layer is pulled out and removed from the circuit, relabel all beamsplitter accordingly. Note that a new set of beamsplitters have been relabeled from 1 to 0, all which lie in the immediate future of the beamsplitters in $G_{j-1}$. This set is $G_j$, to be used in the next step. By analyzing the configurations of \cref{fig:commutingcases}(a-c), it is also easy to see that this step has reduced the length of the shortest \IO-path by 1.
\end{quote}

Thus we have shown that the inductive hypothesis also holds for $j$, and so it holds all the way up to \step $s$. This procedure is clearly efficient and pulls out $s$ uniform layers of losses. The network that remains afterwards is $\tilde{\N}$, as in \cref{thm:PULLoutNOISE}, and is a valid linear-optical network with an efficient description, which completes the proof.

For the particular example of \cref{fig:numberednetwork}, the network $\tilde{\N}$ is shown in \cref{fig:Endexample}. In \cref{app:example} we provide a detailed walkthrough of the steps that lead to $\tilde{\N}$. In \cref{fig:Endexample} we notice that some segments in the network seem to accumulate losses in the process of pulling out the uniform layers. This is a curious feature, but does not affect the proof that the procedure works in any way.

\begin{figure}
    \centering
    \includegraphics[width=0.5\textwidth]{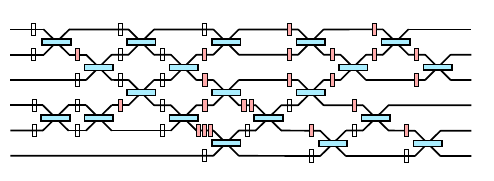}
    \caption{The resulting network after the procedure of \cref{thm:PULLoutNOISE} is applied to the network of \cref{fig:numberednetwork}. Empty boxes correspond to the loss elements that have been removed from the network. Note that the procedure of pulling out the losses may cause the losses that remain to pile up in some locations of the network.}
    \label{fig:Endexample}
\end{figure}



\section{Discussion and open problems} \label{sec:conclusions}

In this paper we proposed a method for efficient classical simulation of linear-optical experiments, most notably boson sampling, in the high-loss regime. We showed that, in an experiment containing initially $n$ photons, if all but $l = \mathrm{o}\left(\sqrt{n}\right)$ of them are lost we can approximate the resulting state by a symmetric separable state of $l$ particles. This state has trace distance $\mathrm{o}\left(l^2/n\right)$ to the state of interest, and so a simulation based on it becomes better and better as the size of the system increases. The fact that this state is separable (in the first quantization representation of linear optics) and has a simple description implies a straightforward efficient classical simulation of any linear-optical protocol that uses it as input. We also showed that in the more physically-realistic model of losses, where each photon has some probability $(1-\tran)$ of being lost (rather than there being a definite number of photons lost) our classical simulation is efficient whenever  $\tran=\o(n^{-1/2})$. Finally, we showed how to perform this simulation in a non-uniform model of loss, i.e.\ when we have different paths in the linear-optical circuit incurring different losses. We showed that, by assuming that \emph{every} beamsplitter in a circuit induces some loss $(1- \tran)$ and that the shortest path that a photon can take through the circuit passes through $s$ beamsplitters, our simulation is efficient whenever $\tranef=\sqrt{m}\tran^s = \o(n^{-1/2})$, or alternatively for networks of minimal depth $s = \mathrm{O}(\log n)$ for some suitably large constant (depending on $\tran$).

Our results imply that, in order to demonstrate some computational advantage within the boson sampling paradigm, experiments will likely have to shift to different architectures. One possibility is to improve the original results of \cite{Aaronson2013a} to show that boson sampling on linear-optical circuits of {\em logarithmic depth} is already hard, for example, because such a circuit can already exhibit sufficient Haar-randomness for the purposes of the proofs of \cite{Aaronson2013a}. Our results place limitations on the capabilities of log-depth circuits, but do not rule out the possibility that sufficiently shallow circuits would lead to hard boson sampling instances--although it is possible to show, using tensor-network methods, that these circuits would require beamsplitters that act between arbitrarily distant modes \cite{Jozsa2006,GarciaPatron2017}. This at least rules out boson sampling with log-depth networks if one is restricted to the two-dimensional integrated architectures of \cite{Broome2013,Crespi2013b,Spring2013,Tillmann2013,Carolan2015}.

We leave several open questions for future work, including:

\begin{itemize}
  \item[(i)] In \cite{Aaronson2016} the authors proved that boson sampling remains hard when O$(1)$ photons are lost, and here we proved that it becomes classically simulable when only O$(\sqrt{n})$ photons remain. The main problem we leave open is at which point between O$(\sqrt{n})$ and $n-$O$(1)$ losses the transition between these two complexity regimes occurs.
  \item[(ii)] Can we get better approximations if we allow for particle-separable states that are not symmetric? Alternatively, can we prove that the closest state to our desired state $\targ{l}$ must be symmetric?
  \item[(iii)] Our results bound the trace distance between the mean-field state and the lossy input state, as in \cref{thm:main1}. This also acts as a bound on the total variation distance between the corresponding distributions when both states are input into \emph{arbitrary} linear-optical transformations. What happens with the total variation distance after a \emph{Haar-random} unitary $U$? Would the distance between the output states after a random $U$ be considerably smaller than the worst-case bound we obtain? Would this change the scaling of our main result and our conclusions?
  \item[(iv)] We showed that, for a network where the shortest path connecting any input to any output is $s$, one can remove $s$ layers of uniform losses. But it is possible to decompose an arbitrary linear-optical unitary such that $s=1$, e.g.\ the standard decomposition of Reck \emph{et al} in \cref{fig:balanced_unbalanced}(b)\cite{Reck1994}. In \cite{Clements2016}, the authors propose a balanced decomposition, as in \cref{fig:balanced_unbalanced}(a), since there losses behave uniformly and induce less non-unitary noise. What if one instead tried to engineer the unbalanced noise in a helpful way? That is, can we argue that boson sampling must be hard even under the non-unitary noise coming from the unbalanced losses in \cref{fig:balanced_unbalanced}(b)? If so, that might provide a way to engineer long networks with small $s$ tailored to circumvent our argument.
  \item[(v)] As discussed in \cref{sec:bosonsampling}, our simulations get better as the size of the problem increases, but seemingly cannot be made arbitrarily precise at the cost of more computational time. This means we do not have an $\epsilon$-simulation as per the definition of \cite{Aaronson2013a,Pashayan2017}, although our notion of simulation agrees with the one used e.g.\ in \cite{Bremner2016}. Can our simulation be improved into the more stringent form of an $\epsilon$-simulation? Alternatively, do these different notions of simulation actually lead to fundamentally different complexity regimes?
  \item[(vi)] Can our simulation be adapted for other variants of boson sampling \cite{Lund2014,Scottblog,Bentivegna2015}, or for other candidates of quantum computational supremacy experiments \cite{Bremner2011,Bremner2016,Morimae2014,Aaronson2016b,Boixo2016}?
  \item[(vii)] Is it possible to derive a version of the de Finetti theorem that, for symmetric states diagonal in Dicke basis, yields a bound stronger than the one in \cref{eq:DEfinetti}?
 
\end{itemize}

\emph{Note:} 
Upon completion of this work, we became aware of a simultaneous work \cite{GarciaPatron2017} that seems to arrive at similar results for classical simulation of boson sampling in the case of uniform losses.

\begin{acknowledgments}
The authors would like to thank Scott Aaronson, Alexander Strelcov, Micha\l\ Horodecki, Hakop Pashayan, Rafa\l\ Demkowicz-Dobrzanski, Zbigniew Pucha\l a,  Martin Kleisch, Ashley Montanaro, and others, for helpful discussions.

M.~O. acknowledges the support of the Homing programme of the Foundation for Polish Science co-financed by the European Union under the European Regional Development Fund. D.~B.~acknowledges financial support by CAPES – Brazilian Federal Agency for Support and Evaluation of Graduate Education within the Ministry of Education of Brazil.
\end{acknowledgments}

\bibliography{bosonsamplingrefs}

\begin{thebibliography}{75}%
\makeatletter
\providecommand \@ifxundefined [1]{%
 \@ifx{#1\undefined}
}%
\providecommand \@ifnum [1]{%
 \ifnum #1\expandafter \@firstoftwo
 \else \expandafter \@secondoftwo
 \fi
}%
\providecommand \@ifx [1]{%
 \ifx #1\expandafter \@firstoftwo
 \else \expandafter \@secondoftwo
 \fi
}%
\providecommand \natexlab [1]{#1}%
\providecommand \enquote  [1]{``#1''}%
\providecommand \bibnamefont  [1]{#1}%
\providecommand \bibfnamefont [1]{#1}%
\providecommand \citenamefont [1]{#1}%
\providecommand \href@noop [0]{\@secondoftwo}%
\providecommand \href [0]{\begingroup \@sanitize@url \@href}%
\providecommand \@href[1]{\@@startlink{#1}\@@href}%
\providecommand \@@href[1]{\endgroup#1\@@endlink}%
\providecommand \@sanitize@url [0]{\catcode `\\12\catcode `\$12\catcode
  `\&12\catcode `\#12\catcode `\^12\catcode `\_12\catcode `\%12\relax}%
\providecommand \@@startlink[1]{}%
\providecommand \@@endlink[0]{}%
\providecommand \url  [0]{\begingroup\@sanitize@url \@url }%
\providecommand \@url [1]{\endgroup\@href {#1}{\urlprefix }}%
\providecommand \urlprefix  [0]{URL }%
\providecommand \Eprint [0]{\href }%
\providecommand \doibase [0]{http://dx.doi.org/}%
\providecommand \selectlanguage [0]{\@gobble}%
\providecommand \bibinfo  [0]{\@secondoftwo}%
\providecommand \bibfield  [0]{\@secondoftwo}%
\providecommand \translation [1]{[#1]}%
\providecommand \BibitemOpen [0]{}%
\providecommand \bibitemStop [0]{}%
\providecommand \bibitemNoStop [0]{.\EOS\space}%
\providecommand \EOS [0]{\spacefactor3000\relax}%
\providecommand \BibitemShut  [1]{\csname bibitem#1\endcsname}%
\let\auto@bib@innerbib\@empty
\bibitem [{\citenamefont {Lund}\ \emph {et~al.}(2017)\citenamefont {Lund},
  \citenamefont {Bremner},\ and\ \citenamefont {Ralph}}]{Lund2017}%
  \BibitemOpen
  \bibfield  {author} {\bibinfo {author} {\bibfnamefont {A.}~\bibnamefont
  {Lund}}, \bibinfo {author} {\bibfnamefont {M.}~\bibnamefont {Bremner}}, \
  and\ \bibinfo {author} {\bibfnamefont {T.}~\bibnamefont {Ralph}},\ }\href
  {\doibase 10.1038/s41534-017-0018-2} {\bibfield  {journal} {\bibinfo
  {journal} {npj Quantum Information}\ }\textbf {\bibinfo {volume} {3}},\
  \bibinfo {pages} {15} (\bibinfo {year} {2017})}\BibitemShut {NoStop}%
\bibitem [{\citenamefont {Harrow}\ and\ \citenamefont
  {Montanaro}(2017)}]{Harrow2017}%
  \BibitemOpen
  \bibfield  {author} {\bibinfo {author} {\bibfnamefont {A.}~\bibnamefont
  {Harrow}}\ and\ \bibinfo {author} {\bibfnamefont {A.}~\bibnamefont
  {Montanaro}},\ }\href {\doibase 10.1038/nature23458} {\bibfield  {journal}
  {\bibinfo  {journal} {Nature}\ }\textbf {\bibinfo {volume} {549}},\ \bibinfo
  {pages} {203} (\bibinfo {year} {2017})}\BibitemShut {NoStop}%
\bibitem [{\citenamefont {Aaronson}\ and\ \citenamefont
  {Arkhipov}(2013)}]{Aaronson2013a}%
  \BibitemOpen
  \bibfield  {author} {\bibinfo {author} {\bibfnamefont {S.}~\bibnamefont
  {Aaronson}}\ and\ \bibinfo {author} {\bibfnamefont {A.}~\bibnamefont
  {Arkhipov}},\ }\href {\doibase 10.4086/toc.2013.v009a004} {\bibfield
  {journal} {\bibinfo  {journal} {Theory of Computing}\ }\textbf {\bibinfo
  {volume} {4}},\ \bibinfo {pages} {143} (\bibinfo {year} {2013})}\BibitemShut
  {NoStop}%
\bibitem [{\citenamefont {Bremner}\ \emph {et~al.}(2011)\citenamefont
  {Bremner}, \citenamefont {Jozsa},\ and\ \citenamefont
  {Shepherd}}]{Bremner2011}%
  \BibitemOpen
  \bibfield  {author} {\bibinfo {author} {\bibfnamefont {M.~J.}\ \bibnamefont
  {Bremner}}, \bibinfo {author} {\bibfnamefont {R.}~\bibnamefont {Jozsa}}, \
  and\ \bibinfo {author} {\bibfnamefont {D.~J.}\ \bibnamefont {Shepherd}},\
  }\href {\doibase 10.1098/rspa.2010.0301} {\bibfield  {journal} {\bibinfo
  {journal} {Proc. R. Soc. A}\ }\textbf {\bibinfo {volume} {467}},\ \bibinfo
  {pages} {459} (\bibinfo {year} {2011})}\BibitemShut {NoStop}%
\bibitem [{\citenamefont {Knill}\ and\ \citenamefont
  {Laflamme}(1998)}]{Knill1998}%
  \BibitemOpen
  \bibfield  {author} {\bibinfo {author} {\bibfnamefont {E.}~\bibnamefont
  {Knill}}\ and\ \bibinfo {author} {\bibfnamefont {R.}~\bibnamefont
  {Laflamme}},\ }\href {https://link.aps.org/doi/10.1103/PhysRevLett.81.5672}
  {\bibfield  {journal} {\bibinfo  {journal} {Phys. Rev. Lett.}\ }\textbf
  {\bibinfo {volume} {81}},\ \bibinfo {pages} {5672} (\bibinfo {year}
  {1998})}\BibitemShut {NoStop}%
\bibitem [{\citenamefont {Morimae}\ \emph {et~al.}(2014)\citenamefont
  {Morimae}, \citenamefont {Fujii},\ and\ \citenamefont
  {Fitzsimons}}]{Morimae2014}%
  \BibitemOpen
  \bibfield  {author} {\bibinfo {author} {\bibfnamefont {T.}~\bibnamefont
  {Morimae}}, \bibinfo {author} {\bibfnamefont {K.}~\bibnamefont {Fujii}}, \
  and\ \bibinfo {author} {\bibfnamefont {J.~F.}\ \bibnamefont {Fitzsimons}},\
  }\href {\doibase 10.1103/PhysRevLett.112.130502} {\bibfield  {journal}
  {\bibinfo  {journal} {Phys. Rev. Lett.}\ }\textbf {\bibinfo {volume} {112}},\
  \bibinfo {pages} {130502} (\bibinfo {year} {2014})}\BibitemShut {NoStop}%
\bibitem [{\citenamefont {{Aaronson}}\ and\ \citenamefont
  {{Chen}}(2016)}]{Aaronson2016b}%
  \BibitemOpen
  \bibfield  {author} {\bibinfo {author} {\bibfnamefont {S.}~\bibnamefont
  {{Aaronson}}}\ and\ \bibinfo {author} {\bibfnamefont {L.}~\bibnamefont
  {{Chen}}},\ }\href@noop {} {\bibfield  {journal} {\bibinfo  {journal} {arXiv
  e-print}\ }\textbf {\bibinfo {volume} {{\!\!}}} (\bibinfo {year} {2016})},\
  \Eprint {http://arxiv.org/abs/1612.05903} {arXiv:1612.05903 [quant-ph]}
  \BibitemShut {NoStop}%
\bibitem [{\citenamefont {{Boixo}}\ \emph {et~al.}(2016)\citenamefont
  {{Boixo}}, \citenamefont {{Isakov}}, \citenamefont {{Smelyanskiy}},
  \citenamefont {{Babbush}}, \citenamefont {{Ding}}, \citenamefont {{Jiang}},
  \citenamefont {{Bremner}}, \citenamefont {{Martinis}},\ and\ \citenamefont
  {{Neven}}}]{Boixo2016}%
  \BibitemOpen
  \bibfield  {author} {\bibinfo {author} {\bibfnamefont {S.}~\bibnamefont
  {{Boixo}}}, \bibinfo {author} {\bibfnamefont {S.~V.}\ \bibnamefont
  {{Isakov}}}, \bibinfo {author} {\bibfnamefont {V.~N.}\ \bibnamefont
  {{Smelyanskiy}}}, \bibinfo {author} {\bibfnamefont {R.}~\bibnamefont
  {{Babbush}}}, \bibinfo {author} {\bibfnamefont {N.}~\bibnamefont {{Ding}}},
  \bibinfo {author} {\bibfnamefont {Z.}~\bibnamefont {{Jiang}}}, \bibinfo
  {author} {\bibfnamefont {M.~J.}\ \bibnamefont {{Bremner}}}, \bibinfo {author}
  {\bibfnamefont {J.~M.}\ \bibnamefont {{Martinis}}}, \ and\ \bibinfo {author}
  {\bibfnamefont {H.}~\bibnamefont {{Neven}}},\ }\href@noop {} {\bibfield
  {journal} {\bibinfo  {journal} {arXiv e-print}\ }\textbf {\bibinfo {volume}
  {{\!\!}}} (\bibinfo {year} {2016})},\ \Eprint
  {http://arxiv.org/abs/1608.00263} {arXiv:1608.00263 [quant-ph]} \BibitemShut
  {NoStop}%
\bibitem [{\citenamefont {Shchesnovich}(2015)}]{Shchesnovich2015}%
  \BibitemOpen
  \bibfield  {author} {\bibinfo {author} {\bibfnamefont {V.~S.}\ \bibnamefont
  {Shchesnovich}},\ }\href {\doibase 10.1103/PhysRevA.91.063842} {\bibfield
  {journal} {\bibinfo  {journal} {Phys. Rev. A}\ }\textbf {\bibinfo {volume}
  {91}},\ \bibinfo {pages} {063842} (\bibinfo {year} {2015})}\BibitemShut
  {NoStop}%
\bibitem [{\citenamefont {Renema}\ \emph {et~al.}(2017)\citenamefont {Renema},
  \citenamefont {Menssen}, \citenamefont {Clements}, \citenamefont {Triginer},
  \citenamefont {Kolthammer},\ and\ \citenamefont {Walmsley}}]{Renema2017}%
  \BibitemOpen
  \bibfield  {author} {\bibinfo {author} {\bibfnamefont {J.~J.}\ \bibnamefont
  {Renema}}, \bibinfo {author} {\bibfnamefont {A.}~\bibnamefont {Menssen}},
  \bibinfo {author} {\bibfnamefont {W.~R.}\ \bibnamefont {Clements}}, \bibinfo
  {author} {\bibfnamefont {G.}~\bibnamefont {Triginer}}, \bibinfo {author}
  {\bibfnamefont {W.~S.}\ \bibnamefont {Kolthammer}}, \ and\ \bibinfo {author}
  {\bibfnamefont {I.~A.}\ \bibnamefont {Walmsley}},\ }\href@noop {} {\bibfield
  {journal} {\bibinfo  {journal} {arXiv e-print}\ }\textbf {\bibinfo {volume}
  {{\!\!}}} (\bibinfo {year} {2017})},\ \Eprint
  {http://arxiv.org/abs/1707.02793} {arXiv:1707.02793 [quant-ph]} \BibitemShut
  {NoStop}%
\bibitem [{\citenamefont {Kalai}\ and\ \citenamefont
  {Kindler}(2014)}]{Kalai2014}%
  \BibitemOpen
  \bibfield  {author} {\bibinfo {author} {\bibfnamefont {G.}~\bibnamefont
  {Kalai}}\ and\ \bibinfo {author} {\bibfnamefont {G.}~\bibnamefont
  {Kindler}},\ }\href@noop {} {\bibfield  {journal} {\bibinfo  {journal} {arXiv
  e-print}\ }\textbf {\bibinfo {volume} {{\!\!}}} (\bibinfo {year} {2014})},\
  \Eprint {http://arxiv.org/abs/1409.3093} {arXiv:1409.3093 [quant-ph]}
  \BibitemShut {NoStop}%
\bibitem [{\citenamefont {Arkhipov}(2015)}]{Arkhipov2015}%
  \BibitemOpen
  \bibfield  {author} {\bibinfo {author} {\bibfnamefont {A.}~\bibnamefont
  {Arkhipov}},\ }\href {\doibase 10.1103/PhysRevA.92.062326} {\bibfield
  {journal} {\bibinfo  {journal} {Phys. Rev. A}\ }\textbf {\bibinfo {volume}
  {92}},\ \bibinfo {pages} {062326} (\bibinfo {year} {2015})}\BibitemShut
  {NoStop}%
\bibitem [{\citenamefont {Leverrier}\ and\ \citenamefont
  {Garc{\'\i}a-Patron}(2014)}]{Leverrier2014}%
  \BibitemOpen
  \bibfield  {author} {\bibinfo {author} {\bibfnamefont {A.}~\bibnamefont
  {Leverrier}}\ and\ \bibinfo {author} {\bibfnamefont {R.}~\bibnamefont
  {Garc{\'\i}a-Patron}},\ }\href@noop {} {\bibfield  {journal} {\bibinfo
  {journal} {Quant. Inf. Comp.}\ }\textbf {\bibinfo {volume} {15}},\ \bibinfo
  {pages} {489} (\bibinfo {year} {2014})}\BibitemShut {NoStop}%
\bibitem [{\citenamefont {Aaronson}\ and\ \citenamefont
  {Brod}(2016)}]{Aaronson2016}%
  \BibitemOpen
  \bibfield  {author} {\bibinfo {author} {\bibfnamefont {S.}~\bibnamefont
  {Aaronson}}\ and\ \bibinfo {author} {\bibfnamefont {D.~J.}\ \bibnamefont
  {Brod}},\ }\href {\doibase 10.1103/PhysRevA.93.012335} {\bibfield  {journal}
  {\bibinfo  {journal} {Phys. Rev. A}\ }\textbf {\bibinfo {volume} {93}},\
  \bibinfo {pages} {012335} (\bibinfo {year} {2016})}\BibitemShut {NoStop}%
\bibitem [{\citenamefont {Rahimi-Keshari}\ \emph {et~al.}(2016)\citenamefont
  {Rahimi-Keshari}, \citenamefont {Ralph},\ and\ \citenamefont
  {Caves}}]{Rahimi-Keshari2016}%
  \BibitemOpen
  \bibfield  {author} {\bibinfo {author} {\bibfnamefont {S.}~\bibnamefont
  {Rahimi-Keshari}}, \bibinfo {author} {\bibfnamefont {T.~C.}\ \bibnamefont
  {Ralph}}, \ and\ \bibinfo {author} {\bibfnamefont {C.~M.}\ \bibnamefont
  {Caves}},\ }\href {\doibase 10.1103/PhysRevX.6.021039} {\bibfield  {journal}
  {\bibinfo  {journal} {Phys. Rev. X}\ }\textbf {\bibinfo {volume} {6}},\
  \bibinfo {pages} {021039} (\bibinfo {year} {2016})}\BibitemShut {NoStop}%
\bibitem [{\citenamefont {Motes}\ \emph {et~al.}(2013)\citenamefont {Motes},
  \citenamefont {Dowling},\ and\ \citenamefont {Rohde}}]{Motes2013}%
  \BibitemOpen
  \bibfield  {author} {\bibinfo {author} {\bibfnamefont {K.~R.}\ \bibnamefont
  {Motes}}, \bibinfo {author} {\bibfnamefont {J.~P.}\ \bibnamefont {Dowling}},
  \ and\ \bibinfo {author} {\bibfnamefont {P.~P.}\ \bibnamefont {Rohde}},\
  }\href {\doibase 10.1103/PhysRevA.88.063822} {\bibfield  {journal} {\bibinfo
  {journal} {Phys. Rev. A}\ }\textbf {\bibinfo {volume} {88}},\ \bibinfo
  {pages} {063822} (\bibinfo {year} {2013})}\BibitemShut {NoStop}%
\bibitem [{\citenamefont {Broome}\ \emph {et~al.}(2013)\citenamefont {Broome},
  \citenamefont {Fedrizzi}, \citenamefont {Rahimi-Keshari}, \citenamefont
  {Dove}, \citenamefont {Aaronson}, \citenamefont {Ralph},\ and\ \citenamefont
  {White}}]{Broome2013}%
  \BibitemOpen
  \bibfield  {author} {\bibinfo {author} {\bibfnamefont {M.~A.}\ \bibnamefont
  {Broome}}, \bibinfo {author} {\bibfnamefont {A.}~\bibnamefont {Fedrizzi}},
  \bibinfo {author} {\bibfnamefont {S.}~\bibnamefont {Rahimi-Keshari}},
  \bibinfo {author} {\bibfnamefont {J.}~\bibnamefont {Dove}}, \bibinfo {author}
  {\bibfnamefont {S.}~\bibnamefont {Aaronson}}, \bibinfo {author}
  {\bibfnamefont {T.~C.}\ \bibnamefont {Ralph}}, \ and\ \bibinfo {author}
  {\bibfnamefont {A.~G.}\ \bibnamefont {White}},\ }\href {\doibase
  10.1126/science.1231440} {\bibfield  {journal} {\bibinfo  {journal}
  {Science}\ }\textbf {\bibinfo {volume} {339}},\ \bibinfo {pages} {794}
  (\bibinfo {year} {2013})}\BibitemShut {NoStop}%
\bibitem [{\citenamefont {Crespi}\ \emph {et~al.}(2013)\citenamefont {Crespi},
  \citenamefont {Osellame}, \citenamefont {Ramponi}, \citenamefont {Brod},
  \citenamefont {Galv{\~a}o}, \citenamefont {Spagnolo}, \citenamefont
  {Vitelli}, \citenamefont {Maiorino}, \citenamefont {Mataloni},\ and\
  \citenamefont {Sciarrino}}]{Crespi2013b}%
  \BibitemOpen
  \bibfield  {author} {\bibinfo {author} {\bibfnamefont {A.}~\bibnamefont
  {Crespi}}, \bibinfo {author} {\bibfnamefont {R.}~\bibnamefont {Osellame}},
  \bibinfo {author} {\bibfnamefont {R.}~\bibnamefont {Ramponi}}, \bibinfo
  {author} {\bibfnamefont {D.~J.}\ \bibnamefont {Brod}}, \bibinfo {author}
  {\bibfnamefont {E.~F.}\ \bibnamefont {Galv{\~a}o}}, \bibinfo {author}
  {\bibfnamefont {N.}~\bibnamefont {Spagnolo}}, \bibinfo {author}
  {\bibfnamefont {C.}~\bibnamefont {Vitelli}}, \bibinfo {author} {\bibfnamefont
  {E.}~\bibnamefont {Maiorino}}, \bibinfo {author} {\bibfnamefont
  {P.}~\bibnamefont {Mataloni}}, \ and\ \bibinfo {author} {\bibfnamefont
  {F.}~\bibnamefont {Sciarrino}},\ }\href {\doibase 10.1038/nphoton.2013.112}
  {\bibfield  {journal} {\bibinfo  {journal} {Nat. Photon.}\ }\textbf {\bibinfo
  {volume} {7}},\ \bibinfo {pages} {545} (\bibinfo {year} {2013})}\BibitemShut
  {NoStop}%
\bibitem [{\citenamefont {Spring}\ \emph {et~al.}(2013)\citenamefont {Spring},
  \citenamefont {Metcalf}, \citenamefont {Humphreys}, \citenamefont
  {Kolthammer}, \citenamefont {Jin}, \citenamefont {Barbieri}, \citenamefont
  {Datta}, \citenamefont {Thomas-Peter}, \citenamefont {Langford},
  \citenamefont {Kundys}, \citenamefont {Gates}, \citenamefont {Smith},
  \citenamefont {Smith},\ and\ \citenamefont {Walmsley}}]{Spring2013}%
  \BibitemOpen
  \bibfield  {author} {\bibinfo {author} {\bibfnamefont {J.~B.}\ \bibnamefont
  {Spring}}, \bibinfo {author} {\bibfnamefont {B.~J.}\ \bibnamefont {Metcalf}},
  \bibinfo {author} {\bibfnamefont {P.~C.}\ \bibnamefont {Humphreys}}, \bibinfo
  {author} {\bibfnamefont {W.~S.}\ \bibnamefont {Kolthammer}}, \bibinfo
  {author} {\bibfnamefont {X.-M.}\ \bibnamefont {Jin}}, \bibinfo {author}
  {\bibfnamefont {M.}~\bibnamefont {Barbieri}}, \bibinfo {author}
  {\bibfnamefont {A.}~\bibnamefont {Datta}}, \bibinfo {author} {\bibfnamefont
  {N.}~\bibnamefont {Thomas-Peter}}, \bibinfo {author} {\bibfnamefont {N.~K.}\
  \bibnamefont {Langford}}, \bibinfo {author} {\bibfnamefont {D.}~\bibnamefont
  {Kundys}}, \bibinfo {author} {\bibfnamefont {J.~C.}\ \bibnamefont {Gates}},
  \bibinfo {author} {\bibfnamefont {B.~J.}\ \bibnamefont {Smith}}, \bibinfo
  {author} {\bibfnamefont {P.~G.~R.}\ \bibnamefont {Smith}}, \ and\ \bibinfo
  {author} {\bibfnamefont {I.~A.}\ \bibnamefont {Walmsley}},\ }\href {\doibase
  10.1126/science.1231692} {\bibfield  {journal} {\bibinfo  {journal}
  {Science}\ }\textbf {\bibinfo {volume} {339}},\ \bibinfo {pages} {798}
  (\bibinfo {year} {2013})}\BibitemShut {NoStop}%
\bibitem [{\citenamefont {Tillmann}\ \emph {et~al.}(2013)\citenamefont
  {Tillmann}, \citenamefont {Daki{\'c}}, \citenamefont {Heilmann},
  \citenamefont {Nolte}, \citenamefont {Szameit},\ and\ \citenamefont
  {Walther}}]{Tillmann2013}%
  \BibitemOpen
  \bibfield  {author} {\bibinfo {author} {\bibfnamefont {M.}~\bibnamefont
  {Tillmann}}, \bibinfo {author} {\bibfnamefont {B.}~\bibnamefont {Daki{\'c}}},
  \bibinfo {author} {\bibfnamefont {R.}~\bibnamefont {Heilmann}}, \bibinfo
  {author} {\bibfnamefont {S.}~\bibnamefont {Nolte}}, \bibinfo {author}
  {\bibfnamefont {A.}~\bibnamefont {Szameit}}, \ and\ \bibinfo {author}
  {\bibfnamefont {P.}~\bibnamefont {Walther}},\ }\href {\doibase
  10.1038/nphoton.2013.102} {\bibfield  {journal} {\bibinfo  {journal} {Nat.
  Photon.}\ }\textbf {\bibinfo {volume} {7}},\ \bibinfo {pages} {540} (\bibinfo
  {year} {2013})}\BibitemShut {NoStop}%
\bibitem [{\citenamefont {Spagnolo}\ \emph {et~al.}(2014)\citenamefont
  {Spagnolo}, \citenamefont {Vitelli}, \citenamefont {Bentivegna},
  \citenamefont {Brod}, \citenamefont {Crespi}, \citenamefont {Flamini},
  \citenamefont {Giacomini}, \citenamefont {Milani}, \citenamefont {Ramponi},
  \citenamefont {Mataloni}, \citenamefont {Osellame}, \citenamefont
  {Galv{\~a}o},\ and\ \citenamefont {Sciarrino}}]{Spagnolo2014}%
  \BibitemOpen
  \bibfield  {author} {\bibinfo {author} {\bibfnamefont {N.}~\bibnamefont
  {Spagnolo}}, \bibinfo {author} {\bibfnamefont {C.}~\bibnamefont {Vitelli}},
  \bibinfo {author} {\bibfnamefont {M.}~\bibnamefont {Bentivegna}}, \bibinfo
  {author} {\bibfnamefont {D.~J.}\ \bibnamefont {Brod}}, \bibinfo {author}
  {\bibfnamefont {A.}~\bibnamefont {Crespi}}, \bibinfo {author} {\bibfnamefont
  {F.}~\bibnamefont {Flamini}}, \bibinfo {author} {\bibfnamefont
  {S.}~\bibnamefont {Giacomini}}, \bibinfo {author} {\bibfnamefont
  {G.}~\bibnamefont {Milani}}, \bibinfo {author} {\bibfnamefont
  {R.}~\bibnamefont {Ramponi}}, \bibinfo {author} {\bibfnamefont
  {P.}~\bibnamefont {Mataloni}}, \bibinfo {author} {\bibfnamefont
  {R.}~\bibnamefont {Osellame}}, \bibinfo {author} {\bibfnamefont {E.~F.}\
  \bibnamefont {Galv{\~a}o}}, \ and\ \bibinfo {author} {\bibfnamefont
  {F.}~\bibnamefont {Sciarrino}},\ }\href {\doibase 10.1038/nphoton.2014.135}
  {\bibfield  {journal} {\bibinfo  {journal} {Nat. Photon.}\ }\textbf {\bibinfo
  {volume} {8}},\ \bibinfo {pages} {615} (\bibinfo {year} {2014})}\BibitemShut
  {NoStop}%
\bibitem [{\citenamefont {Carolan}\ \emph {et~al.}(2014)\citenamefont
  {Carolan}, \citenamefont {Meinecke}, \citenamefont {Shadbolt}, \citenamefont
  {Russell}, \citenamefont {Ismail}, \citenamefont {W{\"o}rhoff}, \citenamefont
  {Rudolph}, \citenamefont {Thompson}, \citenamefont {O'Brien}, \citenamefont
  {Matthews},\ and\ \citenamefont {Laing}}]{Carolan2014}%
  \BibitemOpen
  \bibfield  {author} {\bibinfo {author} {\bibfnamefont {J.}~\bibnamefont
  {Carolan}}, \bibinfo {author} {\bibfnamefont {J.~D.~A.}\ \bibnamefont
  {Meinecke}}, \bibinfo {author} {\bibfnamefont {P.~J.}\ \bibnamefont
  {Shadbolt}}, \bibinfo {author} {\bibfnamefont {N.~J.}\ \bibnamefont
  {Russell}}, \bibinfo {author} {\bibfnamefont {N.}~\bibnamefont {Ismail}},
  \bibinfo {author} {\bibfnamefont {K.}~\bibnamefont {W{\"o}rhoff}}, \bibinfo
  {author} {\bibfnamefont {T.}~\bibnamefont {Rudolph}}, \bibinfo {author}
  {\bibfnamefont {M.~G.}\ \bibnamefont {Thompson}}, \bibinfo {author}
  {\bibfnamefont {J.~L.}\ \bibnamefont {O'Brien}}, \bibinfo {author}
  {\bibfnamefont {J.~C.~F.}\ \bibnamefont {Matthews}}, \ and\ \bibinfo {author}
  {\bibfnamefont {A.}~\bibnamefont {Laing}},\ }\href {\doibase
  10.1038/nphoton.2014.152} {\bibfield  {journal} {\bibinfo  {journal} {Nat.
  Photon.}\ }\textbf {\bibinfo {volume} {8}},\ \bibinfo {pages} {621} (\bibinfo
  {year} {2014})}\BibitemShut {NoStop}%
\bibitem [{\citenamefont {Carolan}\ \emph {et~al.}(2015)\citenamefont
  {Carolan}, \citenamefont {Harrold}, \citenamefont {Sparrow}, \citenamefont
  {Mart{\'\i}n-L{\'o}pez}, \citenamefont {Russell}, \citenamefont
  {Silverstone}, \citenamefont {Shadbolt}, \citenamefont {Matsuda},
  \citenamefont {Oguma}, \citenamefont {Itoh}, \citenamefont {Marshall},
  \citenamefont {Thompson}, \citenamefont {Matthews}, \citenamefont
  {Hashimoto}, \citenamefont {O'Brien},\ and\ \citenamefont
  {Laing}}]{Carolan2015}%
  \BibitemOpen
  \bibfield  {author} {\bibinfo {author} {\bibfnamefont {J.}~\bibnamefont
  {Carolan}}, \bibinfo {author} {\bibfnamefont {C.}~\bibnamefont {Harrold}},
  \bibinfo {author} {\bibfnamefont {C.}~\bibnamefont {Sparrow}}, \bibinfo
  {author} {\bibfnamefont {E.}~\bibnamefont {Mart{\'\i}n-L{\'o}pez}}, \bibinfo
  {author} {\bibfnamefont {N.~J.}\ \bibnamefont {Russell}}, \bibinfo {author}
  {\bibfnamefont {J.~W.}\ \bibnamefont {Silverstone}}, \bibinfo {author}
  {\bibfnamefont {P.~J.}\ \bibnamefont {Shadbolt}}, \bibinfo {author}
  {\bibfnamefont {N.}~\bibnamefont {Matsuda}}, \bibinfo {author} {\bibfnamefont
  {M.}~\bibnamefont {Oguma}}, \bibinfo {author} {\bibfnamefont
  {M.}~\bibnamefont {Itoh}}, \bibinfo {author} {\bibfnamefont {G.~D.}\
  \bibnamefont {Marshall}}, \bibinfo {author} {\bibfnamefont {M.~G.}\
  \bibnamefont {Thompson}}, \bibinfo {author} {\bibfnamefont {J.~C.~F.}\
  \bibnamefont {Matthews}}, \bibinfo {author} {\bibfnamefont {T.}~\bibnamefont
  {Hashimoto}}, \bibinfo {author} {\bibfnamefont {J.~L.}\ \bibnamefont
  {O'Brien}}, \ and\ \bibinfo {author} {\bibfnamefont {A.}~\bibnamefont
  {Laing}},\ }\href {\doibase 10.1126/science.aab3642} {\bibfield  {journal}
  {\bibinfo  {journal} {Science}\ }\textbf {\bibinfo {volume} {349}},\ \bibinfo
  {pages} {711} (\bibinfo {year} {2015})}\BibitemShut {NoStop}%
\bibitem [{\citenamefont {Bentivegna}\ \emph {et~al.}(2015)\citenamefont
  {Bentivegna}, \citenamefont {Spagnolo}, \citenamefont {Vitelli},
  \citenamefont {Flamini}, \citenamefont {Viggianiello}, \citenamefont
  {Latmiral}, \citenamefont {Mataloni}, \citenamefont {Brod}, \citenamefont
  {Galvao}, \citenamefont {Crespi}, \citenamefont {Ramponi}, \citenamefont
  {Osellame},\ and\ \citenamefont {Sciarrino}}]{Bentivegna2015}%
  \BibitemOpen
  \bibfield  {author} {\bibinfo {author} {\bibfnamefont {M.}~\bibnamefont
  {Bentivegna}}, \bibinfo {author} {\bibfnamefont {N.}~\bibnamefont
  {Spagnolo}}, \bibinfo {author} {\bibfnamefont {C.}~\bibnamefont {Vitelli}},
  \bibinfo {author} {\bibfnamefont {F.}~\bibnamefont {Flamini}}, \bibinfo
  {author} {\bibfnamefont {N.}~\bibnamefont {Viggianiello}}, \bibinfo {author}
  {\bibfnamefont {L.}~\bibnamefont {Latmiral}}, \bibinfo {author}
  {\bibfnamefont {P.}~\bibnamefont {Mataloni}}, \bibinfo {author}
  {\bibfnamefont {D.~J.}\ \bibnamefont {Brod}}, \bibinfo {author}
  {\bibfnamefont {E.~F.}\ \bibnamefont {Galvao}}, \bibinfo {author}
  {\bibfnamefont {A.}~\bibnamefont {Crespi}}, \bibinfo {author} {\bibfnamefont
  {R.}~\bibnamefont {Ramponi}}, \bibinfo {author} {\bibfnamefont
  {R.}~\bibnamefont {Osellame}}, \ and\ \bibinfo {author} {\bibfnamefont
  {F.}~\bibnamefont {Sciarrino}},\ }\href {\doibase 10.1126/sciadv.1400255}
  {\bibfield  {journal} {\bibinfo  {journal} {Science Advances}\ }\textbf
  {\bibinfo {volume} {1}},\ \bibinfo {pages} {e1400255} (\bibinfo {year}
  {2015})}\BibitemShut {NoStop}%
\bibitem [{\citenamefont {Loredo}\ \emph {et~al.}(2017)\citenamefont {Loredo},
  \citenamefont {Broome}, \citenamefont {Hilaire}, \citenamefont {Gazzano},
  \citenamefont {Sagnes}, \citenamefont {Lemaitre}, \citenamefont {Almeida},
  \citenamefont {Senellart},\ and\ \citenamefont {White}}]{Loredo2017}%
  \BibitemOpen
  \bibfield  {author} {\bibinfo {author} {\bibfnamefont {J.~C.}\ \bibnamefont
  {Loredo}}, \bibinfo {author} {\bibfnamefont {M.~A.}\ \bibnamefont {Broome}},
  \bibinfo {author} {\bibfnamefont {P.}~\bibnamefont {Hilaire}}, \bibinfo
  {author} {\bibfnamefont {O.}~\bibnamefont {Gazzano}}, \bibinfo {author}
  {\bibfnamefont {I.}~\bibnamefont {Sagnes}}, \bibinfo {author} {\bibfnamefont
  {A.}~\bibnamefont {Lemaitre}}, \bibinfo {author} {\bibfnamefont {M.~P.}\
  \bibnamefont {Almeida}}, \bibinfo {author} {\bibfnamefont {P.}~\bibnamefont
  {Senellart}}, \ and\ \bibinfo {author} {\bibfnamefont {A.~G.}\ \bibnamefont
  {White}},\ }\href {\doibase 10.1103/PhysRevLett.118.130503} {\bibfield
  {journal} {\bibinfo  {journal} {Phys. Rev. Lett.}\ }\textbf {\bibinfo
  {volume} {118}},\ \bibinfo {pages} {130503} (\bibinfo {year}
  {2017})}\BibitemShut {NoStop}%
\bibitem [{\citenamefont {He}\ \emph {et~al.}(2017)\citenamefont {He},
  \citenamefont {Ding}, \citenamefont {Su}, \citenamefont {Huang},
  \citenamefont {Qin}, \citenamefont {Wang}, \citenamefont {Unsleber},
  \citenamefont {Chen}, \citenamefont {Wang}, \citenamefont {He}, \citenamefont
  {Wang}, \citenamefont {Zhang}, \citenamefont {Chen}, \citenamefont
  {Schneider}, \citenamefont {Kamp}, \citenamefont {You}, \citenamefont {Wang},
  \citenamefont {H\"ofling}, \citenamefont {Lu},\ and\ \citenamefont
  {Pan}}]{He2017}%
  \BibitemOpen
  \bibfield  {author} {\bibinfo {author} {\bibfnamefont {Y.}~\bibnamefont
  {He}}, \bibinfo {author} {\bibfnamefont {X.}~\bibnamefont {Ding}}, \bibinfo
  {author} {\bibfnamefont {Z.-E.}\ \bibnamefont {Su}}, \bibinfo {author}
  {\bibfnamefont {H.-L.}\ \bibnamefont {Huang}}, \bibinfo {author}
  {\bibfnamefont {J.}~\bibnamefont {Qin}}, \bibinfo {author} {\bibfnamefont
  {C.}~\bibnamefont {Wang}}, \bibinfo {author} {\bibfnamefont {S.}~\bibnamefont
  {Unsleber}}, \bibinfo {author} {\bibfnamefont {C.}~\bibnamefont {Chen}},
  \bibinfo {author} {\bibfnamefont {H.}~\bibnamefont {Wang}}, \bibinfo {author}
  {\bibfnamefont {Y.-M.}\ \bibnamefont {He}}, \bibinfo {author} {\bibfnamefont
  {X.-L.}\ \bibnamefont {Wang}}, \bibinfo {author} {\bibfnamefont {W.-J.}\
  \bibnamefont {Zhang}}, \bibinfo {author} {\bibfnamefont {S.-J.}\ \bibnamefont
  {Chen}}, \bibinfo {author} {\bibfnamefont {C.}~\bibnamefont {Schneider}},
  \bibinfo {author} {\bibfnamefont {M.}~\bibnamefont {Kamp}}, \bibinfo {author}
  {\bibfnamefont {L.-X.}\ \bibnamefont {You}}, \bibinfo {author} {\bibfnamefont
  {Z.}~\bibnamefont {Wang}}, \bibinfo {author} {\bibfnamefont {S.}~\bibnamefont
  {H\"ofling}}, \bibinfo {author} {\bibfnamefont {C.-Y.}\ \bibnamefont {Lu}}, \
  and\ \bibinfo {author} {\bibfnamefont {J.-W.}\ \bibnamefont {Pan}},\ }\href
  {\doibase 10.1103/PhysRevLett.118.190501} {\bibfield  {journal} {\bibinfo
  {journal} {Phys. Rev. Lett.}\ }\textbf {\bibinfo {volume} {118}},\ \bibinfo
  {pages} {190501} (\bibinfo {year} {2017})}\BibitemShut {NoStop}%
\bibitem [{\citenamefont {Lund}\ \emph {et~al.}(2014)\citenamefont {Lund},
  \citenamefont {Laing}, \citenamefont {Rahimi-Keshari}, \citenamefont
  {Rudolph}, \citenamefont {O'Brien},\ and\ \citenamefont {Ralph}}]{Lund2014}%
  \BibitemOpen
  \bibfield  {author} {\bibinfo {author} {\bibfnamefont {A.~P.}\ \bibnamefont
  {Lund}}, \bibinfo {author} {\bibfnamefont {A.}~\bibnamefont {Laing}},
  \bibinfo {author} {\bibfnamefont {S.}~\bibnamefont {Rahimi-Keshari}},
  \bibinfo {author} {\bibfnamefont {T.}~\bibnamefont {Rudolph}}, \bibinfo
  {author} {\bibfnamefont {J.~L.}\ \bibnamefont {O'Brien}}, \ and\ \bibinfo
  {author} {\bibfnamefont {T.~C.}\ \bibnamefont {Ralph}},\ }\href {\doibase
  10.1103/PhysRevLett.113.100502} {\bibfield  {journal} {\bibinfo  {journal}
  {Phys. Rev. Lett.}\ }\textbf {\bibinfo {volume} {113}},\ \bibinfo {pages}
  {100502} (\bibinfo {year} {2014})}\BibitemShut {NoStop}%
\bibitem [{\citenamefont {{Scott Aaronson’s blog, acknowledged to S.
  Kolthammer}}()}]{Scottblog}%
  \BibitemOpen
  \bibfield  {author} {\bibinfo {author} {\bibnamefont {{Scott Aaronson’s
  blog, acknowledged to S. Kolthammer}}},\ }\href
  {http://www.scottaaronson.com/blog/?p=1579} {\enquote {\bibinfo {title}
  {http://www.scottaaronson.com/blog/?p=1579},}\ }\BibitemShut {NoStop}%
\bibitem [{\citenamefont {Neville}\ \emph {et~al.}(2017)\citenamefont
  {Neville}, \citenamefont {Sparrow}, \citenamefont {Clifford}, \citenamefont
  {Johnston}, \citenamefont {Birchall}, \citenamefont {Montanaro},\ and\
  \citenamefont {Laing}}]{Neville2017}%
  \BibitemOpen
  \bibfield  {author} {\bibinfo {author} {\bibfnamefont {A.}~\bibnamefont
  {Neville}}, \bibinfo {author} {\bibfnamefont {C.}~\bibnamefont {Sparrow}},
  \bibinfo {author} {\bibfnamefont {R.}~\bibnamefont {Clifford}}, \bibinfo
  {author} {\bibfnamefont {E.}~\bibnamefont {Johnston}}, \bibinfo {author}
  {\bibfnamefont {P.~M.}\ \bibnamefont {Birchall}}, \bibinfo {author}
  {\bibfnamefont {A.}~\bibnamefont {Montanaro}}, \ and\ \bibinfo {author}
  {\bibfnamefont {A.}~\bibnamefont {Laing}},\ }\href {\doibase
  doi:10.1038/nphys4270} {\bibfield  {journal} {\bibinfo  {journal} {Nature
  Physics}\ }\textbf {\bibinfo {volume} {13}},\ \bibinfo {pages} {1153}
  (\bibinfo {year} {2017})}\BibitemShut {NoStop}%
\bibitem [{\citenamefont {Clifford}\ and\ \citenamefont
  {Clifford}(2017)}]{Clifford2017}%
  \BibitemOpen
  \bibfield  {author} {\bibinfo {author} {\bibfnamefont {P.}~\bibnamefont
  {Clifford}}\ and\ \bibinfo {author} {\bibfnamefont {R.}~\bibnamefont
  {Clifford}},\ }\href@noop {} {\bibfield  {journal} {\bibinfo  {journal}
  {arXiv e-print}\ }\textbf {\bibinfo {volume} {{\!\!}}} (\bibinfo {year}
  {2017})},\ \Eprint {http://arxiv.org/abs/1706.01260} {arXiv:1706.01260
  [cs.DS]} \BibitemShut {NoStop}%
\bibitem [{\citenamefont {H\"ubener}\ \emph {et~al.}(2009)\citenamefont
  {H\"ubener}, \citenamefont {Kleinmann}, \citenamefont {Wei}, \citenamefont
  {Gonz\'alez-Guill\'en},\ and\ \citenamefont {G\"uhne}}]{Ghune2009}%
  \BibitemOpen
  \bibfield  {author} {\bibinfo {author} {\bibfnamefont {R.}~\bibnamefont
  {H\"ubener}}, \bibinfo {author} {\bibfnamefont {M.}~\bibnamefont
  {Kleinmann}}, \bibinfo {author} {\bibfnamefont {T.-C.}\ \bibnamefont {Wei}},
  \bibinfo {author} {\bibfnamefont {C.}~\bibnamefont {Gonz\'alez-Guill\'en}}, \
  and\ \bibinfo {author} {\bibfnamefont {O.}~\bibnamefont {G\"uhne}},\ }\href
  {\doibase 10.1103/PhysRevA.80.032324} {\bibfield  {journal} {\bibinfo
  {journal} {Phys. Rev. A}\ }\textbf {\bibinfo {volume} {80}},\ \bibinfo
  {pages} {032324} (\bibinfo {year} {2009})}\BibitemShut {NoStop}%
\bibitem [{\citenamefont {Yu}(2016)}]{YU2016}%
  \BibitemOpen
  \bibfield  {author} {\bibinfo {author} {\bibfnamefont {N.}~\bibnamefont
  {Yu}},\ }\href {\doibase 10.1103/PhysRevA.94.060101} {\bibfield  {journal}
  {\bibinfo  {journal} {Phys. Rev. A}\ }\textbf {\bibinfo {volume} {94}},\
  \bibinfo {pages} {060101} (\bibinfo {year} {2016})}\BibitemShut {NoStop}%
\bibitem [{\citenamefont {Quesada}\ \emph {et~al.}(2017)\citenamefont
  {Quesada}, \citenamefont {Rana},\ and\ \citenamefont
  {Sanpera}}]{Quesada2017}%
  \BibitemOpen
  \bibfield  {author} {\bibinfo {author} {\bibfnamefont {R.}~\bibnamefont
  {Quesada}}, \bibinfo {author} {\bibfnamefont {S.}~\bibnamefont {Rana}}, \
  and\ \bibinfo {author} {\bibfnamefont {A.}~\bibnamefont {Sanpera}},\ }\href
  {\doibase 10.1103/PhysRevA.95.042128} {\bibfield  {journal} {\bibinfo
  {journal} {Phys. Rev. A}\ }\textbf {\bibinfo {volume} {95}},\ \bibinfo
  {pages} {042128} (\bibinfo {year} {2017})}\BibitemShut {NoStop}%
\bibitem [{\citenamefont {Reck}\ \emph {et~al.}(1994)\citenamefont {Reck},
  \citenamefont {Zeilinger}, \citenamefont {Bernstein},\ and\ \citenamefont
  {Bertani}}]{Reck1994}%
  \BibitemOpen
  \bibfield  {author} {\bibinfo {author} {\bibfnamefont {M.}~\bibnamefont
  {Reck}}, \bibinfo {author} {\bibfnamefont {A.}~\bibnamefont {Zeilinger}},
  \bibinfo {author} {\bibfnamefont {H.}~\bibnamefont {Bernstein}}, \ and\
  \bibinfo {author} {\bibfnamefont {P.}~\bibnamefont {Bertani}},\ }\href
  {\doibase 10.1103/PhysRevLett.73.58} {\bibfield  {journal} {\bibinfo
  {journal} {Phys. Rev. Lett.}\ }\textbf {\bibinfo {volume} {73}},\ \bibinfo
  {pages} {58} (\bibinfo {year} {1994})}\BibitemShut {NoStop}%
\bibitem [{\citenamefont {Clements}\ \emph {et~al.}(2016)\citenamefont
  {Clements}, \citenamefont {Humphreys}, \citenamefont {Metcalf}, \citenamefont
  {Kolthammer},\ and\ \citenamefont {Walmsley}}]{Clements2016}%
  \BibitemOpen
  \bibfield  {author} {\bibinfo {author} {\bibfnamefont {W.~R.}\ \bibnamefont
  {Clements}}, \bibinfo {author} {\bibfnamefont {P.~C.}\ \bibnamefont
  {Humphreys}}, \bibinfo {author} {\bibfnamefont {B.~J.}\ \bibnamefont
  {Metcalf}}, \bibinfo {author} {\bibfnamefont {W.~S.}\ \bibnamefont
  {Kolthammer}}, \ and\ \bibinfo {author} {\bibfnamefont {I.~A.}\ \bibnamefont
  {Walmsley}},\ }\href@noop {} {\bibfield  {journal} {\bibinfo  {journal}
  {arXiv e-print}\ }\textbf {\bibinfo {volume} {{\!\!}}} (\bibinfo {year}
  {2016})},\ \Eprint {http://arxiv.org/abs/1603.08788} {arXiv:1603.08788
  [quant-ph]} \BibitemShut {NoStop}%
\bibitem [{\citenamefont {Demkowicz-Dobrza\'{n}ski}\ \emph
  {et~al.}(2015)\citenamefont {Demkowicz-Dobrza\'{n}ski}, \citenamefont
  {Jarzyna},\ and\ \citenamefont {Ko\l{}ody\'{n}ski}}]{Demkowicz2015}%
  \BibitemOpen
  \bibfield  {author} {\bibinfo {author} {\bibfnamefont {R.}~\bibnamefont
  {Demkowicz-Dobrza\'{n}ski}}, \bibinfo {author} {\bibfnamefont
  {M.}~\bibnamefont {Jarzyna}}, \ and\ \bibinfo {author} {\bibfnamefont
  {J.}~\bibnamefont {Ko\l{}ody\'{n}ski}},\ }in\ \href {\doibase
  10.1016/bs.po.2015.02.003} {\emph {\bibinfo {booktitle} {Progress in
  Optics}}},\ Vol.~\bibinfo {volume} {60},\ \bibinfo {editor} {edited by\
  \bibinfo {editor} {\bibfnamefont {E.}~\bibnamefont {Wolf}}}\ (\bibinfo
  {publisher} {Elsevier},\ \bibinfo {year} {2015})\ pp.\ \bibinfo {pages}
  {345--435}\BibitemShut {NoStop}%
\bibitem [{\citenamefont {Bachor}\ and\ \citenamefont
  {Ralph}(2004)}]{Bachor2004}%
  \BibitemOpen
  \bibfield  {author} {\bibinfo {author} {\bibfnamefont {H.-A.}\ \bibnamefont
  {Bachor}}\ and\ \bibinfo {author} {\bibfnamefont {T.~C.}\ \bibnamefont
  {Ralph}},\ }\href@noop {} {\emph {\bibinfo {title} {{A Guide to Experiments
  in Quantum Optics}}}}\ (\bibinfo  {publisher} {Wiley},\ \bibinfo {year}
  {2004})\BibitemShut {NoStop}%
\bibitem [{\citenamefont {Weedbrook}\ \emph {et~al.}(2012)\citenamefont
  {Weedbrook}, \citenamefont {Pirandola}, \citenamefont {Garc\'{\i}a-Patr\'on},
  \citenamefont {Cerf}, \citenamefont {Ralph}, \citenamefont {Shapiro},\ and\
  \citenamefont {Lloyd}}]{GaussQuantInfo2012}%
  \BibitemOpen
  \bibfield  {author} {\bibinfo {author} {\bibfnamefont {C.}~\bibnamefont
  {Weedbrook}}, \bibinfo {author} {\bibfnamefont {S.}~\bibnamefont
  {Pirandola}}, \bibinfo {author} {\bibfnamefont {R.}~\bibnamefont
  {Garc\'{\i}a-Patr\'on}}, \bibinfo {author} {\bibfnamefont {N.~J.}\
  \bibnamefont {Cerf}}, \bibinfo {author} {\bibfnamefont {T.~C.}\ \bibnamefont
  {Ralph}}, \bibinfo {author} {\bibfnamefont {J.~H.}\ \bibnamefont {Shapiro}},
  \ and\ \bibinfo {author} {\bibfnamefont {S.}~\bibnamefont {Lloyd}},\ }\href
  {\doibase 10.1103/RevModPhys.84.621} {\bibfield  {journal} {\bibinfo
  {journal} {Rev. Mod. Phys.}\ }\textbf {\bibinfo {volume} {84}},\ \bibinfo
  {pages} {621} (\bibinfo {year} {2012})}\BibitemShut {NoStop}%
\bibitem [{\citenamefont {Eckert}\ \emph {et~al.}(2002)\citenamefont {Eckert},
  \citenamefont {Schliemann}, \citenamefont {Bruß},\ and\ \citenamefont
  {Lewenstein}}]{Eeckert2002}%
  \BibitemOpen
  \bibfield  {author} {\bibinfo {author} {\bibfnamefont {K.}~\bibnamefont
  {Eckert}}, \bibinfo {author} {\bibfnamefont {J.}~\bibnamefont {Schliemann}},
  \bibinfo {author} {\bibfnamefont {D.}~\bibnamefont {Bruß}}, \ and\ \bibinfo
  {author} {\bibfnamefont {M.}~\bibnamefont {Lewenstein}},\ }\href {\doibase
  https://doi.org/10.1006/aphy.2002.6268} {\bibfield  {journal} {\bibinfo
  {journal} {Annals of Physics}\ }\textbf {\bibinfo {volume} {299}},\ \bibinfo
  {pages} {88 } (\bibinfo {year} {2002})}\BibitemShut {NoStop}%
\bibitem [{\citenamefont {Killoran}\ \emph {et~al.}(2014)\citenamefont
  {Killoran}, \citenamefont {Cramer},\ and\ \citenamefont
  {Plenio}}]{Killoran2014}%
  \BibitemOpen
  \bibfield  {author} {\bibinfo {author} {\bibfnamefont {N.}~\bibnamefont
  {Killoran}}, \bibinfo {author} {\bibfnamefont {M.}~\bibnamefont {Cramer}}, \
  and\ \bibinfo {author} {\bibfnamefont {M.~B.}\ \bibnamefont {Plenio}},\
  }\href {\doibase 10.1103/PhysRevLett.112.150501} {\bibfield  {journal}
  {\bibinfo  {journal} {Phys. Rev. Lett.}\ }\textbf {\bibinfo {volume} {112}},\
  \bibinfo {pages} {150501} (\bibinfo {year} {2014})}\BibitemShut {NoStop}%
\bibitem [{\citenamefont {T\'{o}th}\ and\ \citenamefont
  {Apellaniz}(2014)}]{Toth2014}%
  \BibitemOpen
  \bibfield  {author} {\bibinfo {author} {\bibfnamefont {G.}~\bibnamefont
  {T\'{o}th}}\ and\ \bibinfo {author} {\bibfnamefont {I.}~\bibnamefont
  {Apellaniz}},\ }\href {\doibase 10.1088/1751-8113/47/42/424006} {\bibfield
  {journal} {\bibinfo  {journal} {J. Phys. A: Math. Theor.}\ }\textbf {\bibinfo
  {volume} {47}},\ \bibinfo {pages} {424006} (\bibinfo {year}
  {2014})}\BibitemShut {NoStop}%
\bibitem [{\citenamefont {Wasak}\ \emph {et~al.}(2016)\citenamefont {Wasak},
  \citenamefont {Smerzi},\ and\ \citenamefont {Chwedenczuk}}]{Wasak2016}%
  \BibitemOpen
  \bibfield  {author} {\bibinfo {author} {\bibfnamefont {T.}~\bibnamefont
  {Wasak}}, \bibinfo {author} {\bibfnamefont {A.}~\bibnamefont {Smerzi}}, \
  and\ \bibinfo {author} {\bibfnamefont {J.}~\bibnamefont {Chwedenczuk}},\
  }\href@noop {} {\bibfield  {journal} {\bibinfo  {journal} {arXiv e-print}\
  }\textbf {\bibinfo {volume} {{\!\!}}} (\bibinfo {year} {2016})},\ \Eprint
  {http://arxiv.org/abs/1609.01576} {arXiv:1609.01576 [quant-ph]} \BibitemShut
  {NoStop}%
\bibitem [{\citenamefont {Terhal}\ and\ \citenamefont
  {DiVincenzo}(2004)}]{Terhal2004}%
  \BibitemOpen
  \bibfield  {author} {\bibinfo {author} {\bibfnamefont {B.}~\bibnamefont
  {Terhal}}\ and\ \bibinfo {author} {\bibfnamefont {D.}~\bibnamefont
  {DiVincenzo}},\ }\href@noop {} {\bibfield  {journal} {\bibinfo  {journal}
  {Quant. Inf. Comp.}\ }\textbf {\bibinfo {volume} {4}},\ \bibinfo {pages}
  {134} (\bibinfo {year} {2004})}\BibitemShut {NoStop}%
\bibitem [{\citenamefont {Jozsa}(2006)}]{Jozsa2006}%
  \BibitemOpen
  \bibfield  {author} {\bibinfo {author} {\bibfnamefont {R.}~\bibnamefont
  {Jozsa}},\ }\href@noop {} {\bibfield  {journal} {\bibinfo  {journal} {arXiv
  e-print}\ }\textbf {\bibinfo {volume} {{\!\!}}} (\bibinfo {year} {2006})},\
  \Eprint {http://arxiv.org/abs/quant-ph/0603163} {arXiv:quant-ph/0603163}
  \BibitemShut {NoStop}%
\bibitem [{\citenamefont {Brod}(2015)}]{Brod2015}%
  \BibitemOpen
  \bibfield  {author} {\bibinfo {author} {\bibfnamefont {D.~J.}\ \bibnamefont
  {Brod}},\ }\href {\doibase 10.1103/PhysRevA.91.042316} {\bibfield  {journal}
  {\bibinfo  {journal} {Phys. Rev. A}\ }\textbf {\bibinfo {volume} {91}},\
  \bibinfo {pages} {042316} (\bibinfo {year} {2015})}\BibitemShut {NoStop}%
\bibitem [{\citenamefont {Oszmaniec}\ \emph {et~al.}(2016)\citenamefont
  {Oszmaniec}, \citenamefont {Augusiak}, \citenamefont {Gogolin}, \citenamefont
  {Ko{\l}ody{\'n}ski}, \citenamefont {Ac\'{\i}n},\ and\ \citenamefont
  {Lewenstein}}]{Oszmaniec2016}%
  \BibitemOpen
  \bibfield  {author} {\bibinfo {author} {\bibfnamefont {M.}~\bibnamefont
  {Oszmaniec}}, \bibinfo {author} {\bibfnamefont {R.}~\bibnamefont {Augusiak}},
  \bibinfo {author} {\bibfnamefont {C.}~\bibnamefont {Gogolin}}, \bibinfo
  {author} {\bibfnamefont {J.}~\bibnamefont {Ko{\l}ody{\'n}ski}}, \bibinfo
  {author} {\bibfnamefont {A.}~\bibnamefont {Ac\'{\i}n}}, \ and\ \bibinfo
  {author} {\bibfnamefont {M.}~\bibnamefont {Lewenstein}},\ }\href {\doibase
  10.1103/PhysRevX.6.041044} {\bibfield  {journal} {\bibinfo  {journal} {Phys.
  Rev. X}\ }\textbf {\bibinfo {volume} {6}},\ \bibinfo {pages} {041044}
  (\bibinfo {year} {2016})}\BibitemShut {NoStop}%
\bibitem [{\citenamefont {Schumm}\ \emph {et~al.}(2005)\citenamefont {Schumm},
  \citenamefont {Hofferberth}, \citenamefont {Andersson}, \citenamefont
  {Wildermuth}, \citenamefont {Groth}, \citenamefont {Bar-Joseph},
  \citenamefont {Schmiedmayer},\ and\ \citenamefont {Kr{\"u}ger}}]{Schumm2005}%
  \BibitemOpen
  \bibfield  {author} {\bibinfo {author} {\bibfnamefont {T.}~\bibnamefont
  {Schumm}}, \bibinfo {author} {\bibfnamefont {S.}~\bibnamefont {Hofferberth}},
  \bibinfo {author} {\bibfnamefont {L.~M.}\ \bibnamefont {Andersson}}, \bibinfo
  {author} {\bibfnamefont {S.}~\bibnamefont {Wildermuth}}, \bibinfo {author}
  {\bibfnamefont {S.}~\bibnamefont {Groth}}, \bibinfo {author} {\bibfnamefont
  {I.}~\bibnamefont {Bar-Joseph}}, \bibinfo {author} {\bibfnamefont
  {J.}~\bibnamefont {Schmiedmayer}}, \ and\ \bibinfo {author} {\bibfnamefont
  {P.}~\bibnamefont {Kr{\"u}ger}},\ }\href {\doibase 10.1038/nphys125}
  {\bibfield  {journal} {\bibinfo  {journal} {Nature Physics}\ }\textbf
  {\bibinfo {volume} {1}},\ \bibinfo {pages} {57} (\bibinfo {year}
  {2005})}\BibitemShut {NoStop}%
\bibitem [{\citenamefont {Sebby-Strabley}\ \emph {et~al.}(2007)\citenamefont
  {Sebby-Strabley}, \citenamefont {Brown}, \citenamefont {Anderlini},
  \citenamefont {Lee}, \citenamefont {Phillips}, \citenamefont {Porto},\ and\
  \citenamefont {Johnson}}]{Sebby2007}%
  \BibitemOpen
  \bibfield  {author} {\bibinfo {author} {\bibfnamefont {J.}~\bibnamefont
  {Sebby-Strabley}}, \bibinfo {author} {\bibfnamefont {B.~L.}\ \bibnamefont
  {Brown}}, \bibinfo {author} {\bibfnamefont {M.}~\bibnamefont {Anderlini}},
  \bibinfo {author} {\bibfnamefont {P.~J.}\ \bibnamefont {Lee}}, \bibinfo
  {author} {\bibfnamefont {W.~D.}\ \bibnamefont {Phillips}}, \bibinfo {author}
  {\bibfnamefont {J.~V.}\ \bibnamefont {Porto}}, \ and\ \bibinfo {author}
  {\bibfnamefont {P.~R.}\ \bibnamefont {Johnson}},\ }\href {\doibase
  10.1103/PhysRevLett.98.200405} {\bibfield  {journal} {\bibinfo  {journal}
  {Phys. Rev. Lett.}\ }\textbf {\bibinfo {volume} {98}},\ \bibinfo {pages}
  {200405} (\bibinfo {year} {2007})}\BibitemShut {NoStop}%
\bibitem [{\citenamefont {Zhang}\ \emph {et~al.}(2012)\citenamefont {Zhang},
  \citenamefont {McConnell}, \citenamefont {\ifmmode~\acute{C}\else
  \'{C}\fi{}uk}, \citenamefont {Lin}, \citenamefont {Schleier-Smith},
  \citenamefont {Leroux},\ and\ \citenamefont {Vuleti\ifmmode~\acute{c}\else
  \'{c}\fi{}}}]{Zhang2012}%
  \BibitemOpen
  \bibfield  {author} {\bibinfo {author} {\bibfnamefont {H.}~\bibnamefont
  {Zhang}}, \bibinfo {author} {\bibfnamefont {R.}~\bibnamefont {McConnell}},
  \bibinfo {author} {\bibfnamefont {S.}~\bibnamefont {\ifmmode~\acute{C}\else
  \'{C}\fi{}uk}}, \bibinfo {author} {\bibfnamefont {Q.}~\bibnamefont {Lin}},
  \bibinfo {author} {\bibfnamefont {M.~H.}\ \bibnamefont {Schleier-Smith}},
  \bibinfo {author} {\bibfnamefont {I.~D.}\ \bibnamefont {Leroux}}, \ and\
  \bibinfo {author} {\bibfnamefont {V.}~\bibnamefont
  {Vuleti\ifmmode~\acute{c}\else \'{c}\fi{}}},\ }\href {\doibase
  10.1103/PhysRevLett.109.133603} {\bibfield  {journal} {\bibinfo  {journal}
  {Phys. Rev. Lett.}\ }\textbf {\bibinfo {volume} {109}},\ \bibinfo {pages}
  {133603} (\bibinfo {year} {2012})}\BibitemShut {NoStop}%
\bibitem [{\citenamefont {Stroescu}\ \emph {et~al.}(2015)\citenamefont
  {Stroescu}, \citenamefont {Hume},\ and\ \citenamefont
  {Oberthaler}}]{Stroescu2015}%
  \BibitemOpen
  \bibfield  {author} {\bibinfo {author} {\bibfnamefont {I.}~\bibnamefont
  {Stroescu}}, \bibinfo {author} {\bibfnamefont {D.~B.}\ \bibnamefont {Hume}},
  \ and\ \bibinfo {author} {\bibfnamefont {M.~K.}\ \bibnamefont {Oberthaler}},\
  }\href {\doibase 10.1103/PhysRevA.91.013412} {\bibfield  {journal} {\bibinfo
  {journal} {Phys. Rev. A}\ }\textbf {\bibinfo {volume} {91}},\ \bibinfo
  {pages} {013412} (\bibinfo {year} {2015})}\BibitemShut {NoStop}%
\bibitem [{\citenamefont {Barnett}\ \emph {et~al.}(1998)\citenamefont
  {Barnett}, \citenamefont {Jeffers}, \citenamefont {Gatti},\ and\
  \citenamefont {Loudon}}]{Barnett1998}%
  \BibitemOpen
  \bibfield  {author} {\bibinfo {author} {\bibfnamefont {S.~M.}\ \bibnamefont
  {Barnett}}, \bibinfo {author} {\bibfnamefont {J.}~\bibnamefont {Jeffers}},
  \bibinfo {author} {\bibfnamefont {A.}~\bibnamefont {Gatti}}, \ and\ \bibinfo
  {author} {\bibfnamefont {R.}~\bibnamefont {Loudon}},\ }\href {\doibase
  10.1103/PhysRevA.57.2134} {\bibfield  {journal} {\bibinfo  {journal} {Phys.
  Rev. A}\ }\textbf {\bibinfo {volume} {57}},\ \bibinfo {pages} {2134}
  (\bibinfo {year} {1998})}\BibitemShut {NoStop}%
\bibitem [{\citenamefont {Rahimi-Keshari}\ \emph {et~al.}(2013)\citenamefont
  {Rahimi-Keshari}, \citenamefont {Broome}, \citenamefont {Fickler},
  \citenamefont {Fedrizzi}, \citenamefont {Ralph},\ and\ \citenamefont
  {White}}]{RahimiKeshari2013}%
  \BibitemOpen
  \bibfield  {author} {\bibinfo {author} {\bibfnamefont {S.}~\bibnamefont
  {Rahimi-Keshari}}, \bibinfo {author} {\bibfnamefont {M.~A.}\ \bibnamefont
  {Broome}}, \bibinfo {author} {\bibfnamefont {R.}~\bibnamefont {Fickler}},
  \bibinfo {author} {\bibfnamefont {A.}~\bibnamefont {Fedrizzi}}, \bibinfo
  {author} {\bibfnamefont {T.~C.}\ \bibnamefont {Ralph}}, \ and\ \bibinfo
  {author} {\bibfnamefont {A.~G.}\ \bibnamefont {White}},\ }\href {\doibase
  10.1364/OE.21.013450} {\bibfield  {journal} {\bibinfo  {journal} {Opt.
  Express}\ }\textbf {\bibinfo {volume} {21}},\ \bibinfo {pages} {13450}
  (\bibinfo {year} {2013})}\BibitemShut {NoStop}%
\bibitem [{\citenamefont {Scheel}(2004)}]{Scheel2004}%
  \BibitemOpen
  \bibfield  {author} {\bibinfo {author} {\bibfnamefont {S.}~\bibnamefont
  {Scheel}},\ }\href@noop {} {\bibfield  {journal} {\bibinfo  {journal} {arXiv
  e-print}\ }\textbf {\bibinfo {volume} {{\!\!}}} (\bibinfo {year} {2004})},\
  \Eprint {http://arxiv.org/abs/quant-ph/0406127} {arXiv:quant-ph/0406127}
  \BibitemShut {NoStop}%
\bibitem [{\citenamefont {Valiant}(1979)}]{Valiant1979}%
  \BibitemOpen
  \bibfield  {author} {\bibinfo {author} {\bibfnamefont {L.}~\bibnamefont
  {Valiant}},\ }\href@noop {} {\bibfield  {journal} {\bibinfo  {journal}
  {Theoretical Computer Science}\ }\textbf {\bibinfo {volume} {8}},\ \bibinfo
  {pages} {189} (\bibinfo {year} {1979})}\BibitemShut {NoStop}%
\bibitem [{\citenamefont {Bremner}\ \emph {et~al.}(2016)\citenamefont
  {Bremner}, \citenamefont {Montanaro},\ and\ \citenamefont
  {Shepherd}}]{Bremner2016}%
  \BibitemOpen
  \bibfield  {author} {\bibinfo {author} {\bibfnamefont {M.}~\bibnamefont
  {Bremner}}, \bibinfo {author} {\bibfnamefont {A.}~\bibnamefont {Montanaro}},
  \ and\ \bibinfo {author} {\bibfnamefont {D.}~\bibnamefont {Shepherd}},\
  }\href {\doibase 10.1103/PhysRevLett.117.080501} {\bibfield  {journal}
  {\bibinfo  {journal} {Phys. Rev. Lett.}\ }\textbf {\bibinfo {volume} {117}},\
  \bibinfo {pages} {080501} (\bibinfo {year} {2016})}\BibitemShut {NoStop}%
\bibitem [{\citenamefont {{van den Nest}}(2011)}]{Nest2011}%
  \BibitemOpen
  \bibfield  {author} {\bibinfo {author} {\bibfnamefont {M.}~\bibnamefont {{van
  den Nest}}},\ }\href@noop {} {\bibfield  {journal} {\bibinfo  {journal}
  {Quant. Inf. Comp.}\ }\textbf {\bibinfo {volume} {11}},\ \bibinfo {pages}
  {784} (\bibinfo {year} {2011})}\BibitemShut {NoStop}%
\bibitem [{\citenamefont {Pashayan}\ \emph {et~al.}(2017)\citenamefont
  {Pashayan}, \citenamefont {Bartlett},\ and\ \citenamefont
  {Gross}}]{Pashayan2017}%
  \BibitemOpen
  \bibfield  {author} {\bibinfo {author} {\bibfnamefont {H.}~\bibnamefont
  {Pashayan}}, \bibinfo {author} {\bibfnamefont {S.}~\bibnamefont {Bartlett}},
  \ and\ \bibinfo {author} {\bibfnamefont {D.}~\bibnamefont {Gross}},\
  }\href@noop {} {\bibfield  {journal} {\bibinfo  {journal} {arXiv e-print}\
  }\textbf {\bibinfo {volume} {{\!\!}}} (\bibinfo {year} {2017})},\ \Eprint
  {http://arxiv.org/abs/1712.02806} {arXiv:1712.02806 [quant-ph]} \BibitemShut
  {NoStop}%
\bibitem [{\citenamefont {Bartlett}\ and\ \citenamefont
  {Sanders}(2003)}]{Bartlett2003}%
  \BibitemOpen
  \bibfield  {author} {\bibinfo {author} {\bibfnamefont {S.~D.}\ \bibnamefont
  {Bartlett}}\ and\ \bibinfo {author} {\bibfnamefont {B.~C.}\ \bibnamefont
  {Sanders}},\ }\href {\doibase 10.1080/09500340308233564} {\bibfield
  {journal} {\bibinfo  {journal} {Journal of Modern Optics}\ }\textbf {\bibinfo
  {volume} {50}},\ \bibinfo {pages} {2331} (\bibinfo {year}
  {2003})}\BibitemShut {NoStop}%
\bibitem [{\citenamefont {Jerrum}\ \emph {et~al.}(2004)\citenamefont {Jerrum},
  \citenamefont {Sinclair},\ and\ \citenamefont {Vigoda}}]{Jerrum2004}%
  \BibitemOpen
  \bibfield  {author} {\bibinfo {author} {\bibfnamefont {M.}~\bibnamefont
  {Jerrum}}, \bibinfo {author} {\bibfnamefont {A.}~\bibnamefont {Sinclair}}, \
  and\ \bibinfo {author} {\bibfnamefont {E.}~\bibnamefont {Vigoda}},\ }\href
  {\doibase 10.1145/1008731.1008738} {\bibfield  {journal} {\bibinfo  {journal}
  {J. ACM}\ }\textbf {\bibinfo {volume} {51}},\ \bibinfo {pages} {671}
  (\bibinfo {year} {2004})}\BibitemShut {NoStop}%
\bibitem [{\citenamefont {Nielsen}\ and\ \citenamefont
  {Chuang}(2010)}]{Nielsen2010}%
  \BibitemOpen
  \bibfield  {author} {\bibinfo {author} {\bibfnamefont {M.~A.}\ \bibnamefont
  {Nielsen}}\ and\ \bibinfo {author} {\bibfnamefont {I.~L.}\ \bibnamefont
  {Chuang}},\ }\href@noop {} {\emph {\bibinfo {title} {Quantum computation and
  quantum information}}}\ (\bibinfo  {publisher} {Cambridge university press},\
  \bibinfo {year} {2010})\BibitemShut {NoStop}%
\bibitem [{\citenamefont {Chernoff}(1952)}]{Chernoff1952}%
  \BibitemOpen
  \bibfield  {author} {\bibinfo {author} {\bibfnamefont {H.}~\bibnamefont
  {Chernoff}},\ }\href {\doibase 10.1214/aoms/1177729330} {\bibfield  {journal}
  {\bibinfo  {journal} {Ann. Math. Statist.}\ }\textbf {\bibinfo {volume}
  {23}},\ \bibinfo {pages} {493} (\bibinfo {year} {1952})}\BibitemShut
  {NoStop}%
\bibitem [{\citenamefont {Aaronson}\ and\ \citenamefont
  {Kuperberg}()}]{compzoo}%
  \BibitemOpen
  \bibfield  {author} {\bibinfo {author} {\bibfnamefont {S.}~\bibnamefont
  {Aaronson}}\ and\ \bibinfo {author} {\bibfnamefont {G.}~\bibnamefont
  {Kuperberg}},\ }\href {https://complexityzoo.uwaterloo.ca/Complexity_Zoo}
  {\enquote {\bibinfo {title} {{Complexity Zoo}},}\ }\bibinfo {note} {[Online;
  accessed 10-January-2018]}\BibitemShut {NoStop}%
\bibitem [{\citenamefont {Tichy}\ \emph {et~al.}(2014)\citenamefont {Tichy},
  \citenamefont {Mayer}, \citenamefont {Buchleitner},\ and\ \citenamefont
  {M\o{}lmer}}]{Tichy2014}%
  \BibitemOpen
  \bibfield  {author} {\bibinfo {author} {\bibfnamefont {M.~C.}\ \bibnamefont
  {Tichy}}, \bibinfo {author} {\bibfnamefont {K.}~\bibnamefont {Mayer}},
  \bibinfo {author} {\bibfnamefont {A.}~\bibnamefont {Buchleitner}}, \ and\
  \bibinfo {author} {\bibfnamefont {K.}~\bibnamefont {M\o{}lmer}},\ }\href
  {\doibase 10.1103/PhysRevLett.113.020502} {\bibfield  {journal} {\bibinfo
  {journal} {Phys. Rev. Lett.}\ }\textbf {\bibinfo {volume} {113}},\ \bibinfo
  {pages} {020502} (\bibinfo {year} {2014})}\BibitemShut {NoStop}%
\bibitem [{\citenamefont {Aaronson}\ and\ \citenamefont
  {Arkhipov}(2014)}]{Aaronson2014}%
  \BibitemOpen
  \bibfield  {author} {\bibinfo {author} {\bibfnamefont {S.}~\bibnamefont
  {Aaronson}}\ and\ \bibinfo {author} {\bibfnamefont {A.}~\bibnamefont
  {Arkhipov}},\ }\href@noop {} {\bibfield  {journal} {\bibinfo  {journal}
  {Quant. Inf. Comp.}\ }\textbf {\bibinfo {volume} {14}},\ \bibinfo {pages}
  {1383} (\bibinfo {year} {2014})}\BibitemShut {NoStop}%
\bibitem [{\citenamefont {Gogolin}\ \emph {et~al.}(2013)\citenamefont
  {Gogolin}, \citenamefont {Kliesch}, \citenamefont {Aolita},\ and\
  \citenamefont {Eisert}}]{Gogolin2013}%
  \BibitemOpen
  \bibfield  {author} {\bibinfo {author} {\bibfnamefont {C.}~\bibnamefont
  {Gogolin}}, \bibinfo {author} {\bibfnamefont {M.}~\bibnamefont {Kliesch}},
  \bibinfo {author} {\bibfnamefont {L.}~\bibnamefont {Aolita}}, \ and\ \bibinfo
  {author} {\bibfnamefont {J.}~\bibnamefont {Eisert}},\ }\href@noop {}
  {\bibfield  {journal} {\bibinfo  {journal} {arXiv e-print}\ }\textbf
  {\bibinfo {volume} {{\!\!}}} (\bibinfo {year} {2013})},\ \Eprint
  {http://arxiv.org/abs/1306.3995} {arXiv:1306.3995 [quant-ph]} \BibitemShut
  {NoStop}%
\bibitem [{\citenamefont {Doherty}\ \emph {et~al.}(2002)\citenamefont
  {Doherty}, \citenamefont {Parrilo},\ and\ \citenamefont
  {Spedalieri}}]{Doherty2002}%
  \BibitemOpen
  \bibfield  {author} {\bibinfo {author} {\bibfnamefont {A.~C.}\ \bibnamefont
  {Doherty}}, \bibinfo {author} {\bibfnamefont {P.~A.}\ \bibnamefont
  {Parrilo}}, \ and\ \bibinfo {author} {\bibfnamefont {F.~M.}\ \bibnamefont
  {Spedalieri}},\ }\href {\doibase 10.1103/PhysRevLett.88.187904} {\bibfield
  {journal} {\bibinfo  {journal} {Phys. Rev. Lett.}\ }\textbf {\bibinfo
  {volume} {88}},\ \bibinfo {pages} {187904} (\bibinfo {year}
  {2002})}\BibitemShut {NoStop}%
\bibitem [{\citenamefont {Christandl}\ \emph {et~al.}(2007)\citenamefont
  {Christandl}, \citenamefont {K{\"o}nig}, \citenamefont {Mitchison},\ and\
  \citenamefont {Renner}}]{Christandl2007}%
  \BibitemOpen
  \bibfield  {author} {\bibinfo {author} {\bibfnamefont {M.}~\bibnamefont
  {Christandl}}, \bibinfo {author} {\bibfnamefont {R.}~\bibnamefont
  {K{\"o}nig}}, \bibinfo {author} {\bibfnamefont {G.}~\bibnamefont
  {Mitchison}}, \ and\ \bibinfo {author} {\bibfnamefont {R.}~\bibnamefont
  {Renner}},\ }\href {\doibase 10.1007/s00220-007-0189-3} {\bibfield  {journal}
  {\bibinfo  {journal} {Communications in Mathematical Physics}\ }\textbf
  {\bibinfo {volume} {273}},\ \bibinfo {pages} {473} (\bibinfo {year}
  {2007})}\BibitemShut {NoStop}%
\bibitem [{\citenamefont {Harrow}(2013)}]{Harrow2013}%
  \BibitemOpen
  \bibfield  {author} {\bibinfo {author} {\bibfnamefont {A.~W.}\ \bibnamefont
  {Harrow}},\ }\href@noop {} {\bibfield  {journal} {\bibinfo  {journal} {arXiv
  e-print}\ }\textbf {\bibinfo {volume} {{\!\!}}} (\bibinfo {year} {2013})},\
  \Eprint {http://arxiv.org/abs/1308.6595} {arXiv:1308.6595 [quant-ph]}
  \BibitemShut {NoStop}%
\bibitem [{\citenamefont {Oszmaniec}\ and\ \citenamefont
  {Ku{\'s}}(2014)}]{Oszmaniec2014}%
  \BibitemOpen
  \bibfield  {author} {\bibinfo {author} {\bibfnamefont {M.}~\bibnamefont
  {Oszmaniec}}\ and\ \bibinfo {author} {\bibfnamefont {M.}~\bibnamefont
  {Ku{\'s}}},\ }\href
  {http://journals.aps.org/pra/abstract/10.1103/PhysRevA.90.010302} {\bibfield
  {journal} {\bibinfo  {journal} {Phys. Rev. A}\ }\textbf {\bibinfo {volume}
  {90}},\ \bibinfo {pages} {010302} (\bibinfo {year} {2014})}\BibitemShut
  {NoStop}%
\bibitem [{\citenamefont {Tura}\ \emph {et~al.}(2017)\citenamefont {Tura},
  \citenamefont {Aloy}, \citenamefont {Quesada}, \citenamefont {Lewenstein},\
  and\ \citenamefont {Anna}}]{TuraSep2017}%
  \BibitemOpen
  \bibfield  {author} {\bibinfo {author} {\bibfnamefont {J.}~\bibnamefont
  {Tura}}, \bibinfo {author} {\bibfnamefont {A.}~\bibnamefont {Aloy}}, \bibinfo
  {author} {\bibfnamefont {R.}~\bibnamefont {Quesada}}, \bibinfo {author}
  {\bibfnamefont {M.}~\bibnamefont {Lewenstein}}, \ and\ \bibinfo {author}
  {\bibfnamefont {S.}~\bibnamefont {Anna}},\ }\href@noop {} {\bibfield
  {journal} {\bibinfo  {journal} {arXiv e-print}\ }\textbf {\bibinfo {volume}
  {{\!\!}}} (\bibinfo {year} {2017})},\ \Eprint
  {http://arxiv.org/abs/1706.09423} {arXiv:1706.09423 [quant-ph]} \BibitemShut
  {NoStop}%
\bibitem [{\citenamefont {Streltsov}\ \emph {et~al.}(2010)\citenamefont
  {Streltsov}, \citenamefont {Kampermann},\ and\ \citenamefont
  {Bruß}}]{Streltsov2010}%
  \BibitemOpen
  \bibfield  {author} {\bibinfo {author} {\bibfnamefont {A.}~\bibnamefont
  {Streltsov}}, \bibinfo {author} {\bibfnamefont {H.}~\bibnamefont
  {Kampermann}}, \ and\ \bibinfo {author} {\bibfnamefont {D.}~\bibnamefont
  {Bruß}},\ }\href {http://stacks.iop.org/1367-2630/12/i=12/a=123004}
  {\bibfield  {journal} {\bibinfo  {journal} {New Journal of Physics}\ }\textbf
  {\bibinfo {volume} {12}},\ \bibinfo {pages} {123004} (\bibinfo {year}
  {2010})}\BibitemShut {NoStop}%
\bibitem [{\citenamefont {Bouland}\ and\ \citenamefont
  {Aaronson}(2014)}]{Bouland2014}%
  \BibitemOpen
  \bibfield  {author} {\bibinfo {author} {\bibfnamefont {A.}~\bibnamefont
  {Bouland}}\ and\ \bibinfo {author} {\bibfnamefont {S.}~\bibnamefont
  {Aaronson}},\ }\href {\doibase 10.1103/PhysRevA.89.062316} {\bibfield
  {journal} {\bibinfo  {journal} {Phys. Rev. A}\ }\textbf {\bibinfo {volume}
  {89}},\ \bibinfo {pages} {062316} (\bibinfo {year} {2014})}\BibitemShut
  {NoStop}%
\bibitem [{\citenamefont {Sawicki}(2016)}]{Sawicki2016}%
  \BibitemOpen
  \bibfield  {author} {\bibinfo {author} {\bibfnamefont {A.}~\bibnamefont
  {Sawicki}},\ }\href@noop {} {\bibfield  {journal} {\bibinfo  {journal}
  {Quantum Inf. Comput.}\ }\textbf {\bibinfo {volume} {16}},\ \bibinfo {pages}
  {0291} (\bibinfo {year} {2016})}\BibitemShut {NoStop}%
\bibitem [{\citenamefont {Clevenson}\ and\ \citenamefont
  {Watkins}(1991)}]{SchurConcave1991}%
  \BibitemOpen
  \bibfield  {author} {\bibinfo {author} {\bibfnamefont {M.~L.}\ \bibnamefont
  {Clevenson}}\ and\ \bibinfo {author} {\bibfnamefont {W.}~\bibnamefont
  {Watkins}},\ }\href {\doibase 10.2307/2691301} {\bibfield  {journal}
  {\bibinfo  {journal} {Mathematics Magazine}\ }\textbf {\bibinfo {volume}
  {64}},\ \bibinfo {pages} {183} (\bibinfo {year} {1991})}\BibitemShut
  {NoStop}%
\bibitem [{\citenamefont {{Garc{\'{\i}}a-Patr{\'o}n}}\ \emph
  {et~al.}(2017)\citenamefont {{Garc{\'{\i}}a-Patr{\'o}n}}, \citenamefont
  {{Renema}},\ and\ \citenamefont {{Shchesnovich}}}]{GarciaPatron2017}%
  \BibitemOpen
  \bibfield  {author} {\bibinfo {author} {\bibfnamefont {R.}~\bibnamefont
  {{Garc{\'{\i}}a-Patr{\'o}n}}}, \bibinfo {author} {\bibfnamefont {J.~J.}\
  \bibnamefont {{Renema}}}, \ and\ \bibinfo {author} {\bibfnamefont
  {V.}~\bibnamefont {{Shchesnovich}}},\ }\href@noop {} {\bibfield  {journal}
  {\bibinfo  {journal} {arXiv e-print}\ }\textbf {\bibinfo {volume} {{\!\!}}}
  (\bibinfo {year} {2017})},\ \Eprint {http://arxiv.org/abs/1712.10037}
  {arXiv:1712.10037 [quant-ph]} \BibitemShut {NoStop}%
\end{thebibliography}%
	
\appendix
\section{Proofs of technical results}\label{app:technicRES}
In this part of the Appendix we collect proofs of some of the auxiliary results that were needed to obtain main findings of the paper.

\begin{customlemma}{1}\label{lem:auxBOUNDSpr}
Let $ \dsepB{\tran}$ be the trace distance between states $\sigma_\tran$ and $\targ{\tran}$ defined in \cref{eq:FOCKlost} and \cref{eq:diffNsep} respectively. We have the following inequalities 
\begin{align}
\dsepB{\tran}  \leq & \frac{\tran^2 n}{2} + \frac{\tran(1-\tran)}{2}\ , \label{eq:UPPERbound2} \\  
\dsepB{\tran}   \geq & \dsep{\lceil \tran n (1-\delta)n \rceil} \L 1 - \exp\L-\delta^2 n\tran/2\R  \R\ .   \label{eq:LOWERbound2} 
\end{align}
where $\delta\in[0,1]$ and $\lceil x \rceil$ denotes the smallest integer greater than $x$.
\end{customlemma}

\begin{proof}
We first observe that $\dsepB{\tran}$ given in \cref{eq:distETAsep} can be interpreted as the average of $\dsep{l}$, where $l$ is distributed according to the binomial distribution $\Pr(l)$,
\begin{equation}\label{eq:averagePRES}
\dsepB{\tran}= \sum_{l=0}^n \Pr(l) \dsep{l}\ .
\end{equation}  
Now, in order to prove \cref{eq:UPPERbound2} we use the upper bound $\dsep{l}\leq\frac{l^2}{2n}$, which can be immediately proven using mathematical induction in $l$. Applying this bound to \cref{eq:averagePRES} yields 
\begin{equation}
\dsepB{\tran}\leq \sum_{l=0}^n \Pr(l) \frac{l^2}{2n} = \frac{\tran^2 n}{2} + \frac{\tran(1-\tran)}{2}\ ,
\end{equation}
where in the last equality we have used the standard properties of the binomial distribution.

The proof of \cref{eq:LOWERbound2} also uses \cref{eq:averagePRES}. Since $\dsep{l}\geq$ we have 
\begin{equation}
\dsepB{\tran}\geq \sum_{l=l_\ast}^{n} \Pr(l) \dsep{l}\ ,
\end{equation} 
for $l_\ast = \lceil \tran n (1-\delta) \rceil$. Using the fact that $\dsep{l}$ is and increasing function of $l$ (for fixed $n$), we obtain
\begin{equation}
\sum_{l=l_\ast}^{n} \Pr(l) \dsep{l} \geq \dsep{ \lceil \tran n (1-\delta) \rceil} \Pr\L l \geq  \tran n (1-\delta)   \R\ .
\end{equation}
Finally,  \cref{eq:LOWERbound2} follows from the above inequality by applying the Chernoff bound given in \cref{eq:Chernoff}. 
\end{proof}

\begin{customlemma}{3}[Structure of the twirled product states]\label{lem:twirledPRODUCTpr}
Let $\ket{\phi}=\sum_{i=1}^n \alpha_i \ket{i}$ be a pure state on $\C^n$. Let $K_c$, $K_d$ be subgroups of $\LO{n}$ defined in Eq.\eqref{eq:stabSUBROUPS}. Let $\Lambda_K$  be a twirling map defined by Eq.\eqref{eq:twirledState}. Then, we have the following result
\begin{equation}\label{eq:Avgtotalb}
\Lambda_{K_d} \circ \Lambda_{K_c} \left(\ketbra{\phi}{\phi}^{\otimes l}\right) = \sum_{\t: |\t|=l} m_{\t}(\alpha) \binom{l}{\t} \rho_{\t}\ ,
\end{equation}
where we have used the notation introduced below Corollary \ref{cor:doubleAV}.  
\end{customlemma}

\begin{proof}

In many places of this proof we will be using the notation introduced below Corollary \ref{cor:doubleAV} in the main text. We will deal first with the averaging over the group $K_c$. Consider first the pure state $\ketbra{\phi}{\phi}^{\otimes l} \in \D\L \symm{l}{n} \R$, for $\ket{\phi} = \sum_{i=1}^n \alpha_i \ket{i}$ and $\sum_{i=1}^n \lvert \alpha_i \lvert^2$. To simplify the calculations, let us write $\ket{\phi}^{\otimes N} = \sum_{\bi}\alpha_{\bi}\ket{\bi}$, where we defined
\begin{equation}\label{eq:EXTRAnotation}
\begin{cases}
\; \bi \eqdef \left( i_1, i_2 \ldots, i_l \right), \textrm{with } i_k \in \{1, \ldots, n\}\ , \\ 
\; \ket{\bi} \eqdef \ket{i_1} \ldots \ket{i_l}\ ,\\
\; \alpha_{\bi} \eqdef \prod_{j=1}^l \alpha_{i_j}\ .
\end{cases}
\end{equation}
 Now we can write $\ket{\phi}^{\otimes l} = \proj{l}{d} \ket{\phi}^{\otimes l} = \sum_{\bi} \alpha_{\bi} \proj{l}{n} \ket{\bi}$. If we now define $\n(\bi) \eqdef \left(n_1(\bi), \ldots n_n(\bi)\right)$, where $n_k(\bi)$ is the number of elements of $\bi$ equal to $k$, we have
\begin{equation}
\proj{l}{n} \ket{\bi} = N(\n(\bi)) \ket{\n(\bi)}\ ,
\end{equation}
where $\ket{\n(\bi)}$ is the $l$-particle Fock state associated to the occupation number string $\n(\bi)$ and the coefficient $N(\n(\bi))$ is given by
\begin{equation}\label{eq:coefFORM}
N(\n(\bi)) \eqdef \left[\binom{l}{\n(\bi)}\right]^{-1/2} = \sqrt{\frac{\prod_{i=1}^n n_i(\bi)}{l!}} \ ,
\end{equation}
which manifestly depends only on $\n(\bi)$ and not on $\bi$ itself. All these manipulations allow us to write
\begin{equation}
\ketbra{\phi}{\phi}^{\otimes l} = \sum_{\substack {\n, \n': \\ |\n|=|\n'|=l}  } \sum_{\substack{\bi, \bi': \\ \n(\bi)=\n \\ \n(\bi')=\n'}} \alpha_{\bi} \overline{\alpha}_{\bi'} N(\n) N(\n') \ketbra{\n}{\n'}.
\end{equation}
Elements of the form $V^{\otimes l }$, for $V\in K_c$  act on $\ket{\n}$ just as a phase and is easy to see that averaging over $K_c$ just decoheres any state to the Fock basis. In particular we have
\begin{equation}\label{eq:firstAV}
\Lambda_{K_c} \L \ketbra{\phi}{\phi}^{\otimes l} \R = \sum_{\n: |\n|=l} \sum_{\substack{\bi, \bi': \\ \n(\bi)=\n(\bi')=\n}} \alpha_{\bi} \overline{\alpha}_{\bi'} N(\n)^2 \ketbra{\n}{\n}\  .
\end{equation}
We note now that if $\n(\bi)=\n(\bi')$, then $\alpha_{\bi}=\alpha_{\bi'}$  and therefore $\alpha_{\bi} \overline{\alpha}_{\bi'}=|\alpha_{\bi}|^2=\prod_{=1}^d |\alpha_i|^{2 n_i(\bi)}$, which also depends only on $\n(\bi)$. Using the notation $|\alpha|^{2\n(\bi)}=|\alpha_{\bi}|^2$ to reflect that, we can write
\begin{equation}\label{eq:avg0}
\Lambda_{K_c} \L \ketbra{\phi}{\phi}^{\otimes l} \R = \sum_{\n:|\n|=l}  |\alpha|^{2\n} \binom{l}{\n}^2 N(\n)^2 \ketbra{\n}{\n}=  \sum_{\n:|\n|=l} \binom{l}{\n} |\alpha|^{2\n}  \ketbra{\n}{\n} \  .
\end{equation}
where the term $\binom{l}{\n}^2$  came from the sums over $\bi$ and $\bi'$, which just count the number of possibilities for $\bi$ such that $\n(\bi')=\n$ since nothing in the summand in Eq.\eqref{eq:firstAV} depends explicitly on $\bi$.

Let us now consider the averaging of the state in Eq.\eqref{eq:avg0} over the group $K_d$. However, this averaging is straightforward due to the following equality
\begin{equation}\label{eq:avgSd}
\Lambda_{K_d} \left(\ketbra{\n}{\n}\right) = \rho_{\t(\n)}\ ,
\end{equation}
where the state $\rho_{\t(\n)}$ was defined in Eq.\eqref{eq:permutAVfock}  and  $\tau(\n)$ labels the equivalence class of $\n$ up to permutations of the $n$ modes. So combining \eqref{eq:avg0} and \eqref{eq:avgSd} we get our final result
\begin{equation}
\Lambda_{K_d} \circ \Lambda_{K_c} \L \ketbra{\phi}{\phi}^{\otimes l} \R = \sum_{\t:|\t|=l} m_{\t} (\alpha) \binom{l}{\t}  \rho_{\t}\ ,
\end{equation}
where we have used the definition of the function $m_{\t} (\alpha)$ given in Eq.\eqref{eq:symPOLY}. 
\end{proof}
 
 \newpage

\section{Detailed example of the procedure of Theorem 2}\label{app:example}

\begin{figure}[b]
    \centering
    \includegraphics[width=0.6\textwidth]{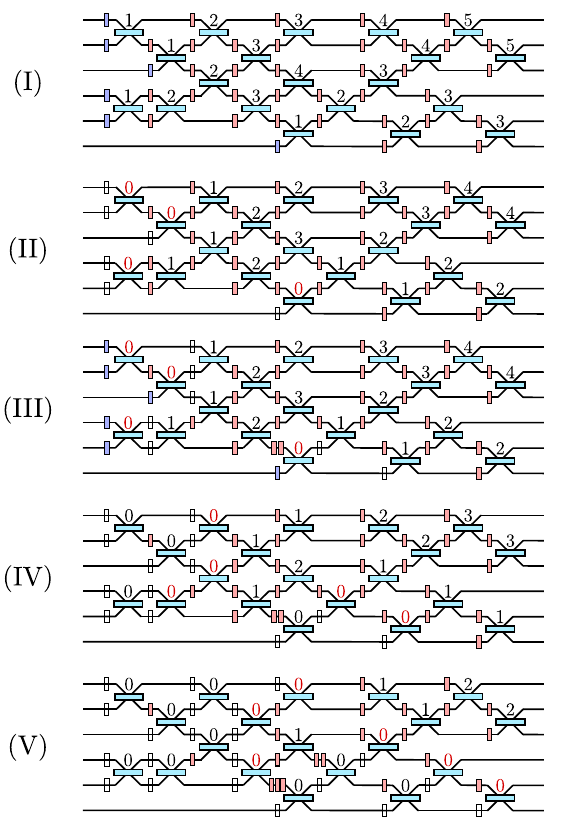}
    \caption{An example of our procedure to pull out losses from a linear-optical network. Beamsplitters labeled in red zeroes are the ones that, in \step $j$, are in $G_{j-1}$, as  defined in the main text. Empty boxes correspond to loss elements that have been commuted elsewhere. The blue loss elements in panel (I) are the same as in \cref{fig:numberednetwork}, and are the first to be pulled out. In panel (III) they exemplify the argument, given in \cref{sec:proofEXTR}, that in every step, after losses are commuted through the red-labeled beamsplitters, they recover the configuration of the previous step, and by the inductive hypothesis they can always be pulled out.}
    \label{fig:APexample}
\end{figure}

In \cref{fig:APexample} we represent the intermediate steps lead from \cref{fig:numberednetwork} to \cref{fig:Endexample}. Panel (I) shows the initial configuration, as in \cref{fig:numberednetwork}. Panel (II) shows the configuration and the relabeling after \step 1 is concluded. The set $G_1$ is identified by the red-labeled beamsplitter. In panel (III), losses have been commuted through the marked beamsplitters, and now occupy the same configuration as before \step 1, which means that by assumption they can now be pulled out as a uniform layer (this is trivial to see in this particular example, because the losses that will be pulled out are already at the input, but the argument holds in general). Panels (IV) and (V) show the configuration and labeling after \step 2 and \step 3, respectively. After \step 3, three layers of losses have been pulled out and the shortest \IO-path has length 0, so the procedure terminates.

\end{document}